\documentclass[a4paper,12pt]{article}

\usepackage{color}
\usepackage{graphicx}
\usepackage{amsmath}
\usepackage{amsfonts}
\usepackage{amsthm}
\usepackage{amscd}
\usepackage{epsfig}
\usepackage{amssymb}
\usepackage{tabularx}
\usepackage{calligra}
\usepackage{enumerate}
\usepackage{hyperref}
\usepackage{stackrel}
\usepackage[T1]{fontenc}
\usepackage{amsthm}
\usepackage{verbatim} 
\usepackage{tikz}
\usetikzlibrary{positioning}

\numberwithin{equation}{section}

\newtheorem{thm}{Theorem}[section]
\newenvironment{thmbis}[1]
  {\renewcommand{\thethm}{\ref*{#1}$'$}%
   \addtocounter{thm}{-1}%
   \begin{thm}}
  {\end{thm}}
\newtheorem{prop}[thm]{Proposition}
\newtheorem{lem}{Lemma}[section]
\newtheorem{cor}[thm]{Corollary}
\newenvironment{corbis}[1]
  {\renewcommand{\thethm}{\ref*{#1}$'$}%
   \addtocounter{thm}{-1}%
   \begin{cor}}
  {\end{cor}}

\newcommand*\boxt[1]{
\begin{tikzpicture}
\draw node[inner sep=1.8pt,draw]{\tiny{\sc T}};
\end{tikzpicture}}

\renewcommand{\a}{\alpha}
\renewcommand{\b}{\beta}
\newcommand{\g}{\gamma}
\newcommand{\G}{\Gamma}
\newcommand{\e}{\epsilon}
\renewcommand{\l}{\lambda}
\newcommand{\m}{\mu}
\newcommand{\n}{\nu}

\renewcommand{\d}{\delta}

\newcommand{\w}{\omega}

\newcommand{\p}{\partial}

\newcommand{\s}{\mathfrak{s}}
\newcommand{\im}{\mathfrak{Im}}

\newcommand{\teuk}{\ensuremath{\hspace{.05cm}\mathaccent\Box{\text{\tiny \textsc{T}}}\hspace{.05cm}}}     
\def\c{\nabla}
\def\wt{\widetilde}

\def\wh{\widehat}

\begin{document}

\title{Symmetry operators and decoupled equations for linear fields on black hole spacetimes}

\author{Bernardo Araneda\footnote{E-mail: \texttt{baraneda@famaf.unc.edu.ar}} \\ 
 \\
Facultad de Matem\'atica, Astronom\'{\i}a y F\'{\i}sica\\
Universidad Nacional de C\'ordoba\\ 
Instituto de F\'{\i}sica Enrique Gaviola, CONICET\\
Ciudad Universitaria, (5000) C\'ordoba, Argentina 
}

\date{March 27, 2017}

\maketitle

\begin{abstract}
In  the class of vacuum Petrov type D spacetimes with cosmological constant, which includes the Kerr-(A)dS black hole as 
a particular case, we find a set of  four-dimensional operators that, when composed  {\em off shell} with the Dirac, Maxwell 
and linearized gravity equations,  give  a system of  equations for spin weighted scalars associated to  the linear fields,  
that decouple {\em on shell}. Using these operator relations we give compact reconstruction formulae for solutions of the original 
spinor and tensor field equations in terms of solutions of the decoupled scalar equations. We also analyze the role of Killing 
spinors and Killing-Yano tensors in the spin weight zero equations and, in the case of spherical symmetry, we compare our 
four-dimensional formulation with the standard $2+2$ decomposition and particularize to the Schwarzschild-(A)dS black hole. 
Our results uncover  a pattern that  generalizes a number of  previous results on Teukolsky-like equations and Debye potentials 
for higher spin fields. 
\end{abstract}

\tableofcontents

\section{Introduction}

One of the most important open issues in General Relativity (GR) is the black hole stability problem, which consists in proving
the dynamical, non-linear stability of the Kerr metric within the set of solutions of  the Einstein equations. Due to the high complexity,
different levels  of simplifications are considered when approaching this problem; in the first place, one considers {\em linear} systems,  
which, from the physical point of view, represent the propagation of the fundamental classical fields on these spacetimes ignoring back reaction. 
The dynamical evolution of these fields is described by solutions of partial differential equations of 
tensorial or spinorial nature on the Lorentzian manifold that represents the spacetime, 
the structure of which depends on the kind of field we are  dealing with.
For example, the linearized Einstein equations are a set of ten (in four-dimensional GR) linear, second order, partial 
differential equations governing the evolution of the linearized metric. \\
The problem of analyzing solutions of these equations would be simplified if we were able to 
find  equivalent field equations for {\em scalar fields} encoding the dynamical degrees of freedom of the perturbative field, as scalar fields are simpler 
and, unlike spinor and tensor fields on a 
Lorentzian manifold,  carry an obvious notion of size.  
This turns out to be the case for the gravitational perturbations of the Schwarzschild black hole, as recently showed in 
\cite{Dotti}.
The proof of nonmodal linear stability of the Schwarzschild black hole in \cite{Dotti}  makes use of the fact that the 
linearization of a scalar curvature invariant $\Phi[h]$, $h_{\alpha \beta}$ the metric perturbation, 
  satisfies a wave-like equation which, according to the conventions 
of the present paper reads\footnote{we take the metric to have signature $(+---)$, whereby $\Box$ corresponds to $-\Box$ 
in \cite{Dotti}}
\begin{equation}\label{we}
 (\Box-\tfrac{8M}{r^3})\Phi=0,
\end{equation}
where $\Box=g^{\a\b}\c_{\a}\c_{\b}$ is the standard D'Alembertian, with $\c_{\a}$ the Levi-Civita connection.
Furthermore, for the odd sector,  a solution of  the linearized Einstein equations can be covariantly reconstructed from a 
scalar field satisfying (\ref{we})  by means of 
\begin{equation}\label{metricschw}
 h^{-}_{\a\b}=\tfrac{r^2}{3M} {}^{*}C_{\a}{}^{\g\d}{}_{\b}\c_{\g}\c_{\d}(r^3\Phi),
\end{equation}
where ${}^{*}C_{\a\b\g\d}$ is the dual Weyl tensor of the background Schwarzschild solution.
This suggests that there  exists a four-dimensional map transforming {\em off-shell} the linearized Einstein tensor, 
regarded as a linear differential operator on $h_{\alpha \beta}$, into the composition of the scalar 
wave operator  acting on $\Phi$ in  (\ref{we}) and 
the linear differential  operator $\Phi[h]$.
By {\em off-shell} we mean that this is an {\em operator} equality for operators acting on 
 the space of symmetric $(0,2)$ tensors (where  the perturbed metric  tensor  lives) and, as such,  
 it holds  regardless of any field equations satisfied by $h_{\alpha \beta}$.
If this is so, a natural question to ask, besides what the explicit form of this map is,  is
whether such an operator equality exists for more general spacetimes, in particular for the Kerr solution.
In this work we address this question for the class of vacuum Petrov type D spacetimes with cosmological constant, 
which includes the Kerr-(A)dS black hole as a particular case.
We  proof the existence of   maps transforming spinor/tensor field operators into scalar operators. These maps  have a universal form that depends on 
 the spin $\s$ of the field (for $\s=\tfrac{1}{2},1,2$) and the spin weight $s$ of the scalar $\Phi$, and 
(\ref{we}) corresponds to the particular case $(\s=2,s=0)$ on a Schwarzschild background. 
We find that the mechanism explaining why (\ref{we})-(\ref{metricschw}) solves the linearized Einstein equations 
is the transposition of linear differential operators introduced by  Wald in \cite{Wald}.
We also investigate the role of Killing-Yano tensors in the description of 
perturbative fields, since, although not stated in \cite{Dotti}, it turns out that the $\Phi$ solving (\ref{we}) 
(in the particular case of the Schwarzschild solution) can be written as 
\begin{equation}\label{Phi}
\Phi=Y^{\a\b}{}^{*}Y_{\g\d}\dot{C}_{\a\b}{}^{\g\d},
\end{equation}
where $Y_{\a\b}$ is a Killing Yano tensor, ${}^{*}Y_{\a\b}$ its dual, and $\dot{C}_{\a\b}{}^{\g\d}$ is the linearized Weyl tensor.\\

As is well-known, perturbations of rotating black holes are traditionally studied by the Teukolsky equations \cite{Teukolsky}, 
which are a set of decoupled scalar equations for the {\em extreme} spin weight components of perturbative fields of 
spin $\frac{1}{2}$, 1 and 2. As showed in \cite{Bini}, these equations can be put in a wave-like form by adding 
to the Levi-Civita derivative a  connection 1-form $\Gamma_{\a}$ (see (\ref{Gamma}) below for an explicit expression), 
that gives a  weigthed covariant derivative $\c_{\a}+p\Gamma_{\a}$, $p\in\mathbb{Z}$, and 
the weighted wave operator \cite{Andersson1}
\begin{equation}
 \teuk_{p}:=g^{\a\b}(\c_{\a}+p\Gamma_{\a})(\c_{\b}+p\Gamma_{\b}).
\end{equation}
The advantage of using this modified wave operator is that the Teukolsky equations adopt a very simple 
and elegant form in terms of it \cite{Bini, Andersson1}: 
\begin{equation}
 (\teuk_{2s}-4s^2\Psi_2)\Phi^{(s)}=0,
\end{equation} 
where the field $\Phi^{(s)}$ has spin weight $s$ and it is assumed a vacuum type D background spacetime
(the adjointness property of the Teukolsky system is also easily seen in terms of the modified wave operator, see 
subsection \ref{adjop-sec} below and references therein).
The extreme spin weight cases are those treated in the original work of Teukolsky. 
However, we are particularly interested in {\em spin weight zero} fields, both because they are truly 
(tetrad independent) scalar fields and also  because the scalar field  in 
(\ref{we}) is of this type. For gravitational perturbations,
decoupled equations for all the perturbed Weyl scalars have been derived in \cite{Andersson1} (for spin weight $s=\pm1$ the equations 
in \cite{Andersson1} are actually not decoupled, in the sense that they involve perturbed quantities other than the Weyl scalars).
In any case, the equations in \cite{Andersson1} are valid {\em on-shell}, i.e. the linearized vacuum Einstein equations are assumed to hold. 
Since we are interested in finding operator relations we cannot make this assumption. Working {\em off shell} 
is what ultimately allows us to find patterns  relating the equations for perturbed Weyl scalars and the linearized Einstein tensor, and these relations allow us 
to construct 
solutions of the linearized Einstein equations 
from  solutions of the decoupled scalar equations. 
In the following section we state our main results. They all have the form of operator relations for operators acting on Dirac, Maxwell 
and perturbed metric fields. On shell, they give a system of decoupled scalar wave-like equations implied by the (Dirac, Maxwell and linear gravity) 
field equations. Their off-shell validity is what allow us to construct solutions for the Dirac, Maxwell and linear gravity equations 
from solutions of simple scalar wave-like equations.

\subsection{Main results}

We recall that, in the Petrov classification,  type D spaces, which include the Kerr family, have  two 
principal null directions (PNDs) $o^A$, $\iota^A$ in terms of which the only non-trivial Weyl scalar of the curvature is $\Psi_2$.
In the following, the spinors $o^A$, $\iota^A$ (and the associated null vectors) will always refer to these PNDs.
In particular, we introduce the anti-self-dual 2-forms 
\begin{equation}\label{Mforms}
 \stackbin[]{0}{M}_{\a\b}:=2l_{[\a}m_{\b]}, \;\;\;
 \stackbin[]{1}{M}_{\a\b}:=2l_{[\a}n_{\b]}+2\bar{m}_{[\a}m_{\b]}, \;\;\; \stackbin[]{2}{M}_{\a\b}:=2\bar{m}_{[\a}n_{\b]},
\end{equation}
associated to the principal null tetrad $\{l^{\a},n^{\a},m^{\a},\bar{m}^{\a}\}$, 
and the following tensors, which are  anti-self-dual in each pair of indices and have the symmetries  of the Weyl tensor:
\begin{eqnarray}
 \stackbin[]{0}{W}_{\a\b\g\d}&:=&\stackbin[]{0}{M}_{\a\b}\stackbin[]{0}{M}_{\g\d}, \label{W0}\\
 \stackbin[]{2}{W}_{\a\b\g\d}&:=&\stackbin[]{0}{M}_{\a\b}\stackbin[]{2}{M}_{\g\d}
 +\stackbin[]{2}{M}_{\a\b}\stackbin[]{0}{M}_{\g\d}+\stackbin[]{1}{M}_{\a\b}\stackbin[]{1}{M}_{\g\d}, \label{W2}\\
 \stackbin[]{4}{W}_{\a\b\g\d}&:=&\stackbin[]{2}{M}_{\a\b}\stackbin[]{2}{M}_{\g\d}. \label{W4}
\end{eqnarray}
In this paper we prove  that there are  four-dimensional maps that transform off-shell (in a sense to be made precise below) 
the field operators of higher spin fields into scalar operators. 
Although these operators  have a generic form that depends on the spin $\s$ of the field and the spin weight $s$ of the 
related scalar (as we show in section \ref{adjop-sec}), 
for clarity purposes we give now, in separate form, the explicit operators for spins $\s=\frac{1}{2}$, 1 and 2, and for zero 
and extreme spin weight, $s=0,\pm\s$.\\

\noindent
Consider first spin $\s=\frac{1}{2}$; this case describes massless Dirac fields. We will use the 2-spinor formalism, 
in which the massless Dirac equation is $\c^{AA'}\chi_A=0$. \\
The proof of the theorem below can be found  in section \ref{diracsec}. \\

\begin{thm}[spin $\s=\frac{1}{2}$]\label{spin1/2}
 Consider an arbitrary spinor field $\chi_A$ on a vacuum Petrov type D spacetime with cosmological constant $\l$. 
 Then we have the following equalities:
 \begin{eqnarray}
 \Psi^{1/3}_2o^{B}\c^{B'}_{B}[\Psi^{-1/3}_2\c^{A}_{B'}\chi_{A}]
 &=&-\tfrac{1}{2}(\teuk_{+1}-\Psi_2 +\tfrac{2}{3}\l)[o^A\chi_A] \label{dirac1} \\
 \iota^{B}\c^{B'}_{B}[\Psi^{-1/3}_2\c^{A}_{B'}\chi_{A}]
 &=&-\tfrac{1}{2}(\teuk_{-1}-\Psi_2 +\tfrac{2}{3}\l)[\Psi^{-1/3}_2\iota^A\chi_A]. \label{dirac2}
\end{eqnarray}
\end{thm}
\noindent

Note that $\chi_{A}$ in equations (\ref{dirac1}) and (\ref{dirac2}) is an {\em arbitrary} $\s=\frac{1}{2}$ field, that is,  not satisfying any 
field equation, these are examples of   what we mean by {\em off shell} equations. {\em If} the field $\chi_A$ satisfies the Dirac equation, the left hand sides 
of (\ref{dirac1}) and (\ref{dirac2}) vanish and we get a system of two decoupled linear homogeneous (Teukolsky) 
equations for the scalar fields $o^A\chi_A$ and $\iota^A\chi_A$. Knowledge of the off shell relations above is crucial 
for constructing solutions of the original (Dirac, in this case) field equations from scalar (Debye) potentials. \\

The spin $\s=1$ case corresponds to Maxwell fields, which are solutions to $\c^{AA'}\phi_{AB}=0$.
The following theorem, proved in section \ref{maxwellsec}, 
shows that a similar structure to that of spin $\s=\frac{1}{2}$ occurs for this case:
\begin{thm}[spin $\s=1$, spinor version]\label{spin1}
 Consider an arbitrary symmetric spinor field $\phi_{AB}$ on a vacuum Petrov type D spacetime with cosmological constant $\l$. 
 Then we have the following equalities:
 \begin{align}
 \Psi^{2/3}_2o^{B}o^{C}\c^{B'}_{C}[\Psi^{-2/3}_2\c^{A}_{B'}\phi_{AB}]
 =-\tfrac{1}{2}(\teuk_{+2}-4\Psi_2 +\tfrac{2}{3}\l)[o^Ao^B\phi_{AB}] \\
 \Psi^{1/3}_2o^{(B}\iota^{C)}\c^{B'}_{C}[\Psi^{-2/3}_2\c^{A}_{B'}\phi_{AB}]
 =-\tfrac{1}{2}(\Box+2\Psi_2 +\tfrac{2}{3}\l)[\Psi^{-1/3}_2o^A\iota^B\phi_{AB}] \\
 \iota^{B}\iota^{C}\c^{B'}_{C}[\Psi^{-2/3}_2\c^{A}_{B'}\phi_{AB}]
 =-\tfrac{1}{2}(\teuk_{-2}-4\Psi_2 +\tfrac{2}{3}\l)[\Psi^{-2/3}_2\iota^A\iota^B\phi_{AB}]
 \end{align}
\end{thm}
\noindent
The tensor version of this theorem is achieved by introducing an anti-self-dual 2-form $\wt{F}_{\a\b}=F_{\a\b}+i {}^{*}F_{\a\b}$, 
and by using the tensors (\ref{Mforms}):
\begin{thmbis}{spin1}[spin $\s=1$, tensor version]
 Consider an arbitrary anti-self-dual 2-form $\wt{F}_{\a\b}$ on a vacuum Petrov type D spacetime with cosmological constant $\l$. 
 Then we have the following equalities:
 \begin{eqnarray}
  \Psi^{2/3}_2\stackbin[]{0}{M}{}^{\b\g}\c_{\g}[\Psi^{-2/3}_2\c^{\a}\wt{F}_{\a\b}]&=&
 -(\teuk_{+2}-4\Psi_2 +\tfrac{2}{3}\l)[\stackbin[]{0}{M}{}^{\a\b}\wt{F}_{\a\b}] \label{maxwell1} \\
 \Psi^{1/3}_2\stackbin[]{1}{M}{}^{\b\g}\c_{\g}[\Psi^{-2/3}_2\c^{\a}\wt{F}_{\a\b}]&=&
 -(\Box+2\Psi_2 +\tfrac{2}{3}\l)[\Psi^{-1/3}_2\stackbin[]{1}{M}{}^{\a\b}\wt{F}_{\a\b}] \label{maxwell2} \\
 \stackbin[]{2}{M}{}^{\b\g}\c_{\g}[\Psi^{-2/3}_2\c^{\a}\wt{F}_{\a\b}]&=&
\nonumber -(\teuk_{-2}-4\Psi_2 +\tfrac{2}{3}\l)[\Psi^{-2/3}_2\stackbin[]{2}{M}{}^{\a\b}\wt{F}_{\a\b}]. \\
 \label{maxwell3}
\end{eqnarray}
\end{thmbis}
As in the Dirac field case, on shell the left hand sides of the above equations vanish and give decoupled scalar field  equations 
for the Maxwell scalars on the right hand side, 
generalizing Teukolsly equations  to non extreme spin weights. Once again, the fact that equations (\ref{maxwell1})-(\ref{maxwell3})
hold for {\em any}  anti-self-dual 2-form $\wt{F}_{\a\b}$ is what interest us most.\\

Finally, spin $\s=2$ corresponds to gravitational perturbations. 
We assume there is a  mono-parametric family of metrics $g_{\a \b}(\e)$, where the unperturbed metric $g_{\a \b}(0)=g_{\a \b}$ 
solves Einstein equations.
In what follows, we use alternatively $\frac{d}{d\e}|_{\e=0}$ and a dot over a quantity to denote linearization.
Assuming linearized Einstein vacuum equations (with cosmological constant) are also satisfied (that is, {\em on shell}), 
the linearized Bianchi identities are formally $\frac{d}{d\e}|_{\e=0}(\c^{AA'}\psi_{ABCD})=0$ (see e.g. \cite[Eq. (2.8)]{curtis}). 
The operators to be applied {\em off shell} to these identities follow a similar pattern to those of spin $\s=\frac{1}{2}$ 
and $\s=1$, as the following theorem shows:
\begin{thm}[spin $\s=2$, spinor version]\label{spin2SV}
 Let $(\mathcal{M}_{\e},g_{\a\b}(\e))$ be a monoparametric family of pseudo-Riemannian manifolds, analytic around $\e=0$, 
 such that $g_{\a\b}(0)$ satisfies the vacuum Einstein equations (with cosmological constant $\lambda$)  and is of Petrov type D.
 Let $\psi_{ABCD}$ be the Weyl curvature spinor of the metric $g_{\a\b}(\e)$.
 Then we have the following equalities:
 \begin{align}
 \tfrac{d}{d\e}|_{\e=0}[\Psi^{4/3}_2o^{B}o^{C}o^{D}o^{E}\c^{B'}_{E}(\Psi^{-4/3}_2\c^{A}_{B'}\psi_{ABCD})]
 =-\tfrac{1}{2}(\teuk_{+4}-16\Psi_2 +\tfrac{2}{3}\l)\dot{\Psi}_0 \\
 \tfrac{d}{d\e}|_{\e=0}[\Psi^{2/3}_2o^{(B}o^{C}\iota^{D}\iota^{E)}\c^{B'}_{E}(\Psi^{-4/3}_2\c^{A}_{B'}\psi_{ABCD})]
 =-\tfrac{3}{2}\tfrac{d}{d\e}|_{\e=0}[(\Box+2\Psi_2 +\tfrac{R}{6})\Psi^{1/3}_2] \\
 \tfrac{d}{d\e}|_{\e=0}[\iota^{B}\iota^{C}\iota^{D}\iota^{E}\c^{B'}_{E}(\Psi^{-4/3}_2\c^{A}_{B'}\psi_{ABCD})]
 =-\tfrac{1}{2}(\teuk_{-4}-16\Psi_2 +\tfrac{2}{3}\l)[\Psi^{-4/3}_2\dot{\Psi}_4]
 \end{align}
 where $\dot{\Psi}_i=\frac{d}{d\e}|_{\e=0}\Psi_i(\e)$, $i=0,4$.
\end{thm}
\noindent
This theorem shows how to map off shell the linearized Bianchi identities to decoupled equations for perturbed Weyl scalars.
However, gravitational perturbations are traditionally described in terms of the perturbed metric 
$h_{\a \b} = \dot g_{\a \b} =\frac{d}{d\e}|_{\e=0} g_{\a \b}(\e)$, which is a solution 
to the linearized Einstein equations $\dot{G}_{\a\b}[h]+\l h_{\a\b}=0$, where 
$\dot{G}_{\a\b}[h] =  \frac{d}{d\e}|_{\e=0} G_{\a \b}(\e)$ is the linearized Einstein tensor, which is a  
--$g_{\a \b}$ dependent-- linear functional on $h_{\a \b}$:
\begin{equation}
 \dot{G}_{\a\b}[h]=-\tfrac{1}{2}\Box h_{\a\b}-\tfrac{1}{2}\c_{\a}\c_{\b}h+\c^{\g}\c_{(\a}h_{\b)\g}
 +\tfrac{1}{2}g_{\a\b}(\Box h -\c^{\g}\c^{\d}h_{\g\d}),
\end{equation}
where $h=g^{\g\d}h_{\g\d}$.
In order to relate the perturbed Weyl scalars in theorem (\ref{spin2SV}) to the linearized Einstein tensor,
we use the linearized Bianchi identities in the following way: 
let $O^{\a\b\g}$ be a linear differential operator such that $O^{\a\b\g}=O^{[\a\b]\g}$ and $g_{\a\g}O^{\a\b\g}=0$ 
(see section \ref{gravsec} for explicit expressions of $O^{\a\b\g}$ in spinor form). 
As we will show, applying $O^{\a\b\g}$ to the Bianchi identities on an {\em arbitrary} spacetime, one gets
\begin{equation}\label{keyeq0}
 O^{\a\b\g}\c^{\d}C_{\a\b\g\d}=-O^{\a\b\g}\c_{\a}R_{\b\g},
\end{equation}
where $C_{\a\b\g\d}$ is the Weyl tensor.
The idea is to choose $O^{\a\b\g}$ such that the left hand side of (\ref{keyeq0})
is a decoupled equation for some Weyl scalar plus additional terms that vanish when linearizing.
When  we linearize the right hand side 
of (\ref{keyeq0}) around a vacuum solution (i.e. with $R_{\a\b}|_{\e=0}=0$), the linearization operator $\frac{d}{d\e}|_{\e=0}$ 
commutes with $O^{\a\b\g}\c_{\a}$ 
and we are left with a {\em background} operator acting on the linearized Ricci tensor:
\begin{equation}
 \left.\tfrac{d}{d\e}\right|_{\e=0}\left[O^{\a\b\g}\c^{\d}C_{\a\b\g\d}\right] 
 =-O^{\a\b\g}\c_{\a}\left[\left.\tfrac{d}{d\e}\right|_{\e=0}R_{\b\g}\right].
\end{equation}
Note that the symmetries of $O^{\a\b\g}$ are such that we can add a term propotional to the metric in 
the right hand side of (\ref{keyeq0}), this allows to replace $R_{\a\b}$ by the Einstein tensor and to include a cosmological 
constant term (in which case we consider {\em $\l$-vacuum} background solutions, $(G_{\a\b}+\l g_{\a\b})|_{\e=0}=0$). 
See section \ref{linbianchi} for details. 
When combined with theorem (\ref{spin2SV}), and using the tensors (\ref{W0}), (\ref{W2}) and (\ref{W4}),
the previous idea leads to the following result:
\begin{thmbis}{spin2SV}[spin $\s=2$, tensor version]\label{spin2}
 Consider an arbitrary metric perturbation $h_{\a\b}$ on a vacuum Petrov type D spacetime with cosmological constant $\l$.
 Then we have the following equalities:
 \begin{equation}
 \Psi_2^{4/3} \stackbin[]{0}{W}{}^{\a\g\b\d}\c_{\d}\left[\Psi_2^{-4/3}\c_{\g}(\dot{G}_{\a\b}[h]+\l h_{\a\b})\right]
 =(\teuk_{+4}-16\Psi_2 +\tfrac{2}{3}\l)\dot\Psi_0[h] \label{grav1} 
 \end{equation}
 \begin{multline}
 \Psi^{2/3}_{2}\stackbin[]{2}{W}{}^{\a\g\b\d}\c_{\d}\left[\Psi^{-4/3}_2\c_{\g}(\dot{G}_{\a\b}[h]+\l h_{\a\b})\right]\\
  =6\left[(\Box+8\Psi_2 +\tfrac{2}{3}\l)[\Psi^{-2/3}_2\dot{\Psi}_2[h]]+3(\dot{\Box}_h+\tfrac{\dot{R}_h}{6})\Psi^{1/3}_2\right], \label{grav2}
 \end{multline}
 \begin{equation}
  \stackbin[]{4}{W}{}^{\a\g\b\d}\c_{\d}\left[\Psi_2^{-4/3}\c_{\g}(\dot{G}_{\a\b}[h]+\l h_{\a\b})\right]
  =(\teuk_{-4}-16\Psi_2 +\tfrac{2}{3}\l)[\Psi^{-4/3}_2\dot\Psi_4[h]].  \label{grav3}
 \end{equation}
\end{thmbis}

In the next sections we put the equalities in the previous theorems in an operator identity form in the spirit of 
\cite{Wald}. This provides a way to reconstruct solutions of the original field equations from solutions of the decoupled 
equations of which  (\ref{we})-(\ref{metricschw}) is a particular case (see lemma \ref{lemma-schw} and below).\\
The proof of all theorems requires  combining spinor and Geroch-Held-Penrose (GHP) techniques, 
although in section \ref{sphsym-sec} we also give some alternative proofs using the $2+2$ decomposition of 
warped product spacetimes, which is useful for connecting our formalism to the traditional approach in the spherically 
symmetric case.

\subsection{Conventions and overview}

We will assume the spacetime to be a $3+1$ dimensional, orientable, Lorentzian manifold with metric signature $(+---)$, 
and we further assume that it admits a spinor structure.
Greek indices refer to spacetime indices, and (primed and unprimed) latin capital indices are spinor indices. 
Additional notation, when needed, will be explained in the corresponding sections.
Throughout the paper, we will omit the soldering forms $\sigma^{\a}{}_{AA'}$ for the correspondence between spinors
and tensors. For background on the 2-spinor and GHP formalisms, see for example \cite{Penrose1,Penrose2}. 
We will often use `$\l$-vacuum spacetime' for referring to a spacetime which is vacuum apart from a nonzero 
cosmological constant $\l$.
For the sign conventions we use regarding curvature tensors, see appendix \ref{appA}.

In section \ref{sec-GHP} we explain the methods we will use in the calculations of this paper, 
in particular the basics of the GHP formalism, the properties of Petrov type D spacetimes relevant 
for this work, and a review of Wald's method of adjoint operators, together with a unified form of 
the operator to be applied to a spin-$\s$ field in order to get decoupled equations for its components.
Sections \ref{diracsec}, \ref{maxwellsec} and \ref{gravsec} are devoted to the proof of theorems 
\ref{spin1/2}, \ref{spin1} and \ref{spin2SV} respectively, we also give covariant, compact expressions for solutions 
of the field equations in terms of solutions of decoupled equations; in particular, in section \ref{gravsec} we 
show in detail how to relate the equations for perturbed Weyl scalars to the linearized Einstein tensor, and then how to 
construct a solution of the linearized Einstein equations from solutions of the decoupled equations.
In section \ref{sphsym-sec} we give the relation of our methods and results with the $2+2$ decomposition of spacetimes 
with warped product structure, particularized to the Schwarzschild-(A)dS solution. 
In particular, we demonstrate the origin of (\ref{we}), (\ref{metricschw}) (section \ref{met.rec.schw}) and 
(\ref{Phi}) (section \ref{kyt.schw}).
Finally, the conclusions of this work are presented in section \ref{conclusions}, together with a summary of previously 
known results. 
We also include an appendix collecting relevant formulae for the proofs of the main theorems.

\section{Spinor and GHP methods}\label{sec-GHP}

The purpose of this section is to introduce the different techniques we will use in the calculations of this paper.
In section \ref{prelimsec} we discuss briefly the spinor fields we will consider in this work and the associated 
scalar, decoupled equations.
In section \ref{ghp-sec} we give the basics of the GHP formalism needed to understand 
the notation and calculations of the next sections (we mainly follow \cite{Geroch:1973am} and section 4.12 in \cite{Penrose1}).
In section \ref{teuk-sec} we present the compact form of the Teukolsky equations using weighted wave operators 
\cite{Bini, Andersson1}, we will use them in the case of extreme spin weight.
The characteristics of Petrov type D spacetimes relevant to this work are presented in section \ref{typeD-sec}, 
together with the properties of the Killing spinor associated to these solutions. Finally, in section \ref{adjop-sec} 
we recall the method of adjoint operators due to Wald \cite{Wald}, that will be central in this work, and we give the 
general $(\s,s)$-operator that maps off-shell the field equations into scalar, decoupled equations for the spin weight $s$ 
component of a spin-$\s$ field.

\subsection{Preliminaries}\label{prelimsec}

The fields one typically considers in the study of black hole stability are obtained as a generalization of the situation
in the Minkowski space. The possible physical fields that can exist on a flat spacetime are in turn determined by very 
general symmetry arguments. More precisely, the (massless) physical fields are classified by studying the massless 
irreducible representations of the universal covering of the Poincar\'e group, which is the isometry group of Minkowski
space. This leads to the notion of massless free fields of {\em spin} \footnote{or, more properly for the massless case, 
{\em helicity}.} $\s$: totally symmetric spinors $\phi_{A_1...A_{2\s}}=\phi_{(A_1...A_{2\s})}$ with $2\s$ indices 
satisfying the equation\footnote{in the case of spin $\s=0$ the field satisfies the massless wave equation $\Box\phi=0$.}
\begin{equation}\label{field-mink}
 \p^{A_1A'_1}\phi_{A_1...A_{2\s}}=0,
\end{equation}
where $\p_{AA'}=\sigma^{\a}{}_{AA'}\p_{\a}$, with $\sigma^{\a}{}_{AA'}$ the soldering form and $\p_{\a}$  
derivatives with respect to global inertial coordinates. Physically important examples of (\ref{field-mink}) 
are the Dirac ($\s=1/2$), Maxwell ($\s=1$) and linearized gravitational fields ($\s=2$).
For curved spacetimes, however, the existence of spinor fields depends on whether or not it is possible to define 
a spinor structure, for which there are some topological obstructions \cite{Geroch,Geroch2}. If  the topological 
conditions are met, spinors are defined by using the local $SO(1,3)$ symmetry, and the generalization of 
(\ref{field-mink}) to a curved space is achieved by the minimal substitution $\p_{\a}\to\c_{\a}$,
\begin{equation}\label{field-curv}
\c^{A_1A'_1}\phi_{A_1...A_{2\s}}=0,
\end{equation}
where now $\phi_{A_1...A_{2\s}}$ is a cross section of the corresponding spinor bundle. 
The {\it spin} now labels the irreducible representations of $SL(2,\mathbb{C})$, which is the covering of $SO(1,3)$.
On the other hand, even if the spacetime admits a spin structure, 
the existence of non trivial solutions of (\ref{field-curv}) for $\s>1$ is constrained by algebraic consistency conditions:
if we assume (\ref{field-curv}) holds and take an extra covariant derivative we find that
\begin{equation}\label{buchdahl}
 \phi_{ABC(A_4...A_{2\s}}\psi_{A_{2\s+1})}{}^{ABC}=0,
\end{equation}
where $\psi_{ABCD}$ is the Weyl curvature spinor.
This equation is sometimes referred to as the {\it Buchdahl constraint}, and it imposes strong restrictions on the geometry 
of the background spacetime (see e.g. \cite{Szekeres}). 
Moreover, the gravitational perturbations of a generic curved space, represented by the linearized Weyl spinor, do not 
satisfy (\ref{field-curv}), i.e. they involve a non-trivial right hand side in this equation (see e.g. \cite{curtis}), 
and the algebraic specialty is generally not preserved under perturbations \cite{Araneda}.
Therefore, we will focus on the spin $\s=\frac{1}{2},1$ cases of (\ref{field-curv}), 
while for the spin 2 case we will use the linearization of (\ref{field-curv}).

As mentioned in the introduction, a useful simplification in the study of solutions of tensorial/spinorial field equations 
would be to find a scalar equation describing the system. Of course, one can obtain a {\em set} of scalar equations 
on an arbitrary spacetime by simply projecting the field equations on a basis frame at each tangent space. 
Simplifications useful for calculations are achieved if the basis 
frame one chooses can be related to the particular geometric structure of the spacetime. This is the case for example when the 
geometry possesses distinguished directions, like in the algebraically special spacetimes of the Petrov classification. The 
Petrov type D is especially relevant for the black hole stability problem, since the Kerr family of stationary, vacuum black 
hole solutions corresponds to this case. Two (repeated) null directions are preferred at each point in this class of spaces, 
and, by adapting a null frame to them, a formalism especially suited for this situation can be implemented, namely the GHP 
formalism. However, the system of equations obtained this way typically consists of several interrelated equations which in 
principle cannot be analyzed separately. That is to say, the equations are generally {\em coupled}, in the sense that each 
one of them involves more than one of the components of the field relative to the basis frame one have chosen. 

In a flat space, given a spin-$\s$ field (\ref{field-mink}), a single scalar equation can be obtained by 
using {\em Killing spinors}\footnote{not to be confused with the homonymous object in the mathematics and supergravity 
communities \cite{vanNieu}} (see \cite[section 6.4]{Penrose2}): if $L_{A_1...A_{2\s}}$ is
a $2\s$-index Killing spinor, the field $\Phi=L^{A_1...A_{2\s}}\phi_{A_1...A_{2\s}}$ satisfies 
the wave equation, $\Box\Phi=0$. In curved spacetimes, the situation is more subtle 
because the existence of Killing spinors imposes restrictions on the curvature to algebraically special cases. 
On the other hand, even if Killing spinors are available, it is expected the appearance of curvature terms in wave-like 
equations for $\Phi=L^{A_1...A_{2\s}}\phi_{A_1...A_{2\s}}$. For example,
Petrov type D spacetimes admit a 2-index Killing spinor $K_{AB}$ (see subsection \ref{typeD-sec} below), which is related 
to various symmetries of these spaces. The scalar field $\Phi\equiv K^{AB}\phi_{AB}$, where $\phi_{AB}$ is a spin-1 field
(\ref{field-curv}), can be shown to satisfy the {\em Fackerell-Ipser equation}
\begin{equation}
 (\Box+2\Psi_2)\Phi=0.
\end{equation}
This equation was found in \cite{Fackerell} by other means in the particular case of the Kerr solution, but it is valid 
for all type D vacuum spacetimes.

On the other hand, it is possible that the scalar equations we are looking for involve wave operators distinct from 
the traditional D'Alembertian $\Box=g^{\a\b}\c_{\a}\c_{\b}$. We can think of this situation in the following geometrical terms.
Let $P\xrightarrow{\pi}\mathcal{M}$ be a principal fiber bundle with structure group $G$ over the spacetime 
$\mathcal{M}$, and let $\w_{\a}$ be a $\mathfrak{g}$-valued connection 1-form on $P$, where $\mathfrak{g}=\text{Lie}(G)$ is 
the Lie algebra of $G$.
Tensor fields on $\mathcal{M}$ are sections of associated bundles to $P$, $E=P\times_{\rho}V$, where
$(\rho,V)$ is a representation of $G$ on the vector space $V$. The covariant derivative on $E$ is induced by the 
connection 1-form on $P$, and, acting on a cross-section $\psi$ of $E$, it is explicitly given by
\begin{equation}\label{der}
 \Theta_{\a}\psi=\p_{\a}\psi-\rho'(\w_{\a})\psi,
\end{equation}
where $\rho':\mathfrak{g}\to \text{gl}(V)$ is the associated representation of the Lie algebra $\mathfrak{g}$ 
(see e.g. \cite{Nakahara}). 
Formula (\ref{der}) is very useful; it generalizes the expression for the covariant derivative of tensor and 
spinor fields occurring in General Relativity or Yang-Mills theories. (For example, the covariant derivative 
appearing in (\ref{field-curv}) is a particular case of (\ref{der}), where the Lie algebra is 
$\mathfrak{g}=\mathfrak{so}(1,3)$ and the connection 1-form is the spin connection.) 
The equation we are looking for may then involve a wave operator formed as $g^{\a\b}\Theta_{\a}\Theta_{\b}$. 
This is actually the case for the Teukolsky equation \cite{Andersson1}, where, in the context of the GHP formalism, 
the gauge group is $\mathbb{C}^{\times}$ and its representations on the fields 
of interest are labeled by an integer number $p$; see section \ref{teuk-sec}.

\subsection{GHP formalism}\label{ghp-sec}

The GHP calculus is especially suited for situations in which two null directions $l^{\a}$ and $n^{\a}$ on the spacetime 
are distinguished, like in the case of Petrov type D spaces we are interested in. 
We align a spin dyad $\{o^A,\iota^A\}$ to these null directions, with $o_A\iota^A=1$.
The relation of this dyad with a null tetrad is as usual,
\begin{equation}\label{tetrad}
 l^{\a}=o^A\bar{o}^{A'}, \;\;\; n^{\a}=\iota^A\bar{\iota}^{A'}, \;\;\; m^{\a}=o^A\bar{\iota}^{A'}, 
 \;\;\; \bar{m}^{\a}=\iota^A\bar{o}^{A'}.
\end{equation}
As the normalization is preserved throughout the spacetime under $o^A\to\lambda o^A$, $\iota^A\to\lambda^{-1}\iota^A$, 
where $\lambda$ is a nowhere vanishing complex scalar field, 
then fixing the null directions reduces the local $SO(1,3)$ freedom in choosing an orthonormal tetrad, to a 
gauge freedom represented by a 2-dimensional subgroup of $SO(1,3)$, which is isomorphic to $\mathbb{C}^{\times}$ 
(the multiplicative group of complex numbers).
In more geometrical terms \cite{ehlers}, we get a reduction of the orthonormal frame bundle, with structure group $SO(1,3)$, to a 
principal fiber bundle $B\xrightarrow{\pi}\mathcal{M}$ with structure group $\mathbb{C}^{\times}$. 
The $\text{Lie}(\mathbb{C}^{\times})$-valued connection form is
\begin{equation}\label{connection}
 \w_{\a}=\e n_{\a}-\e' l_{\a}+\b' m_{\a}-\b\bar{m}_{\a},
\end{equation}
and it transforms under $\mathbb{C}^{\times}$ as the gauge potential of an abelian Lie group, 
$\w_{\a}\to\w_{\a}+\lambda^{-1}\c_{\a}\lambda$.\\
The components of a tensor field projected on the null tetrad (or a spinor field projected on the dyad) 
are complex fields on the spacetime or, more precisely, fields $\eta:B\to\mathbb{C}$, since they are associated to 
a particular frame.
These components have a well-defined transformation law under a change of frame; in other words, 
they transform under the representation $\Pi_{p,q}:\mathbb{C}^{\times}\to\text{GL}(\mathbb{C})$ of 
$\mathbb{C}^{\times}$ on $\mathbb{C}$ given by
\begin{equation}\label{rep}
 \eta\mapsto\Pi_{p,q}(\lambda)\eta:=\lambda^{p}\bar\lambda^{q}\eta,
\end{equation}
for some integers $p,q$. Elements transforming under this representation are known as {\em weighted quantities of type $\{p,q\}$}, 
or, alternatively, quantities of {\em spin weight} $s=(p-q)/2$ and {\em boost weight} $b=(p+q)/2$. 
While the quantities of a well-defined type $\{p,q\}$ form a complex vector space (carrying the representation (\ref{rep}) 
of $\mathbb{C}^{\times}$), the quantities of all types together form a graded algebra.
The properly weighted spin coefficients are $\rho,\tau,\kappa,\sigma,\rho',\tau',\kappa',\sigma'$ 
(see \cite[Eq.(4.5.21)]{Penrose1} for the definition of the spin coefficients as derivatives of the dyad spinors), while 
the coefficients $\b,\e,\b',\e'$ do not have a well-defined type, they enter in the formalism in the definition of 
the connection form (\ref{connection}). 
On the other hand, the components $\chi_0=\chi_A o^A$ and $\chi_1=\chi_A\iota^A$ of a spinor field $\chi_A$ are of type 
$\{1,0\}$ and $\{-1,0\}$ respectively, while the Maxwell components $\phi_i$ are of type $\{2-2i,0\}$, $i=0,1,2$, and
the Weyl scalars $\Psi_i$, $i=0,...,4$ have types $\{4-2i,0\}$. \\
The representation of the Lie algebra $\mathfrak{g}=\text{Lie}(\mathbb{C}^{\times})$ associated to 
(\ref{rep}), $\pi_{p,q}:\mathfrak{g}\to\text{gl}(\mathbb{C})$ , is easily calculated as
\begin{equation}
 \pi_{p,q}(X)\eta=(pX+q\bar{X})\eta.
\end{equation}
Then, according to (\ref{der}), 
the covariant derivative on sections of the associated bundles $E_{p,q}:=B\times_{\Pi_{p,q}}\mathbb{C}$ is
\begin{equation}\label{Theta}
 \Theta_{\a} = \c_{\a}-p\w_{\a}-q\bar\w_{\a}
\end{equation}
(the inclusion of the Levi-Civita derivative $\c_{\a}$ allows to apply this formula to weighted spinor and tensor fields, besides 
weighted scalars).
The traditional weighted derivative operators \th , \th$'$, \dh{} and \dh$'$ are simply the directional derivatives along 
the null tetrad, \th$=l^{\a}\Theta_{\a}$, \th$'=n^{\a}\Theta_{\a}$, \dh$=m^{\a}\Theta_{\a}$ and \dh$'=\bar{m}^{\a}\Theta_{\a}$.
This is in contrast to the non-weighted directional derivatives of the Newman-Penrose formalism, 
$D=l^{\a}\c_{\a}$, $D'=n^{\a}\c_{\a}$, $\d=m^{\a}\c_{\a}$ and $\d'=\bar{m}^{\a}\c_{\a}$. The relation between both classes 
of operators can be inferred from (\ref{Theta}) and (\ref{connection}): acting on a type $\{p,q\}$ quantity, we have
\begin{eqnarray}
 \text{\th}&=&D-p\e-q\bar\e, \label{thorn} \\
 \text{\dh}&=&\d-p\b+q\bar\b', \label{eth} \\
 \text{\th}'&=&D'+p\e'+q\bar\e',\label{thorn'} \\
 \text{\dh}'&=&\d'+p\b'-q\bar\b. \label{eth'}
\end{eqnarray}
A very useful GHP operation taking weighted quantities into weighted quantities is the so-called {\em prime} operation, 
which is defined by the interchange $o^{A}\leftrightarrow\iota^{A}$.
It is easy to see that if $\eta$ is of type $\{p,q\}$, then $\eta'$ is of type $\{-p,-q\}$. 
This operation allows to halve the number of Newman-Penrose equations which are properly weighted, namely
the Ricci identities involving derivatives of the weighted spin coefficients\footnote{the greek letter $\Lambda$ 
(traditionally associated to the cosmological constant) represents the scalar curvature and is conventional in the 
two-spinor formalism \cite{Penrose1,Penrose2}, this is the reason why we use $\l$ for the cosmological constant.},
\begin{eqnarray}
 \text{\th}\rho-\text{\dh}'\kappa&=&\rho^2+\sigma\bar\sigma-\bar\kappa\tau-\tau'\kappa+\Phi_{00}\\
 \text{\th}\sigma-\text{\dh}\kappa&=&(\rho+\bar\rho)\sigma-(\tau+\bar\tau')\kappa+\Psi_0 \label{ricci2}\\
 \text{\th}\tau-\text{\th}'\kappa&=&(\tau-\bar\tau')\rho+(\bar\tau-\tau')\sigma+\Psi_1+\Phi_{01}\\
 \text{\dh}\rho-\text{\dh}'\sigma&=&(\rho-\bar\rho)\tau+(\bar\rho'-\rho')\kappa-\Psi_1+\Phi_{01}\\
 \text{\dh}\tau-\text{\th}'\sigma&=&-\rho'\sigma-\bar\sigma'\rho+\tau^2+\kappa\bar\kappa'+\Phi_{02}\\
 \text{\th}'\rho-\text{\dh}'\tau&=&\rho\bar\rho'+\sigma\sigma'-\tau\bar\tau-\kappa\kappa'-\Psi_2-2\Lambda.
\end{eqnarray}
The prime of these equations gives six more properly weighted Ricci equations. 
The remaining Newman-Penrose equations involve derivatives of spin coefficients not properly weighted; 
in the GHP formalism they enter in the commutation relations for the derivative operators:
\begin{eqnarray}
\nonumber [\text{\th},\text{\th}']&=&(\bar\tau-\tau')\text{\dh}+(\tau-\bar\tau')\text{\dh}'-p(\kappa\kappa'-
   \tau\tau'+\Psi_2+\Phi_{11}-\Lambda)\\
 & & -q(\bar\kappa\bar\kappa'-\bar\tau\bar\tau'+\bar\Psi_2+\Phi_{11}-\Lambda),\\
\nonumber [\text{\th},\text{\dh}]&=&\bar\rho\text{\dh}+\sigma\text{\dh}'-\bar\tau'\text{\th}-\kappa\text{\th}'
  -p(\rho'\kappa-\tau'\sigma+\Psi_1)\\
 & &-q(\bar\sigma'\bar\kappa-\bar\rho\bar\tau'+\Phi_{01}),\\
\nonumber [\text{\dh},\text{\dh}']&=&(\bar\rho'-\rho')\text{\th}+(\rho-\bar\rho')\text{\th}'+p(\rho\rho'+\sigma\sigma'
 +\Psi_2-\Phi_{11}-\Lambda)\\
 & & -q(\bar\rho\bar\rho'-\bar\sigma\bar\sigma'+\bar\Psi_2-\Phi_{11}-\Lambda).
\end{eqnarray}
We also note the following commutation relation, in which $a$ is an arbitrary constant and $\eta$ is type $\{p,0\}$:
\begin{eqnarray}
\nonumber [\text{\th}-a\rho,\text{\dh}-a\tau]\eta &=&\bar{\rho}(\text{\dh}-a\tau)\eta-\bar{\tau}'(\text{\th}-a\rho)\eta-(2a+p)\Psi_1\eta\\
\nonumber & & +\left(\text{\dh}'\eta+p\tau'\eta+a\eta(\text{\dh}'-\bar\tau+\tau')\right)\sigma \\
 & & -\left(\text{\th}'\eta+p\rho'\eta+a\eta(\text{\th}'-\bar\rho'+\rho')\right)\kappa. \label{comm4generic}
\end{eqnarray}
For type D spacetimes, the terms proportional to $\sigma$, $\kappa$ and $\Psi_1$ vanish, and we are left with 
\cite{Andersson1}
\begin{equation}\label{comm4}
 [\text{\th}-a\rho,\text{\dh}-a\tau]\eta=\bar{\rho}(\text{\dh}-a\tau)\eta-\bar{\tau}'(\text{\th}-a\rho)\eta,
\end{equation}
this relation will be very useful in the following sections.

In order to find the spinor operators that map field equations to decoupled scalar equations, 
we will need the explicit form of $\c^{A}_{B'}\chi_{A}$, $\c^{A}_{B'}\phi_{AB}$ and $\c^{A}_{B'}\psi_{ABCD}$ in its components 
in the $\{o_A,\iota_A\}$ basis. 
This can be obtained readily by using formulae (4.12.27) in \cite{Penrose1}. 
For Dirac fields, this gives
\begin{eqnarray}
\nonumber \c^{A}_{B'}\chi_{A}&=&[(\text{\th}'-\rho')\chi_0-(\text{\dh}-\tau)\chi_1]\bar{o}_{B'}\\
 & & +[(\text{\th}-\rho)\chi_1-(\text{\dh}'-\tau')\chi_0]\bar{\iota}_{B'}, \label{diraceq}
\end{eqnarray}
while for the Maxwell spinor, we get
\begin{eqnarray}
\nonumber \c^{A}_{B'}\phi_{AB}&=&[(\text{\dh}-2\tau)\phi_1-(\text{\th}'-\rho')\phi_0+\sigma\phi_2]\iota_B \bar{o}_{B'}\\
\nonumber & & +[(\text{\th}'-2\rho')\phi_1-(\text{\dh}-\tau)\phi_2+\kappa'\phi_0]o_B \bar{o}_{B'}\\
\nonumber & & +[(\text{\dh}'-\tau')\phi_0-(\text{\th}-2\rho)\phi_1-\kappa\phi_2]\iota_B \bar{\iota}_{B'}\\
 & & +[(\text{\th}-\rho)\phi_2-(\text{\dh}'-2\tau')\phi_1-\sigma'\phi_0]o_B \bar{\iota}_{B'}, \label{maxwelleq}
\end{eqnarray}
and similarly for the Weyl spinor
\begin{eqnarray}
\nonumber \c^{A}_{B'}\psi_{ABCD}&=&-[(\text{\th}'-\rho')\Psi_0-(\text{\dh}-4\tau)\Psi_1-3\sigma\Psi_2]\iota_{BCD}\bar{o}_{B'}\\
\nonumber & & +3[(\text{\th}'-2\rho')\Psi_1-(\text{\dh}-3\tau)\Psi_2+\kappa'\Psi_0-2\sigma\Psi_3]\iota_{(BC}o_{D)}\bar{o}_{B'}\\
\nonumber & & -3[(\text{\th}'-3\rho')\Psi_2-(\text{\dh}-2\tau)\Psi_3+2\kappa'\Psi_1-\sigma\Psi_4]\iota_{(B}o_{CD)}\bar{o}_{B'}\\
\nonumber & & +[(\text{\th}'-4\rho')\Psi_3-(\text{\dh}-\tau)\Psi_4+3\kappa'\Psi_2]o_{BCD}\bar{o}_{B'}\\
\nonumber & & -[(\text{\th}-4\rho)\Psi_1-(\text{\dh}'-\tau')\Psi_0+3\kappa\Psi_2]\iota_{BCD}\bar{\iota}_{B'}\\
\nonumber & & +3[(\text{\th}-3\rho)\Psi_2-(\text{\dh}'-2\tau')\Psi_1+2\kappa\Psi_3-\sigma'\Psi_0]\iota_{(BC}o_{D)}\bar{\iota}_{B'}\\
\nonumber & & -3[(\text{\th}-2\rho)\Psi_3-(\text{\dh}'-3\tau')\Psi_2+\kappa\Psi_4-2\sigma'\Psi_1]\iota_{(B}o_{CD)}\bar{\iota}_{B'}\\
 & & +[(\text{\th}-\rho)\Psi_4-(\text{\dh}'-4\tau')\Psi_3-3\sigma'\Psi_2]o_{BCD}\bar{\iota}_{B'}, \label{bianchieq}
\end{eqnarray}
where $\iota_{ABC}=\iota_A\iota_B\iota_C$, $\iota_{AB}=\iota_A\iota_B$, $o_{ABC}=o_Ao_Bo_C$ and $o_{AB}=o_Ao_B$.
The Dirac and Maxwell equations and the vacuum Bianchi identities of the GHP formalism are given simply by setting all previous 
components equal to zero independently.

\subsection{The Teukolsky equations}\label{teuk-sec}

The Teukolsky equations \cite{Teukolsky}, which were originally found by using the Newman-Penrose formalism,
can be put in a compact form by using a modification of the covariant derivative (\ref{Theta}). 
With this purpose we define the 1-form $B_{\a}$ by
\begin{equation}
 B_{AA'}:=-\rho\iota_{A}\bar\iota_{A'}+\tau\iota_{A}\bar{o}_{A'}
\end{equation}
and, following \cite{Andersson1}, we introduce a new connection $\Gamma_{\a}:=B_{\a}-\w_{\a}$ on $E_{p,q}$; explicitly:
\begin{equation}\label{Gamma}
 \Gamma_{\a}= (\e-\rho)n_{\a}-\e' l_{\a}+\b' m_{\a}+(\tau-\b)\bar{m}_{\a}.
\end{equation}
Since in the next sections we will work on the Dirac, Maxwell and Weyl scalars, and they are all type
$\{p,0\}$ quantities, we need only define the weighted wave operator 
\begin{equation}\label{boxp1}
 \teuk_{p}:=(\c^{\a}+p\Gamma^{\a})(\c_{\a}+p\Gamma_{\a}).
\end{equation}
Note that $\teuk_{0}\equiv\Box$.
Its expression in terms of the weighted directional derivatives is
\begin{eqnarray}
\nonumber \teuk_{p}&=&2(\text{\th}-p\rho-\bar\rho)(\text{\th}'-\rho')-2(\text{\dh}-p\tau-\bar\tau')(\text{\dh}'-\tau')\\
& & +[(3p-2)\Psi_2-4\Lambda]+2(p-1)(\kappa\kappa'-\sigma\sigma'), \label{boxp2}
\end{eqnarray}
where $\Lambda=R/24$, with $R$ the Ricci scalar. 
The Teukolsky equations for a field $\Phi^{(s)}$ of spin weight $s$, on a background type D vacuum spacetime, are then \cite{Bini}
\begin{equation}\label{teuk}
 (\teuk_{2s}-4s^2\Psi_2)\Phi^{(s)}=0.
\end{equation}
We will see that several of the identities we will prove follow easily from applying the prime operation to 
other identities. For this, we need to know the behavior of $\teuk_p$ under the prime operation. 
As proved in \cite{Andersson1}, acting on a type $\{p,0\}$ quantity $\Phi$, $\teuk_{p}$ transforms as
\begin{equation}\label{boxt'}
 \teuk'_{p}\Phi'=\Psi^{p/3}_2\teuk_{-p}(\Psi^{-p/3}_2\Phi').
\end{equation}

\subsection{Petrov type D spacetimes}\label{typeD-sec}

In the Petrov classification of spacetimes, type D spaces are characterized by the existence of two (repeated) principal 
null directions (PNDs). As mentioned, the Kerr-Newman-(A)dS family of stationary, electrovacuum black hole 
solutions belongs to this class. Aligning a spin dyad $\{o^A,\iota^A\}$ to the PNDs,
several of the GHP coefficients and Weyl scalars vanish:
\begin{eqnarray}
 && \kappa=\kappa'=\sigma=\sigma'=0=\Psi_0=\Psi_1=\Psi_3=\Psi_4, \label{ghpD} \\
 && \Psi_2=\psi_{ABCD}o^{A}o^{B}\iota^{C}\iota^{D} \neq 0.
\end{eqnarray}
The Weyl curvature spinor has the explicit form
\begin{equation}\label{weylD}
 \psi_{ABCD}=6\Psi_2o_{(A}o_B\iota_C\iota_{D)},
\end{equation}
and the Bianchi identities of a $\l$-vacuum, type D spacetime are simply 
\begin{equation}\label{bianchiD}
\text{\th}\Psi_2=3\rho\Psi_2, \;\;\; \text{\dh}\Psi_2=3\tau\Psi_2 
\end{equation}
and their primed versions. If we introduce a 1-form $A_{AA'}$ as
\begin{equation}\label{AD}
A_{AA'}:=\Psi^{1/3}_2\c_{AA'}\Psi^{-1/3}_2, 
\end{equation}
Bianchi identities imply that
\begin{equation}
 A_{AA'}=-\rho\iota_{A}\bar\iota_{A'}-\rho' o_{A}\bar{o}_{A'}+\tau' o_{A}\bar\iota_{A'}+\tau\iota_{A}\bar{o}_{A'}, \label{A} \\ 
\end{equation}
Expressions (\ref{A}) and (\ref{AD}) will be both very useful in the applications.\\
A very important property of $\l$-vacuum type D spaces, is that they admit a 2-index {\em Killing spinor}, namely 
a symmetric spinor $K_{AB}=K_{(AB)}$ satisfying the twistor equation
\begin{equation}
 \c_{C'(C}K_{AB)}=0,
\end{equation}
see \cite{Penrose3} (also \cite[section 6.7]{Penrose2}). The explicit form of $K_{AB}$ in the principal dyad $\{o^A,\iota^A\}$ is
\begin{equation}\label{KSD}
 K_{AB}=k\Psi^{-1/3}_2 o_{(A}\iota_{B)},
\end{equation}
where $k$ is an arbitrary complex constant.
This object is associated to several kind of symmetries and `hidden' symmetries of the spacetime, as we briefly recall in the 
following. Taking the divergence of (\ref{KSD}) in an unprimed index, we get 
$\xi^{AA'}\equiv\c^{BA'}K_{B}{}^{A}$, which turns out to be complex a Killing vector \cite[Proposition (6.7.17)]{Penrose2}, 
and in the case of the Kerr solution it is proportional to the (asymptotically) timelike Killing field. 
The tensor fields associated to $K_{AB}$ are the 2-forms
\begin{eqnarray}
 Y_{\a\b}&=&iK_{AB}\bar\e_{A'B'}-i\bar{K}_{A'B'}\e_{AB}, \label{KY} \\
 {}^{*}Y_{\a\b}&=&K_{AB}\bar\e_{A'B'}+\bar{K}_{A'B'}\e_{AB}, \label{DKY} 
\end{eqnarray}
and they turn out to be {\em conformal Killing-Yano tensors}.
In the case in which $\xi^{\a}$ is real (for example in the Kerr and Schwarzschild solutions), 
$Y_{\a\b}$ is an ordinary Killing-Yano tensor:
\begin{equation}
 \c_{(\a}Y_{\b)\g}=0
\end{equation}
(see \cite{Jezierski:2005cg} for a thorough account of these tensor fields in the Kerr case). 
In \cite{Gibbons} it was shown that $Y_{\a\b}$ generates conserved supercharges for the supersymmetric extension of the 
geodesic motion (see also \cite{Santillan} where further applications of Killing-Yano tensors are discussed). 
On the other hand, the square $H_{\a\b}=Y_{\a\g}Y^{\g}{}_{\b}$ is a Killing tensor, $\c_{(\a}H_{\b\g)}=0$, 
whose existence in the Kerr spacetime allows to completely integrate the geodesic equation \cite{Penrose3}.
Finally, the vector $\eta^{\a}=H^{\a\b}\xi_{\b}$ is also a Killing vector (which is linearly independent from $\xi^{\a}$ in the Kerr 
case, and it is zero in Schwarzschild). 
Apart from subsection \ref{sw0Msection} below, in this work we do not assume that the Killing vector $\xi^{\a}$ is real.\\
The Weyl spinor (\ref{weylD}) of a type D space can be written in terms of the Killing spinor (\ref{KSD}) in the form
\begin{equation}\label{WSD}
 \psi_{ABCD}=\tfrac{6}{k^2}\Psi^{5/3}_2 K_{AB}K_{CD}-\tfrac{1}{2}\Psi_2(\e_{AD}\e_{CB}+\e_{AC}\e_{DB}).
\end{equation}
This leads to the following expression for the anti-self-dual Weyl tensor in terms of the Killing-Yano tensors
\begin{equation}\label{sdweylD}
 \wt{C}_{\a\b\g\d}=-\tfrac{6}{k^2}\Psi^{5/3}_2\wt{Y}_{\a\b}\wt{Y}_{\g\d}+\Psi_2(g_{\a[\g}g_{\d]\b}+\tfrac{i}{2}\e_{\a\b\g\d}),
\end{equation}
where 
\begin{equation}
 \wt{Y}_{\a\b}:=\tfrac{1}{2}(Y_{\a\b}+i{}^{*}Y_{\a\b}).
\end{equation}
We recall that, according to our conventions, we have
\begin{equation}\label{weyl}
 \wt{C}_{\a\b\g\d}=\tfrac{1}{2}(C_{\a\b\g\d}+i{}^{*}C_{\a\b\g\d}).
\end{equation}
Formula (\ref{sdweylD}) will be particularly useful in section \ref{sphsym-sec}, where we explicitly evaluate 
our results in the Schwarzschild-(A)dS spacetime.

\subsection{Adjoint operators}\label{adjop-sec}

In this subsection we review Wald's idea of adjoint operators \cite{Wald}, since it plays a central role in this work.
Suppose that we are interested in solutions $f$ of the differential equation $\mathcal{E}(f)=0$, where $\mathcal{E}$ is a 
linear differential operator acting on a (spinorial/tensorial) field $f$. 
Suppose also that there exist a new variable of the form $\mathcal{T}(f)$, and linear differential operators $\mathcal{S}$ and 
$\mathcal{O}$ such that, {\it for all} $f$ (not only for solutions of $\mathcal{E}(f)=0$), the following equality holds:
\begin{equation}\label{seot}
 \mathcal{S}\mathcal{E}(f)=\mathcal{O}\mathcal{T}(f).
\end{equation}
Then if $f$ is a solution of $\mathcal{E}(f)=0$, $\Psi=\mathcal{T}(f)$ satisfies the equation $\mathcal{O}(\Psi)=0$.
Furthermore, given that (\ref{seot}) is valid for all $f$, we may introduce a hermitian inner product $\langle\cdot,\cdot\rangle$ and 
define the adjoint of an operator $A$ as $\langle f,A g\rangle=\langle A^{\dag}f,g\rangle$; and, since 
$(AB)^{\dagger}=B^{\dagger} A^{\dagger}$, we have the adjoint of equation (\ref{seot}):
\begin{equation}
 \mathcal{E}^{\dag}\mathcal{S}^{\dag}(\Phi)=\mathcal{T}^{\dag}\mathcal{O}^{\dag}(\Phi).
\end{equation}
This implies that a solution $\Phi$ of $\mathcal{O}^{\dag}(\Phi)=0$ generates a solution of $\mathcal{E}^{\dag}(\chi)=0$, where 
$\chi=\mathcal{S}^{\dag}(\Phi)$. 
Therefore, if the adjoint operators have a particularly useful form, we obtain in this way a mechanism for generating solutions of 
differential equations from solutions of other equations. In practice, the hermitian product we will use is given by
\begin{equation}\label{innerp}
 \langle f,g \rangle = \int_{\mathcal{M}}\bar{f}g,
\end{equation}
where a total contraction of all the indices of $f$ and $g$ is understood. 
We will further assume that all fields decay to zero at infinity, so that divergence terms will be neglected.

In the next sections, we apply this idea to spinor fields of spin $\frac{1}{2}$, 1 and 2.
The decoupled equations for Dirac, Maxwell and linearized gravitational fields on vacuum type D spacetimes with 
cosmological constant can be obtained from linear 
differential operators, acting on the corresponding spinor fields,
that have a generic form. 
More precisely, for a totally symmetric spinor $\phi_{A_1...A_{2\s}}=\phi_{(A_1...A_{2\s})}$, we will show that 
applying the operator given by
\begin{equation}\label{PQphi}
\Psi^{2\s/3}_2 P^{B_1A_2...A_{2\s}}_{(\s,s)}\c^{B'}_{B_1}\left(\Psi^{-2\s/3}_2\c^{A_1}_{B'}\phi_{A_1A_2....A_{2\s}}\right)
\end{equation}
where 
\begin{equation}\label{P}
 P^{A_1...A_{2\s}}_{(\s,s)}:=\tfrac{(2\s)!}{(\s-s)!(\s+s!)}
 \Psi^{(s-\s)/3}_2\iota^{(A_1}...\iota^{A_{\s-s}}o^{A_{\s-s+1}}...o^{A_{2\s})},
\end{equation}
and then linearizing around a type D $\l$-vacuum background, the result is a decoupled equation for the spin weight $s$ 
component of the field, with $s=0,\pm\s$. 
Since we are assuming a $\l$-vacuum solution with no background spin $\s=\tfrac{1}{2}$ or $\s=1$ fields, 
the linearization is actually only needed for spin $\s=2$, and we mention that in this case 
it should be understood in a `tensor sense'\footnote{this is because 
the linearization of a spinor is a rather delicate issue, see \cite{Backdahl}.}, that is to say, we linearize tensor quantities 
(we can do this because integer spin fields can be equivalently described by either spinor or tensor fields).
We note that, writing (\ref{PQphi}) in the form $P^{A_1...A_{2\s}}_{(\s,s)}(Q\phi)_{A_1...A_{2\s}}$, where 
\begin{equation}
 (Q\phi)_{A_1...A_{2\s}}:=\Psi^{2\s/3}_2\c^{B'}_{(A_1}\left(\Psi^{-2\s/3}_2\c^{B}_{|B'|}\phi_{A_2....A_{2\s})B}\right),
\end{equation}
the operator $Q$ coincides with the operator (2.13) recently presented in \cite{Aksteiner}\footnote{I thank T. B\"ackdahl for this 
observation.} (we also note that in this last reference, Wald's method of adjoint operators is also applied 
to construct higher order symmetry operators for the Teukolsky equations and the Teukolsky-Starobinsky identities in the cases
of spins 1 and 2).
The following sections are therefore mostly dedicated to prove that the linearization of
\begin{equation}
P^{A_1...A_{2\s}}_{(\s,s)}(Q\phi)_{A_1...A_{2\s}}=
\Psi^{2\s/3}_2 P^{B_1A_2...A_{2\s}}_{(\s,s)}\c^{B'}_{B_1}\left(\Psi^{-2\s/3}_2\c^{A_1}_{B'}\phi_{A_1A_2....A_{2\s}}\right) 
\end{equation}
leads to decoupled equations for (rescaled) components of the field $\phi_{A_1...A_{2\s}}$, 
i.e., to prove theorems \ref{spin1/2}, \ref{spin1} and \ref{spin2SV}.

The only cases we will not worry about in this work are $(\s=2,s=\pm1)$, which correspond to the linearized 
Weyl scalars $\dot{\Psi}_1$ and $\dot{\Psi}_3$; this is because they do not satisfy decoupled equations, 
as showed in \cite{Andersson1}.
On the other hand, we note that for spins $\s=1,2$ and spin-weight $s=0$, (\ref{P}) turns out to be a Killing spinor, which 
explains the appearance of this object on the field equations for $s=0$ in the Maxwell and linearized gravity systems. 

We finally mention that, in the next sections, the operator $\mathcal{O}$ of (\ref{seot}) will always have the form of the 
modified wave operator (\ref{boxp1}) (for some weight $p$) plus a (complex) potential $V$,
\begin{equation}
 \mathcal{O}=\teuk_{p}+V.
\end{equation}
Using that $\teuk_{p}=(\c^{\a}+p\Gamma^{\a})(\c_{\a}+p\Gamma_{\a})$, the adjoint $\mathcal{O}^{\dag}$ with respect to (\ref{innerp})
is easily calculated as
\begin{equation}\label{Oadj}
 \mathcal{O}^{\dag}=\bar{\teuk}_{-p}+\bar{V},
\end{equation}
an identity that will be extensively used in the next sections when calculating adjoint equations. 
This adjointness property is very important in the Teukolsky system, see \cite{Wald} and the recent article \cite{Aksteiner}.

\section{Dirac fields on type D spaces}\label{diracsec}

In this section we prove the theorem \ref{spin1/2} for spin $\s=\frac{1}{2}$,
which corresponds to massless Dirac fields. We recall that we use two-component (Weyl) spinors, 
corresponding to the $(\tfrac{1}{2},0)$ (or $(0,\frac{1}{2})$) irreducible representation of $SL(2,\mathbb{C})$. 
As is well-known, Dirac spinors, more commonly used in quantum field theory, transform under the (reducible) representation 
$(\tfrac{1}{2},0)\oplus(0,\tfrac{1}{2})$.\\
For notational convenience we define
\begin{equation}
 P^{B}_{(\frac{1}{2},s)}=:\stackbin[]{s}{P}{}^{B}, \;\;\;\; s=\pm\tfrac{1}{2}.
\end{equation}
Then, according to (\ref{P}), we have
\begin{equation}
 \stackbin[]{\frac{1}{2}}{P}{}^{B}= o^B, \;\;\;\;\;\; \stackbin[]{-\frac{1}{2}}{P}{}^{B}=\Psi^{-1/3}_2\iota^B.
\end{equation}
\begin{thm}[spin $\s=\frac{1}{2}$]\label{thmDirac}
 Consider a vacuum spacetime of Petrov type D with cosmological constant $\l$, and let $s=\pm\frac{1}{2}$. 
 Then for all spinor field $\chi_A$, the following equality holds:
 \begin{equation}\label{seotD}
  \mathcal{S}_{D,s}\mathcal{E}_{D}(\chi_A)=\mathcal{O}_{D,s}\mathcal{T}_{D,s}(\chi_A),
 \end{equation}
 where the linear differential operators are
 \begin{eqnarray}
  \mathcal{S}_{D,s}(J_{B'}) &:=&\Psi^{1/3}_2\stackbin[]{s}{P}{}^{B}\c^{B'}_{B}[\Psi^{-1/3}_2 J_{B'}], \\
  \mathcal{E}_{D}(\chi_A) &:=&\c^{A}_{B'}\chi_{A}, \\
  \mathcal{O}_{D,s}(\Phi) &:=&(\teuk_{2s}-\Psi_2 +\tfrac{2}{3}\l)\Phi, \\
  \mathcal{T}_{D,s}(\chi_A) &:=&-\tfrac{1}{2}\stackbin[]{s}{P}{}^{A}\chi_{A}.
 \end{eqnarray}
\end{thm}

\begin{proof}
Consider first the spin weight $s=\tfrac{1}{2}$ case. Using the expression (\ref{AD}) for the 1-form $A_{\a}$, we have
\begin{equation}\label{seD}
 \mathcal{S}_{D,\frac{1}{2}}\mathcal{E}_{D}(\chi_A)=o^B\c^{B'}_{B}\c^{A}_{B'}\chi_A+o^BA^{B'}_{B}\c^{A}_{B'}\chi_A.
\end{equation}
For the term with second derivatives of $\chi_A$, we use Leibniz rule:
\begin{equation}\label{leibD}
 o^B\c^{B'}_{B}\c^{A}_{B'}\chi_A=\c^{B'}_{B}(o^B\c^{A}_{B'}\chi_A)-(\c^{B'}_{B}o^B)(\c^{A}_{B'}\chi_A).
\end{equation}
The first term on the RHS of this equation gives:
\begin{eqnarray*}
 \c^{B'}_{B}(o^B\c^{A}_{B'}\chi_A)&=&-\c_{CC'}(o^C\bar\e^{C'B'}\c^{A}_{B'}\chi_A)\\
 &=&-\c_{CC'}(o^C\bar{o}^{C'}\bar\iota^{B'}\c^{A}_{B'}\chi_A)+\c_{CC'}(o^C\bar\iota^{C'}\bar{o}^{B'}\c^{A}_{B'}\chi_A)\\
 &=&-(D+(\c_{\a}l^{\a}))[\bar\iota^{B'}\c^{A}_{B'}\chi_A]+(\d+(\c_{\a}m^{\a}))[\bar{o}^{B'}\c^{A}_{B'}\chi_A]\\
 &=&-(D+\e+\bar\e-\rho-\bar\rho)[\bar\iota^{B'}\c^{A}_{B'}\chi_A]\\
 & & +(\d+\b+\bar\b'-\tau-\bar\tau')[\bar{o}^{B'}\c^{A}_{B'}\chi_A]
\end{eqnarray*}
where we have used (\ref{tetrad}) for the relation between the dyad and the tetrad vectors, and also expressions (\ref{dl}) and 
(\ref{dm}) for the divergence of the tetrad. 
For the second term in (\ref{leibD}) we use equation (\ref{edo}) for the derivative of $o^B$, then
\begin{equation}
 -(\c^{B'}_{B}o^B)(\c^{A}_{B'}\chi_A)=-(\b-\tau)\bar{o}^{B'}\c^{A}_{B'}\chi_A+(\e-\rho)\bar\iota^{B'}\c^{A}_{B'}\chi_A.
\end{equation}
For the term with first derivatives of $\chi_{A}$ in (\ref{seD}), we use the expression (\ref{A}) of $A_{\a}$, 
which implies 
\begin{equation}
o^{B}A^{B'}_{B}=\rho\bar\iota^{B'}-\tau\bar{o}^{B'}.
\end{equation}
Combining with the previous calculations, using (\ref{thorn})-(\ref{eth}) for the definition of the weighted 
derivatives $\text{\th}$ and $\text{\dh}$, and (\ref{diraceq}) for the corresponding components of $\c^{A}_{B'}\chi_A$, 
we have
\begin{eqnarray}
\nonumber \mathcal{S}_{D,\frac{1}{2}}\mathcal{E}_{D}(\chi_A)&=&-(D+\bar\e-\rho-\bar\rho)[\bar\iota^{B'}\c^{A}_{B'}\chi_A]
 +(\d+\bar\b'-\tau-\bar\tau')[\bar{o}^{B'}\c^{A}_{B'}\chi_A]\\
\nonumber &=&-(\text{\th}-\rho-\bar\rho)[\bar\iota^{B'}\c^{A}_{B'}\chi_A]+(\text{\dh}-\tau-\bar\tau')[\bar{o}^{B'}\c^{A}_{B'}\chi_A]\\
\nonumber &=&-(\text{\th}-\rho-\bar\rho)[(\text{\th}'-\rho')\chi_0-(\text{\dh}-\tau)\chi_1]\\
\nonumber & & + (\text{\dh}-\tau-\bar\tau')[(\text{\th}-\rho)\chi_1-(\text{\dh}'-\tau')\chi_0]\\
\nonumber &=&-[(\text{\th}-\rho-\bar\rho)(\text{\th}'-\rho')-(\text{\dh}-\tau-\bar\tau')(\text{\dh}'-\tau')]\chi_0\\
 & & +[(\text{\th}-\rho-\bar\rho)(\text{\dh}-\tau)-(\text{\dh}-\tau-\bar\tau')(\text{\th}-\rho)]\chi_1.
\end{eqnarray}
Using the explicit expression for the weighted wave operator (\ref{boxp2}) with $p=1$, we see that the term with $\chi_0$ 
in the previous equation is just
\begin{equation}
 -[(\text{\th}-\rho-\bar\rho)(\text{\th}'-\rho')-(\text{\dh}-\tau-\bar\tau')(\text{\dh}'-\tau')]\chi_0
 =-\tfrac{1}{2}(\teuk_{1}-\Psi_2+\tfrac{2}{3}\l)\chi_0.
\end{equation}
For the term with $\chi_1$, we use the commutation relation (\ref{comm4}) with $a=1$:
\begin{multline}\
 [(\text{\th}-\rho-\bar\rho)(\text{\dh}-\tau)-(\text{\dh}-\tau-\bar\tau')(\text{\th}-\rho)]\chi_1\\
 = [\text{\th}-\rho,\text{\dh}-\tau]\chi_1-\bar{\rho}(\text{\dh}-\tau)\chi_1+\bar\tau'(\text{\th}-\rho)\chi_1\equiv0, \label{commD}
\end{multline}
and therefore we finally obtain
\begin{equation}
 \mathcal{S}_{D,\frac{1}{2}}\mathcal{E}_{D}(\chi_A)=-\tfrac{1}{2}(\teuk_{1}-\Psi_2+\tfrac{2}{3}\l)\chi_0.
\end{equation}

For the spin weight $s=-\tfrac{1}{2}$ case, we just have to apply the prime operation and use formula (\ref{boxt'}) 
for the transformation law of $\teuk_{1}$:
\begin{equation}
 \mathcal{S}_{D,-\frac{1}{2}}\mathcal{E}_{D}(\chi_A)=-\tfrac{1}{2}(\teuk_{-1}-\Psi_2+\tfrac{2}{3}\l)[\Psi^{-1/3}_2\chi_1].
\end{equation}

\end{proof}

In order to generate Dirac fields from solutions of the decoupled equations, 
we take the adjoint equation to (\ref{seotD}) in the manner described in subsection \ref{adjop-sec}
(in particular we use (\ref{Oadj})), this gives
\begin{equation}\label{adjD}
 -\c^{B'}_{B}[\mathcal{S}^{\dag}_{D,s}(\Phi)]^{B}=
 -\tfrac{1}{2}\stackbin[]{s}{\bar{P}}{}^{B'}(\bar{\teuk}_{-2s}-\bar\Psi_2+\tfrac{2}{3}\l)\Phi, 
\end{equation}
where 
\begin{equation}
 [\mathcal{S}^{\dag}_{D,s}(\Phi)]^{B}=\bar\Psi^{-1/3}_2\c^{BB'}[\bar\Psi^{1/3}_2\stackbin[]{s}{\bar{P}}_{B'}\Phi].
\end{equation}
Equation (\ref{adjD}) implies then the following corollary:
\begin{cor}
  Let $\Phi$ be a solution to the decoupled equation $\bar{\mathcal{O}}_{D,-s}(\Phi)=0$, which is the spin weight 
 $\mp\tfrac{1}{2}$ Teukolsky equation for $s=\pm\tfrac{1}{2}$, in a $\l$-vacuum type D spacetime. Then:
 \begin{enumerate}
 \item 
  The spinor field
  \begin{equation}
   \stackbin[]{s}{\phi}_{A}(\Phi)=\bar\Psi^{-1/3}_2\c^{B'}_{A}[\bar\Psi^{1/3}_2\stackbin[]{s}{\bar{P}}_{B'}\Phi]
  \end{equation}
  is a solution to the massless Dirac equation, $\c^{AA'}\stackbin[]{s}{\phi}_{A}=0$.
  \item The operator $\mathcal{A}_{D,s}$ defined by
  \begin{equation}
   \mathcal{A}_{D,s}(\Phi)=\stackbin[]{s}{P}{}^{A}\stackbin[]{s}{\phi}_{A}(\Phi)
  \end{equation}
  maps solutions of $\mathcal{O}_{D,s}(\Phi)=0$ into solutions of $\bar{\mathcal{O}}_{D,-s}(\Phi)=0$.
 \end{enumerate}
\end{cor}

For further symmetry operators for the massless Dirac equation, we refer to \cite{Andersson2}
(and references therein).

\section{Maxwell fields on type D spaces}\label{maxwellsec}

We now prove the theorem of spin $\s=1$, corresponding to Maxwell fields. 
The proof is very similar to the previous case, in the sense that the manipulations for extreme spin weight are the same.
For spin weight zero, the proof can be done either by the same lines or by using the fact that the corresponding object 
is a Killing spinor.\\
Once more, for notational convenience we define
\begin{equation}
 P^{AB}_{(1,s)}=: \stackbin[]{s}{P}{}^{AB}, \;\;\;\; s=0,\pm1.
\end{equation}
Explicitly, we have
\begin{eqnarray}
 \stackbin[]{1}{P}{}^{AB}&=&o^{A}o^{B} =: o^{AB}, \\
 \stackbin[]{0}{P}{}^{AB}&=&2\Psi^{-1/3}_2o^{(A}\iota^{B)}, \label{WM0} \\
 \stackbin[]{-1}{P}{}^{AB}&=&\Psi^{-2/3}_2\iota^{A}\iota^{B} =: \Psi^{-2/3}_2\iota^{AB}.
\end{eqnarray}
Note that $\stackbin[]{0}{P}{}^{AB}$ coincides with the Killing spinor (\ref{KSD}) (with $k=2$).\\
We recall that for spin weight $s=\pm1$, theorem \ref{spin1} should give us the $s=\pm1$ Teukolsky equations for 
electromagnetic perturbations, while for $s=0$ we should obtain the Fackerell-Ipser equation. This is summarized as follows:
\begin{thm}[spin $\s=1$]\label{thm-maxwell}
 Consider a vacuum spacetime of Petrov type D with cosmological constant $\l$, and let $s=0,\pm1$. 
 Then for all symmetric spinor field $\phi_{AB}=\phi_{(AB)}$,  the following equality holds:
 \begin{equation}\label{seotM}
  \mathcal{S}_{M,s}\mathcal{E}_{M}(\phi_{AB})=\mathcal{O}_{M,s}\mathcal{T}_{M,s}(\phi_{AB}),
 \end{equation}
 where the linear differential operators are
 \begin{eqnarray}
  \mathcal{S}_{M,s}(J_{B'B}) &:=&\Psi^{2/3}_2\stackbin[]{s}{P}{}^{AB}\c^{B'}_{A}[\Psi^{-2/3}_2 J_{B'B}], \\
  \mathcal{E}_{M}(\phi_{AB}) &:=&\c^{A}_{B'}\phi_{AB}, \\
  \mathcal{O}_{M,s}(\Phi) &:=&\left(\teuk_{2s}+2(1-3s^2)\Psi_2 +\tfrac{2}{3}\l\right)\Phi, \\
  \mathcal{T}_{M,s}(\phi_{AB}) &:=&-\tfrac{1}{2}\stackbin[]{s}{P}{}^{AB}\phi_{AB}.
 \end{eqnarray}
\end{thm}

\begin{proof}

We start with the spin weight $s=1$ case:
\begin{eqnarray}
\nonumber \mathcal{S}_{M,1}\mathcal{E}_{M}(\phi_{AB})&=&(o^{BC}\c^{B'}_{C}+2o^{BC}A^{B'}_{C})\c^{A}_{B'}\phi_{AB}\\
 &=&o^{BC}\c^{B'}_{C}\c^{A}_{B'}\phi_{AB}+2o^{BC}A^{B'}_{C}\c^{A}_{B'}\phi_{AB}, \label{seM1}
\end{eqnarray}
where we have used the expression (\ref{AD}) for the 1-form $A_{AA'}$.
Leibniz rule for the term with second derivatives of $\phi_{AB}$ gives:
\begin{equation}\label{leibM}
 o^{BC}\c^{B'}_{C}\c^{A}_{B'}\phi_{AB}=\c^{B'}_{C}(o^{BC}\c^{A}_{B'}\phi_{AB})-(\c^{B'}_{C}o^{BC})(\c^{A}_{B'}\phi_{AB}).
\end{equation}
The first and second terms in the right hand side of this equation are treated in a similar way as was done 
for the Dirac case in equation (\ref{leibD}); using expressions for derivatives of the dyad spinors and 
tetrad vectors given in appendix \ref{appB}, we get
\begin{equation}
 o^{BC}\c^{B'}_{C}\c^{A}_{B'}\phi_{AB}=-(\text{\th}-\bar\rho)[o^{B}\bar{\iota}^{B'}\c^{A}_{B'}\phi_{AB}]
+(\text{\dh}-\bar\tau')[o^{B}\bar{o}^{B'}\c^{A}_{B'}\phi_{AB}],
\end{equation}
where we have also used the definition of the operators (\ref{thorn})-(\ref{eth}) acting on the corresponding weighted quantities.
On the other hand, using the expression (\ref{A}) for $A_{\a}$, the second term in (\ref{seM1}) is
\begin{equation}
 2o^{BC}A^{B'}_{C}\c^{A}_{B'}\phi_{AB}=2\rho o^{B}\bar{\iota}^{B'}\c^{A}_{B'}\phi_{AB}-2\tau o^{B}\bar{o}^{B'}\c^{A}_{B'}\phi_{AB}.
\end{equation}
Thus:
\begin{equation}
 \mathcal{S}_{M,1}\mathcal{E}_{M}(\phi_{AB})=-(\text{\th}-2\rho-\bar\rho)[o^{B}\bar{\iota}^{B'}\c^{A}_{B'}\phi_{AB}]
+(\text{\dh}-2\tau-\bar\tau')[o^{B}\bar{o}^{B'}\c^{A}_{B'}\phi_{AB}].
\end{equation}
Now we use the formula (\ref{maxwelleq}) for expressing $o^{B}\bar{\iota}^{B'}\c^{A}_{B'}\phi_{AB}$ and 
$o^{B}\bar{o}^{B'}\c^{A}_{B'}\phi_{AB}$ in GHP form; the result, after reordering terms in $\phi_0,\phi_1$ and $\phi_2$, is:
\begin{eqnarray*}
 \mathcal{S}_{M,1}\mathcal{E}_{M}(\phi_{AB})
 &=&-[(\text{\th}-2\rho-\bar\rho)(\text{\th}'-\rho')-(\text{\dh}-2\tau-\bar\tau')(\text{\dh}'-\tau')]\phi_0\\
 & & +[(\text{\th}-2\rho-\bar\rho)(\text{\dh}-2\tau)-(\text{\dh}-2\tau-\bar\tau')(\text{\th}-2\rho)]\phi_1\\
 & &+(\text{\th}-2\rho-\bar\rho)[\sigma\phi_2]-(\text{\dh}-2\tau-\bar\tau')[\kappa\phi_2].
\end{eqnarray*}
For the term with $\phi_0$, using (\ref{boxp2}) we see that
\begin{multline}
 -[(\text{\th}-2\rho-\bar\rho)(\text{\th}'-\rho')-(\text{\dh}-2\tau-\bar\tau')(\text{\dh}'-\tau')]\phi_0\\
 =-\tfrac{1}{2}(\teuk_{2}-4\Psi_2+\tfrac{2}{3}\l)\phi_0+2(\kappa\kappa'-\sigma\sigma')\phi_0.
\end{multline}
The term with $\phi_1$ identically vanishes because of (\ref{comm4}) with $a=2$, similarly as in (\ref{commD}).
Finally, using (\ref{ghpD}) for a type D background, we get:
\begin{equation}
 \mathcal{S}_{M,1}\mathcal{E}_{M}(\phi_{AB})=-\tfrac{1}{2}(\teuk_{2}-4\Psi_2+\tfrac{2}{3}\l)\phi_0.
\end{equation}

This completes the proof of spin weight $s=1$.
For $s=-1$, as with the Dirac case, the corresponding identity follows by applying a prime to the previous equation 
and using (\ref{boxt'}) with $p=2$.

Consider now the spin weight $s=0$ case. We will use the fact that $\stackbin[]{0}{P}{}^{AB}$ 
in (\ref{WM0}) coincides with the Killing spinor (\ref{KSD}), 
\begin{equation}
 \stackbin[]{0}{P}_{AB}\equiv K_{AB}.
\end{equation}
We have:
\begin{eqnarray*}
 \mathcal{S}_{M,0}\mathcal{E}_{M}(\phi_{AB})&=&\Psi^{2/3}_2K^{BC}\c^{B'}_{C}[\Psi^{-2/3}_2 \c^{A}_{B'}\phi_{AB}]\\
 &=&K^{BC}\c^{B'}_{C}\c^{A}_{B'}\phi_{AB}+2K^{BC}A^{B'}_{C}\c^{A}_{B'}\phi_{AB}\\
 &=&-\tfrac{1}{2}K^{AB}\Box\phi_{AB}+2K^{BC}\Box_{CA}\phi^{A}{}_{B}+2K^{BC}A^{B'}_{C}\c^{A}_{B'}\phi_{AB}.
\end{eqnarray*}
Using the explicit action of the curvature operator $\Box_{CA}$, we get
\begin{eqnarray*}
 \mathcal{S}_{M,0}\mathcal{E}_{M}(\phi_{AB})&=&-\tfrac{1}{2}K^{AB}\Box\phi_{AB}
 +K^{BC}[X_{CAD}{}^{A}\phi^{D}{}_{B}-X_{CAB}{}^{D}\phi^{A}{}_{D}]\\
 & & +2K^{BC}A^{B'}_{C}\c^{A}_{B'}\phi_{AB}\\
 &=&-\tfrac{1}{2}K^{AB}\Box\phi_{AB}+K^{AB}[-\tfrac{R}{6}\phi_{AB}+\psi_{ABCD}\phi^{CD}]\\
 & & +2K^{BC}A^{B'}_{C}\c^{A}_{B'}\phi_{AB},
\end{eqnarray*}
where we have used the identity (\ref{X2}), together with the decomposition (\ref{X1}) of the curvature 
spinor $X_{ABCD}$. Now, the identities (\ref{k1}) and (\ref{k3}) for the Killing spinor allow us to write
\begin{eqnarray*}
 -\tfrac{1}{2}K^{AB}\Box\phi_{AB}&=&-\tfrac{1}{2}\Box(K^{AB}\phi_{AB})+\tfrac{1}{2}\phi_{AB}\Box K^{AB}
 +\c^{C'C}K^{AB}\c_{C'C}\phi_{AB}\\
 &=&-\tfrac{1}{2}(\Box-2\Psi_2-\tfrac{R}{6})K^{AB}\phi_{AB}+\tfrac{2}{3}\c^{C'D}K_{D}{}^{A}\c^{B}_{C'}\phi_{AB}.
\end{eqnarray*}
Furthermore, using (\ref{k4}) and the definition of $A_{AA'}$ it is easy to see that $K^{BC}A^{B'}_{C}=\tfrac{1}{3}\c^{B'}_{C}K^{BC}$; 
then combining with (\ref{k2}) we finally have
\begin{eqnarray*}
 \mathcal{S}_{M,0}\mathcal{E}_{M}(\phi_{AB})&=&-\tfrac{1}{2}(\Box-2\Psi_2-\tfrac{R}{6})K^{AB}\phi_{AB}
 +\tfrac{2}{3}\c^{C'D}K_{D}{}^{A}\c^{B}_{C'}\phi_{AB}\\
 & & -(2\Psi_2+\tfrac{R}{6})K^{AB}\phi_{AB}+\tfrac{2}{3}\c^{B'}_{C}K^{BC}\c^{A}_{B'}\phi_{AB}\\
 &=&-\tfrac{1}{2}(\Box+2\Psi_2+\tfrac{R}{6})K^{AB}\phi_{AB}.
\end{eqnarray*}
Finally, replacing $R=4\l$ we obtain the desired formula.
\end{proof}

Now we want to see how to generate Maxwell fields from solutions of the decoupled equations.
If we take the adjoint equation to (\ref{seotM}), we get
\begin{equation}
 -\c^{A(A'}[\mathcal{S}^{\dag}_{M,s}(\Phi)]^{B')}_{A}
 =-\tfrac{1}{2}\stackbin[]{s}{\bar{P}}{}^{A'B'}\left(\bar{\teuk}_{-2s}+2(1-3s^2)\bar\Psi_2+\tfrac{2}{3}\l\right)\Phi, 
\end{equation}
where 
\begin{equation}
 [\mathcal{S}^{\dag}_{M,s}(\Phi)]^{BB'}=-\bar\Psi^{-2/3}_2\c^{B}_{A'}[\bar\Psi^{2/3}_{2}\stackbin[]{s}{\bar{P}}{}^{A'B'}\Phi].
\end{equation}
This implies that if $\Phi$ is a solution to $\bar{\mathcal{O}}_{M,-s}(\Phi)=0$, then 
\begin{equation}
 \c^{A(A'}[\mathcal{S}^{\dag}_{M,s}(\Phi)]^{B')}_{A}=0.
\end{equation}
Evidently, these are not Maxwell equations. In order to construct a Maxwell field, we need the following lemma:
\begin{lem}\label{lemma-pot-max}
Let $\a^{A'}_A$ be a solution of $\c^{A(B'}\a^{A')}_{A}=0$ on an arbitrary spacetime. Then $\phi_{AB}:=\c_{(A|B'|}\a^{B'}_{B)}$
is a Maxwell field, $\c^{AA'}\phi_{AB}=0$.
\end{lem}
\begin{proof}
If $\alpha^{A'}_A$ satisfies $\c^{A(A'}\alpha^{B')}_{A}=0$, then taking an additional derivative it is easy to see that
\begin{equation}\label{nab2}
 0=-\tfrac{1}{2}\Box\alpha^{A'}_{B}+\Phi_{BAQ'}{}^{A'}\alpha^{Q'A}-\tfrac{R}{8}\alpha^{A'}_{B}+\c_{BB'}\c^{A'}_{A}\alpha^{B'A}.
\end{equation}
On the other hand, if $\phi_{AB}=\c_{(A|A'|}\alpha^{A'}_{B)}$, then
\begin{equation*}
 2\c^{AA'}\phi_{AB}=\tfrac{1}{2}\Box\alpha^{A'}_{B}-\Phi_{BAQ'}{}^{A'}\alpha^{Q'A}+\tfrac{R}{8}\alpha^{A'}_{B}-\c^{A'}_{A}\c_{BB'}\alpha^{B'A}.
\end{equation*}
Note that the only difference between this equation and (\ref{nab2}) is, besides a global sign, the order of the derivatives in the 
last term on the right hand side. Using (\ref{boxAB}), we have
\begin{eqnarray*}
\c^{A'}_{A}\c_{BB'}\alpha^{B'A}&=&\bar\epsilon^{A'C'}(\c_{BB'}\c_{AC'}+\bar\epsilon_{C'B'}\Box_{AB}+\epsilon_{AB}\Box_{C'B'})\alpha^{B'A}\\
&=& \bar\epsilon^{A'C'}\left[\c_{BB'}\c_{AC'}\alpha^{B'A}+\bar\epsilon_{C'B'}\left(\Phi_{ABQ'}{}^{B'}\alpha^{Q'A}
+X_{ABQ}{}^{A}\alpha^{B'Q}\right) \right.\\
& & \left.+\epsilon_{AB}\left(\bar{X}_{C'B'Q'}{}^{B'}\alpha^{Q'A}+\Phi_{C'B'Q}{}^{A}\alpha^{B'Q}\right)\right]\\
&=& \c_{BB'}\c^{A'}_{A}\alpha^{B'A},
\end{eqnarray*}
where the identity (\ref{X2}) and its complex conjugate were also used in the intermediate steps. It follows that
\begin{equation*}
 \c^{AA'}\phi_{AB}=0.
\end{equation*}
\end{proof}

Combining theorem (\ref{thm-maxwell}) with the results of the previous lemma, we have the following corollary:
\begin{cor}[Spinor version]\label{gen-max1}
 Consider a vacuum type D spacetime with cosmological constant $\lambda$.  
 Let $\Phi$ be a solution of the decoupled equation $\bar{\mathcal{O}}_{M,-s}(\Phi)=0$, which is the
 spin weight $\mp1$ Teukolsky equation for $s=\pm1$, and the Fackerell-Ipser equation for $s=0$. Then:
 \begin{enumerate}
  \item The spinor field 
 \begin{equation}\label{maxfield}
  \stackbin[]{s}{\phi}_{AB}(\Phi) =
  -2\c_{B'(A}\left[\bar\Psi^{-2/3}_2\c_{B)C'}[\bar\Psi^{2/3}_{2}\stackbin[]{s}{\bar{P}}{}^{B'C'}\Phi]\right]
 \end{equation}
 is a solution to Maxwell equations, $\c^{AA'}\phi_{AB}=0$.
 \item The operator $\mathcal{A}_{M,s}$ defined by 
 \begin{equation}
  \mathcal{A}_{M,s}(\Phi)=\stackbin[]{s}{P}{}^{AB}\stackbin[]{s}{\phi}_{AB}(\Phi)
 \end{equation}
 maps solutions of $\mathcal{O}_{M,s}(\Phi)=0$ into solutions of $\bar{\mathcal{O}}_{M,-s}(\Phi)=0$.
 \end{enumerate}
\end{cor}

We refer once more to \cite{Andersson2} for further symmetry operators for Maxwell equations. 
We also note the recent work \cite{Aksteiner} in which symmetry operators for spin 1 and the 
connection with Teukolsky systems and Debye potentials are studied.

\subsection{Tensor expressions}

We now put in tensor form the spinor expressions for the Maxwell field. First, we need to introduce 
the anti-self-dual 2-form
\begin{equation}\label{Pab}
 \stackbin[]{s}{P}_{\a\b}:=\stackbin[]{s}{P}{}_{AB}\bar\e_{A'B'}, \;\;\;\; s=0,\pm1,
\end{equation}
or, explicitly,
\begin{eqnarray}
 \stackbin[]{+1}{P}_{\a\b}&:=&2l_{[\a}m_{\b]}, \\
 \stackbin[]{0}{P}_{\a\b}&:=&2\Psi^{-1/3}_2(l_{[\a}n_{\b]}+\bar{m}_{[\a}m_{\b]}), \label{w0max}\\
 \stackbin[]{-1}{P}_{\a\b}&:=&2\Psi^{-2/3}_2\bar{m}_{[\a}n_{\b]}.
\end{eqnarray}
Note that (\ref{w0max}) is the tensor version of the Killing spinor (\ref{KSD}), therefore, it is the sum of a Killing-Yano 
tensor and its dual,
\begin{equation}
 \stackbin[]{0}{P}_{\a\b}=-\tfrac{i}{2}(Y_{\a\b}+i{}^{*}Y_{\a\b}).
\end{equation}
The tensor version of corollary \ref{gen-max1} is the following:
\begin{corbis}{gen-max1}[Tensor version]\label{gen-max2}
 Consider a vacuum type D spacetime with cosmological constant $\l$.  
 Let $\Phi$ be a solution to $\bar{\mathcal{O}}_{M,-s}(\Phi)=0$, which is the spin weight $\mp1$ Teukolsky equation for $s=\pm1$, 
 and the Fackerell-Ipser equation for $s=0$. Then: 
 \begin{enumerate}
  \item The tensor field
  \begin{equation}
  \stackbin[]{s}{\wt{F}}_{\a\b}(\Phi)=\stackbin[]{s}{E}_{\a\b}(\Phi)-i\stackbin[]{s}{{}^{*}E}_{\a\b}(\Phi),
 \end{equation}
 where
 \begin{equation}
  \stackbin[]{s}{E}_{\a\b}(\Phi)=-2\c_{[\a}[\Psi^{-2/3}_{2}\c^{\g}(\stackbin[]{s}{P}_{\b]\g}\Psi^{2/3}_2 \Phi)].
 \end{equation}
 is a (complex) solution to Maxwell equations, $\c^{\a}\stackbin[]{s}{\wt{F}}_{\a\b}=0$.
 \item The operator defined by
 \begin{equation}
  \mathcal{A}_{M,s}(\Phi)=\tfrac{1}{2}\stackbin[]{s}{P}{}^{\a\b}\stackbin[]{s}{\wt{F}}_{\a\b}(\Phi)
 \end{equation}
 maps solutions of $\mathcal{O}_{M,s}(\Phi)=0$ into solutions of $\bar{\mathcal{O}}_{M,-s}(\Phi)=0$.
 \end{enumerate}
\end{corbis}

\begin{proof}
We need only translate the spinor expressions into tensor form.
It is easy to see that
\begin{eqnarray}
\nonumber -2\c_{[\a}[\bar\Psi^{-2/3}_2\c^{\g}(\stackbin[]{s}{\bar{P}}_{\b]\g}\bar\Psi^{2/3}_2\Phi)]
 &=&\bar\e_{A'B'}\c_{D'(A}[\bar\Psi^{-2/3}_2\c^{C'}_{B)}(\stackbin[]{s}{\bar{P}}_{C'}{}^{D'}\bar\Psi^{2/3}_2\Phi)]\\
\nonumber & &+\e_{AB}\c_{D(A'}[\bar\Psi^{-2/3}_2\c^{C'D}(\stackbin[]{s}{\bar{P}}_{B')C'}\bar\Psi^{2/3}_2\Phi)]\\
\end{eqnarray}
The dual to this 2-form is (see e.g. \cite[Eq.(3.4.22)]{Penrose1})
\begin{eqnarray}
\nonumber -\e_{\a\b}{}^{\g\d}\c_{\g}[\bar\Psi^{-2/3}_2\c^{\e}(\stackbin[]{s}{\bar{P}}_{\d\e}\bar\Psi^{2/3}_2\Phi)]
 &=&-i\bar\e_{A'B'}\c_{D'(A}[\bar\Psi^{-2/3}_2\c^{C'}_{B)}(\stackbin[]{s}{\bar{P}}_{C'}{}^{D'}\bar\Psi^{2/3}_2\Phi)]\\
\nonumber & &+i\e_{AB}\c_{D(A'}[\bar\Psi^{-2/3}_2\c^{C'D}(\stackbin[]{s}{\bar{P}}_{B')C'}\bar\Psi^{2/3}_2\Phi)]\\
\end{eqnarray}
Recalling the expression (\ref{maxfield}) for $\stackbin[]{s}{\phi}_{AB}$ we get:
\begin{equation}
 \stackbin[]{s}{\phi}_{AB}\bar\e_{A'B'}
 =-2\c_{[\a}[\bar\Psi^{-2/3}_2\c^{\g}(\stackbin[]{s}{\bar{P}}_{\b]\g}\bar\Psi^{2/3}_2\Phi)]
 +i\e_{\a\b}{}^{\g\d}\c_{\g}[\bar\Psi^{-2/3}_2\c^{\e}(\stackbin[]{s}{\bar{P}}_{\d\e}\bar\Psi^{2/3}_2\Phi)],
\end{equation}
which implies that
\begin{equation}
 \stackbin[]{s}{\wt{F}}_{\a\b}(\Phi) = \stackbin[]{s}{\phi}_{AB}\bar\e_{A'B'}
 =\stackbin[]{s}{E}_{\a\b}(\Phi)-i\stackbin[]{s}{{}^{*}E}_{\a\b}(\Phi),
\end{equation}
where 
\begin{equation}
 \stackbin[]{s}{E}_{\a\b}(\Phi)=-2\c_{[\a}[\bar\Psi^{-2/3}_2\c^{\g}(\stackbin[]{s}{\bar{P}}_{\b]\g}\bar\Psi^{2/3}_2\Phi)].
\end{equation}
The proof of item 2. is immediate from corollary \ref{gen-max1} and equation (\ref{Pab}).
\end{proof}

\subsubsection{Spin weight zero}\label{sw0Msection}

We now consider in more detail the spin weight $s=0$ case of (\ref{maxfield}), in order to understand the 
role that Killing spinors and Killing-Yano tensors have in the description of the Maxwell field.
In this subsection we assume that the Killing vector $\xi^{AA'}=\c^{BA'}K_{B}{}^{A}$ is real.
First, we need to put (\ref{maxfield}) (for $s=0$) in terms of the Killing-Yano tensor:
\begin{lem}
 The spinor field given by (\ref{maxfield}) with $s=0$ can be rewritten as
 \begin{equation}\label{max2}
  \stackbin[]{0}{\phi}_{AB}(\Phi) = 2i\c^{B'}_{A}[Y_{BB'CC'}\c^{CC'}\Phi]+K_{AB}(\Box+2\Psi_2+\tfrac{2}{3}\l)\Phi,
 \end{equation}
 where $Y_{\a\b}$ is the Killing-Yano tensor (\ref{KY}).
\end{lem}

\begin{proof}
We have
\begin{eqnarray*}
 \stackbin[]{0}{\phi}_{AB}(\Phi)&=&-2\c_{B'(A}\left[\bar\Psi^{-2/3}_2\c_{B)C'}[\bar\Psi^{2/3}_{2}\bar{K}^{B'C'}\Phi]\right]\\
 &=&4\c_{B'(A}[\bar{A}_{B)C'}\bar{K}^{B'C'}\Phi]-2\c_{B'(A}\c_{B)C'}(\bar{K}^{B'C'}\Phi).
\end{eqnarray*}
Using (\ref{k4}) and the definition of the (real, Killing) vector $\xi^{AA'}$, we get 
$\bar{A}_{BC'}\bar{K}^{B'C'}=-\frac{1}{3}\xi^{B'}_{B}$, and then
\begin{eqnarray}
\nonumber  \stackbin[]{0}{\phi}_{AB}(\Phi)&=&-\tfrac{4}{3}\c_{B'(A}[\xi^{B'}_{B)}\Phi]-2\c_{B'(A}\left[(\c_{B)C'}\bar{K}^{B'C'})\Phi
 +\bar{K}^{B'C'}\c_{B)C'}\Phi\right]\\
\nonumber  &=&\tfrac{8}{3}\xi^{A'}_{(A}\c_{B)A'}\Phi+\tfrac{2}{3}(\c_{B'(A}\xi^{B'}_{B)})\Phi-2\bar{K}^{A'B'}\c_{AA'}\c_{BB'}\Phi\\
 &=&\tfrac{8}{3}\xi^{A'}_{(A}\c_{B)A'}\Phi+(2\Psi_2+\tfrac{R}{6})K_{AB}\Phi-2\bar{K}^{A'B'}\c_{AA'}\c_{BB'}\Phi \label{phi0-1}
\end{eqnarray}
where we also used (\ref{k5}) for the divergence of $\xi^{AA'}$. On the other hand,
\begin{eqnarray}
\nonumber 2i\c^{B'}_{(A}[Y_{B)B'CC'}\c^{CC'}\Phi] &=& 2i[\c^{B'}_{(A}Y_{B)B'CC'}]\c^{CC'}\Phi+2Y_{B'(B}{}^{CC'}\c^{B'}_{A)}\c_{CC'}\Phi\\
\nonumber &=& -2\c^{B'}_{(A}K_{B)C}\bar{\epsilon}_{B'C'}\c^{CC'}\Phi+2\c^{B'}_{(A}\bar{K}_{|B'C'|}\epsilon_{B)C}\c^{CC'}\Phi\\
\nonumber & & -2K_{BC}\bar{\epsilon}_{B'C'}\c^{B'}_{A}\c^{CC'}\Phi+2\bar{K}_{B'C'}\epsilon_{BC}\c^{B'}_{A}\c^{CC'}\Phi\\
\nonumber &=& +\tfrac{2}{3}\xi_{C'(B}\epsilon_{A)C}\c^{CC'}\Phi - 2\xi_{C'(A}\c^{C'}_{B)}\Phi -K_{AB}\Box \Phi \\
\nonumber & & - 2\bar{K}^{A'B'}\c_{AA'}\c_{BB'}\Phi\\
\nonumber  &=& \tfrac{8}{3}\xi^{C'}_{(A}\c_{B)C'}\Phi-K_{AB}\Box \Phi-2\bar{K}^{A'B'}\c_{AA'}\c_{BB'}\Phi.\\ 
\label{phi0-2}
\end{eqnarray}
Combining (\ref{phi0-1}) and (\ref{phi0-2}), (\ref{max2}) follows immediately. 
\end{proof}

Now we give the tensor form of (\ref{max2}). It is convenient to separate $\Phi$ into its real and imaginary parts 
in the form $\Phi \equiv u+iv$, with $u$ and $v$ real scalar fields.

\begin{cor}\label{cor-sw0M}
 Let $\Phi=u+iv$ be a solution of the Fackerell-Ipser equation 
 $\bar{\mathcal{O}}_{M,0}=(\Box+2\bar{\Psi}_2+\tfrac{2}{3}\l)\Phi=0$ on a $\l$-vacuum type D spacetime, then: 
\begin{enumerate}
 \item The tensor field $F_{\a\b}(\Phi)= E_{\a\b}(u)-{}^{*}E_{\a\b}(v)$, where
 \begin{equation}\label{bivector}
   {}^{*}E_{\a\b}(v)= -4\c_{[\a}(Y_{\b]}{}^{\g}\c_{\g}v)-4\im(\Psi_2){}^{*}Y_{\a\b}v, 
 \end{equation}
  is a solution to Maxwell equations, $\c^{\a}F_{\a\b}=0=\c^{\a}{}^{*}F_{\a\b}$.
 \item The operator $\mathcal{A}_{M,0}$ defined by
 \begin{equation}\label{Atensor}
  \mathcal{A}_{M,0}(\Phi) = -(Y^{\a\b}+i^{*}Y^{\a\b})\c_{\a}(Y_{\b}{}^{\g}\c_{\g}\Phi)
  -8i\im(\Psi_2)\Psi^{-2/3}_2\Phi
 \end{equation}
 maps solutions of $\bar{\mathcal{O}}_{M,0}(\Phi)=0$ into solutions of $\mathcal{O}_{M,0}(\mathcal{A}_{M,0}(\Phi))=0$.
\end{enumerate}
\end{cor}

\noindent
We note that several simplifications in the above formulae occur in the case in which $\Psi_2$ is real. 
In the first place, the second terms in the RHS of (\ref{bivector}) and (\ref{Atensor}) vanish. Furthermore, using 
(\ref{KY}), (\ref{DKY}) and (\ref{KSD}) (for $k=2$), it is not difficult to see that
\begin{eqnarray}
 {}^{*}Y^{\a\b}Y_{\b}{}^{\g}&=&-2\im(\Psi^{-2/3}_2)g^{\a\g},\\
 Y_{\b}{}^{\g}\xi^{\b}&=&-3\c^{\g}\im(\Psi^{-2/3}_2),
\end{eqnarray}
and then
\begin{equation}\label{AM0real}
  \mathcal{A}_{M,0}(\Phi) = -Y^{\a\b}\c_{\a}(Y_{\b}{}^{\g}\c_{\g}\Phi), \;\;\;\;\;\; \Psi_2\in\mathbb{R},
\end{equation}
which coincides with the well-known Carter operator \cite{Carter}. 
The most important case in our present work in which $\Psi_2$ is real is the Schwarzschild solution, where the Carter operator
coincides in turn with the laplacian on the sphere. These observations are relevant in section \ref{maxschw}, where we 
apply our general results to Maxwell fields on the Schwarzschild-(A)dS solution.

\section{Gravitational perturbations of type D spaces}\label{gravsec}

We now turn our attention to linearized gravity on curved, Petrov type D backgrounds, which include the 
stationary, $\l$-vacuum black hole solutions of the Kerr-(A)dS family. 
Metric perturbations of rotating black holes are traditionally studied by the Teukolsky equations \cite{Teukolsky}, 
which are decoupled, separable differential equations for the extreme perturbed Weyl scalars $\dot{\Psi}_0$ and $\dot{\Psi}_4$. 
These fields have the desirable property of being tetrad and coordinate gauge invariant.
The spin weight zero Weyl scalar $\dot{\Psi}_2$ (which is just tetrad gauge invariant),
on the other hand, has proven to be useful in the spherically symmetric case 
\cite{Dotti}, since, for the odd sector of gravitational perturbations of the Schwarzschild black hole, 
(a rescaled version of) the imaginary part $\im\dot{\Psi}_2$ is gauge invariant, satisfies a wave-like equation (\ref{we}),
and encodes all the information of the gravitational perturbation (in \cite{Dotti} it is used 
as the linearization of a curvature invariant, an identity valid for all type D spacetimes).
Furthermore, the perturbed metric can be reconstructed from this quantity in a covariant, compact form (\ref{metricschw}). 
The application of the spin $\s=2$ theorem \ref{spin2} will allow us to find the origin of this 
reconstruction, as well as similar covariant, compact maps from solutions of the Teukolsky equations to 
metric perturbations.

\subsection{The Bianchi identities and the linearized Einstein tensor}\label{linbianchi}

We now explain how to relate {\em off-shell} the decoupled equations for perturbed Weyl scalars to the linearized 
Einstein equations. For this we use the Bianchi identities. 
As these identities are a consequence of the definition of the curvature tensor, 
$\c_{[\a}R_{\b\g]\d\e}=0$, they are valid in a generic spacetime regardless of the field equations. 
Contracting with the metric, they imply
\begin{equation}
 \c^{\d}R_{\a\b\g\d}=-2\c_{[\a}R_{\b]\g},
\end{equation}
or, in terms of the Weyl tensor,
\begin{equation}
 \c^{\d}C_{\a\b\g\d}=-\c_{[\a}R_{\b]\g}-g_{\g[\a}\c^{\d}R_{\b]\d}+\tfrac{1}{3}g_{\g[\a}\c_{\b]}R.
\end{equation}
Consider now a linear, covariant differential operator $O^{\a\b\g}=O^{[\a\b]\g}$, with $g_{\a\g}O^{\a\b\g}=0$. 
Applying $O^{\a\b\g}$ to the previous identity, one gets
\begin{equation}\label{keyeq3}
 O^{\a\b\g}\c^{\d}C_{\a\b\g\d}=-O^{\a\b\g}\c_{\a}R_{\b\g}.
\end{equation}
Note that the trace-free condition of the operator $O^{\a\b\g}$ implies that we can add to $R_{\a\b}$ a term proportional 
to the metric; this way we can replace $R_{\a\b}$ with the Einstein tensor and add a cosmological constant term:
\begin{equation}\label{keyeq1}
 O^{\a\b\g}\c^{\d}C_{\a\b\g\d}=-O^{\a\b\g}\c_{\a}(G_{\b\g}+\l g_{\b\g}).
\end{equation}
We claim that this equation is the key to relate the decoupled equations for the perturbed Weyl scalars to the 
linearized Einstein equations. 
In the following section we will choose $O^{\a\b\g}$ such that the left hand side of (\ref{keyeq1})
is a decoupled equation for some Weyl scalar plus additional terms that vanish when linearizing.
On the other hand, if we linearize the right hand side 
of (\ref{keyeq1}) around a $\l$-vacuum solution, the linearization operator $\frac{d}{d\e}|_{\e=0}$ 
commutes with $O^{\a\b\g}\c_{\a}$ (because $(G_{\a\b}+\l g_{\a\b})|_{\e=0}=0$)
and we are left with a {\em background} operator acting on the linearized Einstein tensor:
\begin{equation}
 \left.\tfrac{d}{d\e}\right|_{\e=0}\left[O^{\a\b\g}\c^{\d}C_{\a\b\g\d}\right] 
 =-O^{\a\b\g}\c_{\a}\left[\left.\tfrac{d}{d\e}\right|_{\e=0}(G_{\b\g}+\l g_{\b\g})\right].
\end{equation}
The operator $O^{\a\b\g}$ will have the generic form
\begin{equation}
 O^{\a\b\g}=W^{\a\b\g\m}(\c_{\m}+nA_{\m}),
\end{equation}
for some constant $n$, where $W^{\a\b\g\m}$ has the symmetries of the Weyl tensor, 
and the 1-form $A_{\mu}$ is the tensorial counterpart of the spinor $A_{AA'}$ introduced before.

We find that the calculations are most easily performed using the Bianchi identities in spinor form. Following \cite{Penrose1},
contracting with the volume form they are equivalent to $\c^{\a}{}^{*}R_{\a\b\g\d}=0$, where ${}^{*}R_{\a\b\g\d}$ is the 
left-dual Riemann tensor, ${}^{*}R_{\a\b\g\d}=\tfrac{1}{2}\e_{\a\b}{}^{\m\n}R_{\m\n\g\d}$.
In spinor terms (see \cite[section 4.10]{Penrose1}), one gets
\begin{equation}\label{fullbianchi}
 \c^{A}_{B'}\psi_{ABCD}=\c^{A'}_{B}\Phi_{CDA'B'}-2\e_{B(C}\c_{D)B'}\Lambda,
\end{equation}
where $\Lambda=R/24$ (with $R$ the curvature scalar), and $\Phi_{CDA'B'}$ is the spinor analogue of the trace-free Ricci tensor,
\begin{equation}
 \Phi_{ABA'B'}=-\tfrac{1}{2}R_{ABA'B'}+\tfrac{R}{8}\e_{AB}\bar\e_{A'B'}.
\end{equation}
If we apply a linear differential spinor operator $O^{B'BCD}=O^{B'(BCD)}$ in (\ref{fullbianchi}), the trace part vanishes 
because of the symmetries of $O^{B'BCD}$, and, analogously as in (\ref{keyeq1}), we can replace $\Phi_{CDA'B'}$ with the 
Einstein tensor plus a cosmological constant term:
\begin{equation}\label{keyeq2}
 O^{B'BCD}\c^{A}_{B'}\psi_{ABCD}=-\tfrac{1}{2}O^{B'BCD}\c^{A'}_{B}(G_{CDA'B'}+\l\e_{CD}\bar{\e}_{A'B'}).
\end{equation}

\subsection{The decoupled equations}

In \cite{Andersson1}, decoupled equations for all the perturbed Weyl scalars are obtained, assuming that the linearized 
Einstein equations are satisfied (that is, {\em on-shell}). These equations are the 
Teukolsky equations for spin weight $s=\pm2$, corresponding to $\dot{\Psi}_0$ and $\dot{\Psi}_4$;
the `linearized Fackerell-Ipser equation' for spin weight $s=0$, which corresponds to $\dot{\Psi}_2$;
and two more equations for spin weight $s=\pm1$ that are not decoupled in the sense that they involve perturbed 
quantities other than the corresponding scalars $\dot{\Psi}_1$ and $\dot{\Psi}_3$. 
As we mentioned in section \ref{adjop-sec}, we will focus only in the spin weight $s=0,\pm2$ cases.
We recall that, in what follows, all expressions containing linearization of spinors are purely formal; they should 
be understood as the linearization of the corresponding tensor expressions, which is always possible because we are working with 
fields of integer spin (see for example footnote \ref{fn} below).
On the other hand, when linearizing tetrad components of tensors, we assume that there is a monoparametric family 
$\{l^{\a}(\e),n^{\a}(\e),m^{\a}(\e),\bar{m}^{\a}(\e)\}$ such that, in the background, 
$\{l^{\a}(0),n^{\a}(0),m^{\a}(0),\bar{m}^{\a}(0)\}$ is the principal tetrad of a type D space. 
Thus, for example, when the quantity $\dot{\Psi}_0:=\tfrac{d}{d\e}|_{\e=0}\Psi_0(\e)$ appears below, one has
\begin{eqnarray}
\nonumber \tfrac{d}{d\e}|_{\e=0}\Psi_0(\e)&=&\tfrac{d}{d\e}|_{\e=0}(\wt{C}_{\a\b\g\d}l^{\a}m^{\b}l^{\g}m^{\d})\\
\nonumber &=&(\tfrac{d}{d\e}|_{\e=0}\wt{C}_{\a\b\g\d})(l^{\a}m^{\b}l^{\g}m^{\d})|_{\e=0}
  +\wt{C}_{\a\b\g\d}|_{\e=0}\tfrac{d}{d\e}|_{\e=0}(l^{\a}m^{\b}l^{\g}m^{\d})\\
\end{eqnarray}
where $(l^{\a}m^{\b}l^{\g}m^{\d})|_{\e=0}$ refers to the principal tetrad of the background. 
We will not need to work explicitly with the perturbed tetrad.

We now demonstrate the spin $\s=2$ theorem \ref{spin2SV}, and in the following subsection 
we use the Bianchi identities to relate the decoupled equations for the Weyl scalars to the linearized Einstein equations 
(i.e. we evaluate explicitly (\ref{keyeq2})).

The objects (\ref{P}) for the $\s=2$ case are:
\begin{eqnarray}
 P^{A_1A_2A_3A_4}_{(2,2)}&=&o^{A_1}o^{A_2}o^{A_3}o^{A_4}=: o^{A_1A_2A_3A_4}, \\
 P^{A_1A_2A_3A_4}_{(2,0)}&=&6\Psi^{-2/3}_2o^{(A_1}o^{A_2}\iota^{A_3}\iota^{A_4)}=: L^{A_1A_2A_3A_4}, \label{L} \\
 P^{A_1A_2A_3A_4}_{(2,-2)}&=&\Psi^{-4/3}_2\iota^{A_1}\iota^{A_2}\iota^{A_3}\iota^{A_4}=: \Psi^{-4/3}_2 \iota^{A_1A_2A_3A_4}.
\end{eqnarray}
Note that (\ref{L}) is a four-index Killing spinor, $\c_{E'(E}L_{ABCD)}=0$ (the product of two $K_{AB}$'s (\ref{KSD})).
For spin weight $s=\pm2$, theorem \ref{spin2SV} give the $s=\pm2$ Teukolsky equations for 
gravitational perturbations, while for $s=0$ we obtain the linearization of the Fackerell-Ipser operator. 
\begin{thm}[spin $\s=2$]\label{thm-spin2}
 Let $(\mathcal{M}_{\e},g_{\a\b}(\e))$ be a monoparametric family of pseudo-Riemannian manifolds, analytic around $\e=0$, 
 such that $g_{\a\b}(0)$ satisfies the vacuum Einstein equations (with cosmological constant $\lambda$)
 and is of Petrov type D.
 Let $\psi_{ABCD}$ be the Weyl curvature spinor of the metric $g_{\a\b}(\e)$, and let $s=0,\pm2$, then the following equality holds:
  \begin{equation}\label{SEOTspin2}
  \left.\frac{d}{d\e}\right|_{\e=0}\left[\mathcal{S}_{G,s}\mathcal{E}_{G}(\psi_{ABCD})\right]
  =\left.\frac{d}{d\e}\right|_{\e=0}\left[\mathcal{O}_{G,s}\mathcal{T}_{G,s}(\psi_{ABCD})\right],
 \end{equation}
 where the linear differential operators are
 \begin{eqnarray}
  \mathcal{S}_{G,s}(J_{B'BCD}) &:=&\Psi^{4/3}_2 P^{ABCD}_{(2,s)}\c^{B'}_{A}[\Psi^{-4/3}_2 J_{B'BCD}], \\
  \mathcal{E}_{G}(\psi_{ABCD}) &:=&\c^{A}_{B'}\psi_{ABCD}, \\
  \mathcal{O}_{G,s}(\Phi) &:=&\left(\teuk_{2s}+2(1-\tfrac{9}{4}s^2)\Psi_2 +\tfrac{R}{6}\right)\Phi, \\
  \mathcal{T}_{G,s}(\psi_{ABCD}) &:=&-\tfrac{(3-|s|)}{2} P^{ABCD}_{(2,s)}\psi_{ABCD}.
 \end{eqnarray}
\end{thm}

\begin{proof}

We start by the spin weight $s=+2$ case. We have
\begin{eqnarray}
\nonumber \mathcal{S}_{G,+2}\mathcal{E}_{G}(\psi_{ABCD})&=&\left[o^{EBCD}\c^{B'}_{E}+4 o^{EBCD}A^{B'}_{E}\right]\c^{A}_{B'}\psi_{ABCD}\\
\nonumber &=&o^{EBCD}\c^{B'}_{E}\c^{A}_{B'}\psi_{ABCD}+4 o^{EBCD}A^{B'}_{E}\c^{A}_{B'}\psi_{ABCD}.\\
 \label{Sbianchi}
\end{eqnarray}
Noting that $(\c^{A}_{B'}\psi_{ABCD})|_{\e=0}=0$, we can evaluate the term $o^{EBCD}A^{B'}_{E}$ in the background; 
thus, using expression (\ref{A}) for the 1-form $A_{AA'}$, the second term in the bottom line of (\ref{Sbianchi}) gives
\begin{equation}\label{4oA}
 4 o^{EBCD}A^{B'}_{E}\c^{A}_{B'}\psi_{ABCD}=4(\rho\bar{\iota}^{B'}-\tau\bar{o}^{B'}) o^B o^C o^D \c^{A}_{B'}\psi_{ABCD}
\end{equation}
On the other hand, the term with second derivatives of $\psi_{ABCD}$ in (\ref{Sbianchi}) is treated along similar lines as in 
the Dirac and Maxwell cases. Leibniz rule gives
\begin{equation}\label{leibniz}
 o^{EBCD}\c^{B'}_{E}\c^{A}_{B'}\psi_{ABCD}=\c^{B'}_{E}(o^{EBCD}\c^{A}_{B'}\psi_{ABCD})-(\c^{B'}_{E}o^{EBCD})(\c^{A}_{B'}\psi_{ABCD}),
\end{equation}
and manipulations analogous to those performed in (\ref{leibD}) lead to
\begin{eqnarray*}
 o^{EBCD}\c^{B'}_{E}\c^{A}_{B'}\psi_{ABCD}&=&-(\text{\th}-\bar\rho)[o^B o^C o^D\bar{\iota}^{B'}\c^{A}_{B'}\psi_{ABCD}]\\
 & & +(\text{\dh}-\bar\tau')[o^B o^C o^D\bar{o}^{B'}\c^{A}_{B'}\psi_{ABCD}] 
\end{eqnarray*}
where we used the definition of the operators (\ref{thorn}) and (\ref{eth}).
Combining this expression with (\ref{4oA}), we get
\begin{multline}
 \left[o^{EBCD}\c^{B'}_{E} + 4 o^{EBCD}A^{B'}_{E}\right]\c^{A}_{B'}\psi_{ABCD}\\
 =-(\text{\th}-4\rho-\bar\rho)[o^B o^C o^D\bar{\iota}^{B'}\c^{A}_{B'}\psi_{ABCD}]
  +(\text{\dh}-4\tau-\bar\tau')[o^B o^C o^D\bar{o}^{B'}\c^{A}_{B'}\psi_{ABCD}]. 
\end{multline}
Now we just have to put in GHP form the spinor terms in the last expression, for which we use (\ref{bianchieq}), 
and then use the same arguments as in \cite{Andersson1} in order to arrive to the decoupled equation 
(we repeat them here for completeness). 
Reordering terms in $\Psi_0$, $\Psi_1$ and $\Psi_2$, we have
\begin{eqnarray*}
& & \left[o^{EBCD}\c^{B'}_{E} + 4 o^{EBCD}A^{B'}_{E}\right]\c^{A}_{B'}\psi_{ABCD}\\
&=& \left[-(\text{\th}-4\rho-\bar\rho)(\text{\th}'-\rho')+(\text{\dh}-4\tau-\bar\tau')(\text{\dh}'-\tau')\right]\Psi_0\\
 & & +\left[(\text{\th}-4\rho-\bar\rho)(\text{\dh}-4\tau)-(\text{\dh}-4\tau-\bar\tau')(\text{\th}-4\rho)\right]\Psi_1\\
 & & +3(\text{\th}-4\rho-\bar\rho)[\sigma\Psi_2]-3(\text{\dh}-4\tau-\bar\tau')[\kappa\Psi_2].
\end{eqnarray*}
Using (\ref{boxp2}), we see that the term involving $\Psi_0$ is just 
\begin{multline}
 \left[-(\text{\th}-4\rho-\bar\rho)(\text{\th}'-\rho')+(\text{\dh}-4\tau-\bar\tau')(\text{\dh}'-\tau')\right]\Psi_0\\
 =-\frac{1}{2}(\teuk_{4}-10\Psi_2 +\tfrac{R}{6})\Psi_0-3(\kappa\kappa'-\sigma\sigma')\Psi_0.
\end{multline}
On the other hand, for the $\Psi_1$ term we use the commutation relation (\ref{comm4generic}) applied to $\Psi_1$,
with $a=4$ and $p=2$:
\begin{eqnarray}
 & & \left[(\text{\th}-4\rho-\bar\rho)(\text{\dh}-4\tau)-(\text{\dh}-4\tau-\bar\tau')(\text{\th}-4\rho)\right]\Psi_1\\
 &=&[\text{\th}-4\rho,\text{\dh}-4\tau]\Psi_1-\bar\rho(\text{\dh}-4\tau)\Psi_1+\bar\tau'(\text{\th}-4\rho)\Psi_1 \\
\nonumber &=& -10\Psi^2_1+(\text{\dh}'\Psi_1+2\tau'\Psi_1+4\Psi_1(\text{\dh}-\bar\tau+\tau'))\sigma \\
 & & -(\text{\th}'\Psi_1+2\rho'\Psi_1+4\Psi_1(\text{\th}'-\bar\rho'+\rho'))\kappa.
\end{eqnarray}
For the $\Psi_2$ term, we only need to use the Ricci identities (\ref{ricci2}):
\begin{eqnarray*}
 & & 3(\text{\th}-4\rho-\bar\rho)[\sigma\Psi_2]-3(\text{\dh}-4\tau-\bar\tau')[\kappa\Psi_2]\\
 &=& 3\{(\text{\th}\sigma-\text{\dh}\kappa)\Psi_2+\sigma\text{\th}\Psi_2-\kappa\text{\dh}\Psi_2+
 (4\tau\kappa+\bar\tau'\kappa-4\rho\sigma-\bar\rho\sigma)\Psi_2\}\\
 &=&3\Psi_2\Psi_0+3\sigma(\text{\th}-3\rho)\Psi_2-3\kappa(\text{\dh}-3\tau)\Psi_2.
\end{eqnarray*}
Then, recalling (\ref{Sbianchi}) we get
\begin{equation}
 \nonumber \mathcal{S}_{G,2}\mathcal{E}_{G}(\psi_{ABCD})=-\tfrac{1}{2}(\teuk_{4}-16\Psi_2+\tfrac{R}{6})\Psi_0
 +B[\Psi_0,\Psi_1,\kappa,\sigma]
\end{equation}
where
\begin{eqnarray}
\nonumber B[\Psi_0,\Psi_1,\kappa,\sigma]&:=&-3(\kappa\kappa'-\sigma\sigma')\Psi_0-10\Psi^2_1\\
\nonumber & & +(\text{\dh}'\Psi_1+2\tau'\Psi_1+4\Psi_1(\text{\dh}-\bar\tau+\tau'))\sigma \\
\nonumber & & -(\text{\th}'\Psi_1+2\rho'\Psi_1+4\Psi_1(\text{\th}'-\bar\rho'+\rho'))\kappa\\
 & & +3\sigma(\text{\th}-3\rho)\Psi_2-3\kappa(\text{\dh}-3\tau)\Psi_2.
\end{eqnarray}
Linearizing this expression (in the sense described at the beginning of this subsection)
around a type D background spacetime, and using the Bianchi identities (\ref{bianchiD}), we get
\begin{equation}
 \left.\frac{d}{d\e}\right|_{\e=0}B[\Psi_0,\Psi_1,\kappa,\sigma]=0
\end{equation}
(because all terms in $B[\Psi_0,\Psi_1,\kappa,\sigma]$ are at least order $\e^2$)
and therefore, recalling that $\Psi_0|_{\e=0}=0$,
\begin{equation}\label{dSE2}
 \left.\frac{d}{d\e}\right|_{\e=0}[\mathcal{S}_{G,2}\mathcal{E}_{G}(\psi_{ABCD})]
 =-\tfrac{1}{2}(\teuk_{4}-16\Psi_2+\tfrac{2}{3}\l)\dot\Psi_0,
\end{equation}
which is what we wanted to prove.\footnote{\label{fn}The linearization of the spinor expression 
$\mathcal{S}_{G,2}\mathcal{E}_{G}(\psi_{ABCD})$ in the LHS of (\ref{dSE2}) is an example of what we mean by 
`understood in a tensor sense', since one uses the equality
\begin{equation*}
 \Psi^{4/3}_2o^{EBCD}\c^{B'}_{E}[\Psi^{-4/3}_2\c^{A}_{B'}\psi_{ABCD}]
 =-\tfrac{1}{2}\Psi^{4/3}_2\stackrel[]{0}{W}{}^{\a\b\g\e}\c_{\e}[\Psi^{-4/3}_2\c^{\d}\wt{C}_{\a\b\g\d}]
\end{equation*}
}

For spin weight $s=-2$, as observed in \cite{Andersson1}, this case follows from the previous one by 
simply applying the prime operation and using the transformation law (\ref{boxt'}).

Consider now the spin weight $s=0$ case.
The proof of this case goes along similar lines as those of the previous one: we start by
\begin{eqnarray}
\nonumber \mathcal{S}_{G,0}\mathcal{E}_{G}(\psi_{ABCD})&=& \left[L^{EBCD}\c^{B'}_{E} + 4 L^{EBCD}A^{B'}_{E}\right]\c^{A}_{B'}\psi_{ABCD}\\
\nonumber &= & L^{EBCD}\c^{B'}_{E}\c^{A}_{B'}\psi_{ABCD}+4 L^{EBCD}A^{B'}_{E}\c^{A}_{B'}\psi_{ABCD},\\
\end{eqnarray}
and use Leibniz rule for the term with second derivatives of $\psi_{ABCD}$:
\begin{equation}\label{leibniz-sw0}
 L^{EBCD}\c^{B'}_{E}\c^{A}_{B'}\psi_{ABCD}=\c^{B'}_{E}(L^{EBCD}\c^{A}_{B'}\psi_{ABCD})-(\c^{B'}_{E}L^{EBCD})(\c^{A}_{B'}\psi_{ABCD}).
\end{equation}
Now, the second term in this equation is more easily calculated taking into account that in the end we want to 
linearize around a $\l$-vacuum solution, such that $(\c^{A}_{B'}\psi_{ABCD})|_{\e=0}=0$; then
\begin{multline}
 \left.\frac{d}{d\e}\right|_{\e=0}\left[-(\c^{B'}_{E}L^{EBCD})(\c^{A}_{B'}\psi_{ABCD})\right]\\
 =-(\c^{B'}_{E}L^{EBCD})|_{\e=0}\left.\frac{d}{d\e}\right|_{\e=0}\left[\c^{A}_{B'}\psi_{ABCD}\right]
\end{multline}
This implies that we can use identities from the unperturbed spacetime. From the definition of $L_{ABCD}$, eq. (\ref{L}), 
we see that it is propotional to the Weyl spinor of the type D background: 
$L_{ABCD}=\Psi^{-5/3}_2\mathring{\psi}_{ABCD}$, where $\mathring{\psi}_{ABCD}=(\psi_{ABCD})|_{\e=0}$. 
Using the background Bianchi identities, we then have
\begin{equation}
 (\c^{B'}_{E}L^{EBCD})|_{\e=0}=(5L^{EBCD}A^{B'}_{E})|_{\e=0},
\end{equation}
where we have used expression (\ref{AD}). 
Therefore,
\begin{equation}
 \left.\frac{d}{d\e}\right|_{\e=0}\left[-(\c^{B'}_{E}L^{EBCD})(\c^{A}_{B'}\psi_{ABCD})\right]
 =\left.\frac{d}{d\e}\right|_{\e=0}\left[-5L^{EBCD}A^{B'}_{E}\c^{A}_{B'}\psi_{ABCD}\right],
\end{equation}
and then
\begin{multline}\label{dSE}
\left.\frac{d}{d\e}\right|_{\e=0} [\mathcal{S}_{G,0}\mathcal{E}_{G}(\psi_{ABCD})]\\
=\left.\frac{d}{d\e}\right|_{\e=0}\left[\c^{B'}_{E}(L^{EBCD}\c^{A}_{B'}\psi_{ABCD})-L^{EBCD}A^{B'}_{E}\c^{A}_{B'}\psi_{ABCD}\right].
\end{multline}
The term inside the bracket in (\ref{dSE}) can be calculated without linearizing, following similar manipulations 
as in previous cases and using the explicit expressions (\ref{A}) and (\ref{L}). The result is
\begin{eqnarray*}
& & \c^{B'}_{E}(L^{EBCD}\c^{A}_{B'}\psi_{ABCD})-L^{EBCD}A^{B'}_{E}\c^{A}_{B'}\psi_{ABCD}\\
 &=& -3(\text{\th}-\bar\rho)[\Psi^{-2/3}_2o^{B}\iota^{C}\iota^{D}\bar{\iota}^{B'}\c^{A}_{B'}\psi_{ABCD}]\\
 & & -3(\text{\dh}'-\bar\tau)[\Psi^{-2/3}_2o^{B}o^{C}\iota^{D}\bar{\iota}^{B'}\c^{A}_{B'}\psi_{ABCD}]\\
 & & +3(\text{\dh}-\bar\tau')[\Psi^{-2/3}_2o^{B}\iota^{C}\iota^{D}\bar{o}^{B'}\c^{A}_{B'}\psi_{ABCD}]\\
 & & +3(\text{\th}'-\bar\rho')[\Psi^{-2/3}_2o^{B}o^{C}\iota^{D}\bar{o}^{B'}\c^{A}_{B'}\psi_{ABCD}].
\end{eqnarray*}

Now we just have to use (\ref{bianchieq}) for the corresponding components of $\c^{A}_{B'}\psi_{ABCD}$ 
(note that we need the second, third, sixth and seventh equations in (\ref{bianchieq})), and the fact that 
\begin{equation}
 \Psi^{-2/3}_2(\text{\th}'-3\rho')\Psi_2=3(\text{\th}'-\rho')\Psi^{1/3}_2
\end{equation}
and similarly for the other derivatives. This gives
\begin{eqnarray*}
 & & \c^{B'}_{E}(L^{EBCD}\c^{A}_{B'}\psi_{ABCD})-L^{EBCD}A^{B'}_{E}\c^{A}_{B'}\psi_{ABCD}\\
 &=& -3(\text{\th}-\bar\rho)\left[3(\text{\th}'-\rho')\Psi^{1/3}_2
 +\Psi^{-2/3}_2[-(\text{\dh}-2\tau)\Psi_3+2\kappa'\Psi_1-\sigma\Psi_4]\right]\\
 & & -3(\text{\dh}'-\bar\tau)\left[-3(\text{\dh}-\tau)\Psi^{1/3}_2
 +\Psi^{-2/3}_2[(\text{\th}'-2\rho')\Psi_1-2\sigma\Psi_3+\kappa'\Psi_0]\right]\\
 & & +3(\text{\dh}-\bar\tau')\left[3(\text{\dh}'-\tau')\Psi^{1/3}_2
 +\Psi^{-2/3}_2[-(\text{\th}-2\rho)\Psi_3+2\sigma'\Psi_1-\kappa\Psi_4]\right]\\
 & & +3(\text{\th}'-\bar\rho')\left[-3(\text{\th}-\rho)\Psi^{1/3}_2
 +\Psi^{-2/3}_2[(\text{\dh}'-2\tau')\Psi_1-2\kappa\Psi_3+\sigma'\Psi_0]\right] 
\end{eqnarray*}
Note that the sum of the second and fourth lines is just the primed version of the sum of the first and third ones, 
and then we only calculate the latter:
\begin{eqnarray}
\nonumber && -3(\text{\th}-\bar\rho)\left[3(\text{\th}'-\rho')\Psi^{1/3}_2
  +\Psi^{-2/3}_2[-(\text{\dh}-2\tau)\Psi_3+2\kappa'\Psi_1-\sigma\Psi_4]\right]\\
\nonumber && +3(\text{\dh}-\bar\tau')\left[3(\text{\dh}'-\tau')\Psi^{1/3}_2
  +\Psi^{-2/3}_2[-(\text{\th}-2\rho)\Psi_3+2\sigma'\Psi_1-\kappa\Psi_4]\right]\\
\nonumber &=& -9\left[(\text{\th}-\bar\rho)(\text{\th}'-\rho')-(\text{\dh}-\bar\tau')(\text{\dh}'-\tau')\right]\Psi^{1/3}_2\\
\nonumber & & +3(\text{\th}-\bar\rho)[\Psi^{-2/3}_2(\text{\dh}-2\tau)\Psi_3]-3(\text{\dh}-\bar\tau')
   [\Psi^{-2/3}_2(\text{\th}-2\rho)\Psi_3]\\
\nonumber & & -3(\text{\th}-\bar\rho)[\Psi^{-2/3}_2(2\kappa'\Psi_1-\sigma\Psi_4)]+3(\text{\dh}-\bar\tau')
   [\Psi^{-2/3}_2(2\sigma'\Psi_1-\kappa\Psi_4)].\\
\label{1+3}
\end{eqnarray}
Recalling the explicit expression (\ref{boxp2}) of the weighted wave operator $\teuk_{p}$ (and taking into account 
that $\Psi_2$ is type $\{0,0\}$), we see that the term with $\Psi^{1/3}_2$ is just
\begin{multline}
 \left[(\text{\th}-\bar\rho)(\text{\th}'-\rho')-(\text{\dh}-\bar\tau')(\text{\dh}'-\tau')\right]\Psi^{1/3}_2\\
 =\tfrac{1}{2}(\Box+2\Psi_2 +\tfrac{R}{6})\Psi^{1/3}_2+(\kappa\kappa'-\sigma\sigma')\Psi^{1/3}_2
\end{multline}
The second term in the bottom line of the last equation will vanish when we linearize around a type D spacetime 
(because of (\ref{ghpD})), and so will
the terms with $\Psi_3$ and with $\Psi_1,\Psi_4$ in (\ref{1+3}); for the term with $\Psi_3$ we need first reorder as
\begin{eqnarray*}
 & & (\text{\th}-\bar\rho)[\Psi^{-2/3}_2(\text{\dh}-2\tau)\Psi_3]-(\text{\dh}-\bar\tau')[\Psi^{-2/3}_2(\text{\th}-2\rho)\Psi_3]\\
 &=&\tfrac{2}{3}\Psi^{-5/3}_2(\text{\th}\Psi_2)(\text{\dh}-2\tau)\Psi_3-\Psi^{-2/3}_2(\text{\th}-\bar\rho)(\text{\dh}-2\tau)\Psi_3\\
 & &-\tfrac{2}{3}\Psi^{-5/3}_2(\text{\dh}\Psi_2)(\text{\th}-2\rho)\Psi_3+\Psi^{-2/3}_2(\text{\th}-\bar\tau')(\text{\th}-2\rho)\Psi_3,
\end{eqnarray*}
and recall the background Bianchi identities (\ref{bianchiD}) and the commutation relation (\ref{comm4}).
Linearizing and taking into account the vanishing of the terms just mentioned, we finally get
\begin{eqnarray*}
 & & \left.\frac{d}{d\e}\right|_{\e=0}[\mathcal{S}_{G,0}\mathcal{E}_{G}(\psi_{ABCD})]\\
 &=&\left.\frac{d}{d\e}\right|_{\e=0}\left[\c^{B'}_{E}(L^{EBCD}\c^{A}_{B'}\psi_{ABCD})-L^{EBCD}A^{B'}_{E}\c^{A}_{B'}\psi_{ABCD}\right]\\
 &=&-9\left.\frac{d}{d\e}\right|_{\e=0}\left[(\Box+2\Psi_2 +\tfrac{R}{6})\Psi^{1/3}_2\right],
\end{eqnarray*}
which gives the desired result.

\end{proof}

\subsection{Tensor expressions and Einstein equations}

Using the idea we described in section \ref{linbianchi},
we now give the identities that combine the previous decoupled equations with the linearized Ricci tensor.
First, we recall the definition of the anti-self-dual 2-forms given in the introduction
\begin{equation}
 \stackbin[]{0}{M}_{\a\b}:=2l_{[\a}m_{\b]}, \;\;\;
 \stackbin[]{1}{M}_{\a\b}:=2l_{[\a}n_{\b]}+2\bar{m}_{[\a}m_{\b]}, \;\;\; \stackbin[]{2}{M}_{\a\b}:=2\bar{m}_{[\a}n_{\b]},
\end{equation}
and the anti-self-dual tensors with the symmetries of the Weyl tensor:
\begin{eqnarray}
 \stackbin[]{0}{W}_{\a\b\g\d}&:=&\stackbin[]{0}{M}_{\a\b}\stackbin[]{0}{M}_{\g\d}, \\
 \stackbin[]{2}{W}_{\a\b\g\d}&:=&\stackbin[]{0}{M}_{\a\b}\stackbin[]{2}{M}_{\g\d}
 +\stackbin[]{2}{M}_{\a\b}\stackbin[]{0}{M}_{\g\d}+\stackbin[]{1}{M}_{\a\b}\stackbin[]{1}{M}_{\g\d}, \\
 \stackbin[]{4}{W}_{\a\b\g\d}&:=&\stackbin[]{2}{M}_{\a\b}\stackbin[]{2}{M}_{\g\d}.
\end{eqnarray}
We also recall the expression of the linearized Einstein tensor in terms of the metric perturbation:
\begin{equation}
 \dot{G}_{\a\b}[h]=-\tfrac{1}{2}\Box h_{\a\b}-\tfrac{1}{2}\c_{\a}\c_{\b}h+\c^{\g}\c_{(\a}h_{\b)\g}
 +\tfrac{1}{2}g_{\a\b}(\Box h -\c^{\g}\c^{\d}h_{\g\d}),
\end{equation}
where $h=g^{\g\d}h_{\g\d}$.
The combination of theorem \ref{thm-spin2} with the analysis of section \ref{linbianchi} then leads to the following 
corollary (which is just theorem \ref{spin2}):
\begin{cor}
Let $(\mathcal{M}_{\e},g_{\a\b}(\e))$ be a monoparametric family of pseudo-Riemannian manifolds, analytic 
around $\e=0$, such that $g_{\a\b}(0)$ is of Petrov type D and satisfies the vacuum Einstein equations 
with cosmological constant $\lambda$. 
Denoting the linearization of a quantity $T$ by $\dot{T}:=\frac{d}{d\e}|_{\e=0}T(\e)$, we have the following equalities:
\begin{equation}
\Psi_2^{4/3} \stackbin[]{0}{W}{}^{\a\g\b\d}\c_{\d}\left[\Psi_2^{-4/3}\c_{\g}(\dot{G}_{\a\b}[h]+\l h_{\a\b})\right]
 =(\teuk_{+4}-16\Psi_2+\tfrac{2}{3}\l)\dot\Psi_0[h], \label{sw-2G}
\end{equation}
\begin{multline}
\Psi^{2/3}_{2}\stackbin[]{2}{W}{}^{\a\g\b\d}\c_{\d}\left[\Psi^{-4/3}_2\c_{\g}(\dot{G}_{\a\b}[h]+\l h_{\a\b})\right]\\
 =6\left[(\Box+8\Psi_2+\tfrac{2}{3}\l)[\Psi^{-2/3}_2\dot{\Psi}_2[h]]+3(\dot{\Box}_h+\tfrac{\dot{R}_h}{6})\Psi^{1/3}_2\right], \label{sw0G}
\end{multline}
\begin{equation}
\stackbin[]{4}{W}{}^{\a\g\b\d}\c_{\d}\left[\Psi_2^{-4/3}\c_{\g}(\dot{G}_{\a\b}[h]+\l h_{\a\b})\right]
 =(\teuk_{-4}-16\Psi_2+\tfrac{2}{3}\l)[\Psi^{-4/3}_2\dot\Psi_4[h]]. \label{sw+2G}
\end{equation}
\end{cor}

The previous equations show that if the linearized Einstein equations $\dot{G}_{\a\b}[h]+\l h_{\a\b}=0$ are satisfied, 
then we have decoupled equations for the perturbed Weyl scalars. On the other hand,
in order to see whether we can construct solutions of the linearized Einstein equations from solutions of the decoupled equations, 
we can put these identities in an operator equality form such as (\ref{SEOTspin2}).
We separate cases according to extreme and zero spin weight, since there are important differences between them.

\subsubsection{Extreme spin weight}

For $s=\pm2$ we define new operators $\mathcal{S}_s$ and $\mathcal{E}$ such that
equations (\ref{sw-2G}) and (\ref{sw+2G}) adopt the form
\begin{equation}\label{SEhOTh}
 \mathcal{S}_s\mathcal{E}(h_{\a\b})=\mathcal{O}_s\mathcal{T}_s(h_{\a\b})
\end{equation}
for all symmetric tensor field $h_{\a\b}=h_{(\a\b)}$, where
\begin{eqnarray}
 \mathcal{S}_s(H_{\a\b})&:=&\Psi_2^{(s+2)/3} \stackbin[]{2-s}{W}{}^{\a\g\b\d}\c_{\d}[\Psi_2^{-4/3}\c_{\g}H_{\a\b}], \\
 \mathcal{E}(h_{\a\b})&:=&\dot{G}_{\a\b}[h]+\l h_{\a\b}, \\
 \mathcal{O}_s(\Phi)&:=&(\teuk_{2s}-16\Psi_2+\tfrac{2}{3}\l)\Phi, \\
 \mathcal{T}_s(h_{\a\b})&:=& \Psi^{(s-2)/3}_2 \dot{\Psi}_{2-s}[h].
\end{eqnarray}
Since $\mathcal{E}$ is self-adjoint, and $\mathcal{O}^{\dag}_{s}=\bar{\mathcal{O}}_{-s}$ (eq. (\ref{Oadj})), the adjoint equation 
$\mathcal{E}\mathcal{S}^{\dag}_s(\Phi)=\mathcal{T}^{\dag}_s\mathcal{O}^{\dag}_s(\Phi)$ leads immediately to the following 
corollary:
\begin{cor}
Consider a vacuum type D spacetime with cosmological constant, 
and let $\Phi_s$ be a solution of the spin-weight $s=\pm2$ Teukolsky equation. Then
\begin{equation}\label{m-rec-esw}
 \stackbin[]{s}{h}_{\a\b}(\Phi)=\c_{\g}[\Psi_2^{-4/3}\c_{\d}(\stackbin[]{2+s}{W}_{(\a}{}^{\g\d}{}_{\b)}\Psi_2^{(2-s)/3}\Phi_s)]
\end{equation}
is a complex solution of the linearized Einstein equations.
\end{cor}

It can be shown that (\ref{m-rec-esw}) for $s=-2$ coincides with the Kegeles $\&$ Cohen ansatz \cite[Eq.(5.4)]{Kegeles} 
(in that work $h_{\a\b}$ is given in spinor form and in terms of a Hertz spinor and a gauge spinor). 
We also note that the difference between the metric perturbations constructed in the form (\ref{m-rec-esw}) for 
$s=+2$ and $s=-2$ is described in the recent work \cite{Aksteiner2}, and that further symmetry operators for extreme spin 
weight are constructed in \cite{Aksteiner}.

\subsubsection{Spin weight zero, real $\Psi_2$ case}\label{sw0G-sec}

The `inhomogeneous' term in the right hand side of (\ref{sw0G}), namely $(\dot{\Box}_h+\tfrac{\dot{R}_h}{6})\Psi^{1/3}_2$, 
makes it more difficult to formulate an operator equality like (\ref{SEhOTh}) for the spin weight zero case.
The simplest possibility is in the case in which $\Psi_2$ is a real field, since then 
we can take the imaginary part in (\ref{sw0G}) and get
$\mathcal{S}_0\mathcal{E}(h_{\a\b})=\mathcal{O}_0\mathcal{T}_0(h_{\a\b})$ for all $h_{\a\b}=h_{(\a\b)}$, where 
\begin{eqnarray}
 \mathcal{S}_0(H_{\a\b})&=&\tfrac{1}{2}\Psi^{2/3}_{2}\stackbin[]{2}{{}^{*}W}{}^{\a\g\b\d}\c_{\d}[\Psi^{-4/3}_2\c_{\g}H_{\a\b}],\\
 \mathcal{E}(h_{\a\b})&=&\dot{G}_{\a\b}[h]+\l h_{\a\b},\\
 \mathcal{O}_0(\Phi)&=&6(\Box+8\Psi_2+\tfrac{2}{3}\lambda)\Phi,\\
 \mathcal{T}_0(h_{\a\b})&=&\Psi^{-2/3}_{2}\im\{\dot{\Psi}_2[h]\},
\end{eqnarray}
with $\stackbin[]{2}{{}^{*}W}_{\a\b\g\d}=-2\im\{\stackbin[]{2}{W}_{\a\b\g\d}\}$. 
Taking the adjoint equation, and using the fact that $\mathcal{E}$ and $\mathcal{O}_0$ are both self-adjoint, we
obtain that if $\Phi$ is a solution to $(\Box+8\Psi_2+\tfrac{2}{3}\lambda)\Phi=0$, then the tensor field
\begin{equation}\label{m-rec-sw0}
 h_{\a\b}(\Phi)=\tfrac{1}{2}\c_{\g}[\Psi^{-4/3}_2\c_{\d} (\stackbin[]{2}{{}^{*}W}_{(\a}{}^{\g\d}{}_{\b)}\Psi^{2/3}_2\Phi)]
\end{equation}
is a solution to the linearized Einstein equations, $\dot{G}_{\a\b}[h]+\l h_{\a\b}=0$.
As we show in section \ref{gravschw} below, this applies to the odd sector of gravitational perturbations of the 
Schwarzschild-(A)dS solution. 

Note that, since we are working on a background type D spacetime, we have
\begin{equation}
 \stackbin[]{2}{W}_{\a\b\g\d}\equiv\Psi^{-1}_2\wt{C}_{\a\b\g\d},
\end{equation}
hence we can replace the corresponding expressions with the background Weyl tensor.

On the other hand, we also note that the field $\Psi^{-2/3}_{2}\dot{\Psi}_2$ can be put in terms of the Killing-Yano 
tensors in the following way. In a generic spacetime, we have the identity \cite[Eqs. (8.3.8) and (8.3.10)]{Penrose1}
\begin{equation}
 \psi_{AB}{}^{CD}\psi_{CD}{}^{AB}=6\Psi^{2}_2+2\Psi_0\Psi_4-8\Psi_1\Psi_3,
\end{equation}
or equivalently
\begin{equation}
 6\Psi^{2}_2=\tfrac{1}{4}\wt{C}_{\a\b}{}^{\g\d}\wt{C}_{\g\d}{}^{\a\b}-2\Psi_0\Psi_4+8\Psi_1\Psi_3.
\end{equation}
Linearizing this equation around a type D background, we get
\begin{equation}
 \dot{\Psi}_2=\tfrac{1}{24}\Psi^{-1}_2\wt{C}_{\g\d}{}^{\a\b}\dot{\wt{C}}_{\a\b}{}^{\g\d}.
\end{equation}
Now, using the expression (\ref{sdweylD}) for $\wt{C}_{\a\b}{}^{\g\d}$ and the symmetries of the Weyl tensor, we obtain
\begin{equation}\label{KYCgeneric}
 \Psi^{-2/3}_2\dot{\Psi}_2=-\tfrac{1}{4k^2}\wt{Y}_{\g\d}\wt{Y}^{\a\b}\dot{\wt{C}}_{\a\b}{}^{\g\d}.
\end{equation}
We will work further this expression below, when we apply our results to the Schwarzschild-(A)dS solution.

\section{Spherically symmetric spacetimes}\label{sphsym-sec}

In this section we show the relation of our results, particularized to the Schwarzschild-(A)dS case, with the so-called 
$2+2$ decomposition valid in spherically symmetric spacetimes.
For the latter formalism, we follow closely \cite{Sarbach} (note however that we take the metric to have signature $(+---)$). 
This approach takes advantage of the warped product structure of the 
background manifold $\mathcal{M}=\wt{\mathcal{M}}\times_{r^2}S^2$, with coordinates $z^{\a}=(x^a,y^i)$ and metric
\begin{equation}\label{wp}
 g_{\a\b}(z)dz^{\a}dz^{\b}=\wt{g}_{ab}(x)dx^a dx^b+r^2\wh{g}_{ij}(y)dy^i dy^j.
\end{equation}
Lowercase latin indices $a,b,c,...$ denote quantities in the orbit space $\wt{\mathcal{M}}=\mathcal{M}/SO(3)$, while indices
$i,j,k,...$ refer to quantities on the sphere $S^2$. The metric, covariant derivative and volume form of $\wt{\mathcal{M}}$ are 
respectively $\wt{g}_{ab}$, $\wt{D}_a$ and $\wt{\e}_{ab}$; whereas those of $S^2$ are $\wh{g}_{ij}$, 
$\wh{D}_i$ and $\wh{\e}_{ij}$.
The wave operators in $\wt{\mathcal{M}}$ and $S^2$ are then $\wt{\Delta}:=\wt{g}^{ab}\wt{D}_a\wt{D}_b$ 
and $\wh{\Delta}:=\wh{g}^{ij}\wh{D}_i\wh{D}_j$, respectively.
The relation between the Christoffel symbols of $g_{\a\b}$ and those of $\wt{g}_{ab}$ and $\wh{g}_{ij}$ is
\begin{eqnarray}
 & & \G^{d}{}_{ab}=\wt{\G}^{d}{}_{ab}, \;\;\;\; \G^{d}{}_{ai}=0, \;\;\;\; \G^{d}{}_{ij}=-rr^{d}\wh{g}_{ij}, \\
 & & \G^{i}{}_{ab}=0, \;\;\;\; \G^{i}{}_{aj}=\tfrac{r_a}{r}\d^{i}_{j}, \;\;\;\; \G^{k}{}_{ij}=\wh{\G}^{k}{}_{ij}. 
\end{eqnarray}
For further relations we refer the reader to \cite{Sarbach}.

In the Schwarzschild-(A)dS spacetime, we have $\Psi_2=-M/r^3$.
We take the constant $k$ in the definition (\ref{KSD}) of the Killing spinor to be real, and for 
convenience we define $b:=-kM^{-1/3}$. 
The Killing-Yano tensor (\ref{KY}) and its dual (\ref{DKY}) are then
\begin{eqnarray}
  Y_{\a\b}dz^{\a}\wedge dz^{\b} &\equiv& br^3\wh{\e}_{ij}dy^{i}\wedge dy^{j}, \label{kyschw}\\
  {}^{*}Y_{\a\b}dz^{\a}\wedge dz^{\b} &\equiv& br\wt{\e}_{ab}dx^{a}\wedge dx^{b} \label{kydschw},
\end{eqnarray}
where in Schwarzschild coordinates $\{ t,r,\theta,\varphi \}$, $\wh\e=\sin\theta d\theta\wedge d\varphi$ and $\wt{\e}=dt\wedge dr$.
The Weyl tensor and its dual can be deduced from (\ref{sdweylD}), (\ref{weyl}):
\begin{eqnarray}
 C_{\a\b\g\d}&=&\tfrac{3M}{b^2r^5}(Y_{\a\b}Y_{\g\d}-{}^{*}Y_{\a\b} {}^{*}Y_{\g\d})-\tfrac{2M}{r^3}g_{\a[\g}g_{\d]\b}, \label{Cschw}\\
 {}^{*}C_{\a\b\g\d}&=&\tfrac{3M}{b^2r^5}({}^{*}Y_{\a\b}Y_{\g\d}+Y_{\a\b} {}^{*}Y_{\g\d})+\tfrac{M}{r^3}\e_{\a\b\g\d} \label{C*schw}.
\end{eqnarray}
With our signature conventions, the Schwarzschild-(A)dS metric (\ref{wp}) in coordinates $\{ t,r,\theta,\varphi \}$ has the form
\begin{equation}\label{Sads}
 ds^2=f(r)dt^2-\frac{dr^2}{f(r)}-r^2(d\theta^2+\sin^2\theta d\varphi^2),
\end{equation}
where
\begin{equation}\label{f}
 f(r)=1-\tfrac{2M}{r}-\tfrac{\l}{3}r^2.
\end{equation}
The Ricci tensor of (\ref{Sads}) solves\footnote{the Ricci tensor we use has the opposite sign to the conventional one \cite{Wald2},
see appendix \ref{appA}.}
\begin{equation}
 R_{\a\b}-\l g_{\a\b}=0.
\end{equation}

\subsection{Maxwell fields}\label{maxschw}

According to the $2+2$ decomposition of the Maxwell field performed in \cite{Sarbach}, the information of the field is contained in 
two master scalar variables, $\phi^{-}$ and $\phi^{+}$, codifying respectively the odd and even parity sectors of the electromagnetic 
perturbation. 
It can be shown that spherically symmetric Maxwell fields (i.e., with $\ell=0$ in a decomposition into spherical harmonics) are
static (see e.g. \cite{Price}, also \cite[Appendix A]{Blue}), therefore they are not interesting for the stability problem and we 
can then take $\ell\geq1$, which implies that the laplacian $\wh\Delta$ is invertible.
Assuming vacuum Maxwell equations hold, the reconstruction of the field from the variables $\phi^{-}$ and $\phi^{+}$ is:
\begin{equation}\label{Fsarbach}
 F_{ab}=-\tfrac{1}{r^2}\wt{\e}_{ab}\phi^{+}, \;\;\;\;
 F_{ai}=\wh\e_{i}{}^{j}\wh{D}_j\wt{D}_a\phi^{-}+\wh{D}_i\wt{\e}_{a}{}^{b}\wt{D}_b\wh\Delta^{-1}\phi^{+}, \;\;\;\;
 F_{ij}=-\wh\e_{ij}\wh\Delta\phi^{-}.
\end{equation}
The wave equations satisfied by $\phi^{\pm}$ (see \cite{Sarbach} and the decomposition (\ref{waveop}) of the wave operator 
below) are equivalent to
\begin{equation}\label{FIS}
 (\Box+2\Psi_2+\tfrac{2}{3}\l)[\tfrac{\phi^{+}}{r}+i\tfrac{\phi^{-}}{r}]=0.
\end{equation}
The scalar field $\Phi:=\tfrac{\phi^{+}}{r}+i\tfrac{\phi^{-}}{r}$ satisfies then the Fackerell-Ipser equation, thus we can construct a
{\em new} electromagnetic field using corollary \ref{cor-sw0M}. 
In order to see the relation between this new field and the original one (\ref{Fsarbach}),
we need calculate the components of the tensors $E_{\a\b}(u)$ and ${}^{*}E_{\a\b}(v)$ of formula (\ref{bivector})
according to the $2+2$ decomposition. 
Using the explicit form of the Killing-Yano tensor, and the fact that $\im(\Psi_2)=0$, we get:
\begin{eqnarray}
\nonumber & & {}^{*}E_{ab}(v)=0, \;\;\;\; {}^{*}E_{ai}(v)=-2b\wh\e_{i}{}^{j}\wh{D}_{j}\wt{D}_{a}(rv), 
 \;\;\;\; {}^{*}E_{ij}(v)=2b\wh\e_{ij}\wh\Delta(rv), \\
 & & \label{E*} \\
\nonumber & & E_{ab}(u)=-\tfrac{2b}{r^2}\wt{\e}_{ab}\wh\Delta(ru), 
 \;\;\;\; E_{ai}(u)=2b\wh{D}_{i}\wt{\e}_{a}{}^{b}\wt{D}_{b}(ru), \;\;\;\; E_{ij}(u)=0. \\
& &  \label{E}
\end{eqnarray}
Furthermore, if (\ref{FIS}) holds, then using that $[\wh\Delta^{-1},\Box+2\Psi_2+2\l/3]=0$, it also holds
\begin{equation}
 (\Box+2\Psi_2+\tfrac{2}{3}\l)[\wh\Delta^{-1}(\tfrac{\phi^{+}}{2br})-i\tfrac{\phi^{-}}{2br}]=0,
\end{equation}
and therefore the electromagnetic field constructed from this solution in the form (\ref{E*})-(\ref{E}) 
(i.e. replacing $u\equiv\wh\Delta^{-1}(\tfrac{\phi^{+}}{2br})$ and $v\equiv-\tfrac{\phi^{-}}{2br}$ in those expressions)
coincides exactly with the original field.

Recalling that $\wh{\Delta}=\mathcal{A}_{M,0}$ (see equation (\ref{AM0real})), we summarize the previous results as follows:
\begin{thm}\label{est-max}
 The dynamics of the Maxwell field on the Schwarzschild-(A)dS spacetime is governed by solutions $\Phi=u+iv$ of 
 the Fackerell-Ipser equation, 
 \begin{equation}
  \left(\Box-\frac{2M}{r^3}+\frac{2}{3}\l\right)\Phi=0,
 \end{equation}
 where the real and imaginary parts of $\Phi$ codify respectively the information of the even and odd parity sectors,
 and the covariant four-dimensional reconstruction of the electromagnetic field is
 \begin{equation}
  F_{\a\b}=-\tfrac{2}{b}\c_{[\a}(Y_{\b]}{}^{\g}\c_{\g}v)+\tfrac{1}{b}\e_{\a\b}{}^{\g\d}\c_{\g}(Y_{\d}{}^{\e}\c_{\e}(\mathcal{A}^{-1}_{M,0}u)).
 \end{equation}
\end{thm}

\noindent
Using this result, the linear stability of the Maxwell field on Schwarzschild-dS can be proved along similar lines
as those used in \cite{Dotti} for spin 2 (the problem for the Anti-de Sitter case is more delicate because of the boundary 
conditions \cite{AranedaDotti}).
This is the spin 1 analogue of the results of \cite{Dotti}.

\subsection{Gravitational perturbations}\label{gravschw}

In this subsection we apply the general results of section \ref{gravsec} to linearized gravity on Schwarzschild-(A)dS. 
We will only work on the odd sector of gravitational perturbations, where the perturbed metric is \cite{Sarbach}
\begin{equation}\label{odds}
 h^{-}_{ab}=0, \;\;\;\; h^{-}_{ai}=\wh{\e}_{i}{}^{j}\wh{D}_{j}h_{a}, \;\;\;\; h^{-}_{ij}=0.
\end{equation}
The corresponding linearized Ricci tensor is\footnote{recall that our Ricci tensor has the opposite sign to the one of \cite{Sarbach}} 
\begin{eqnarray}
 \dot{R}^{-}_{ab}&=&0, \\
 \dot{R}^{-}_{ai}&=&+\tfrac{1}{2}\wh{\e}_{i}{}^{j}\wh{D}_{j}\left[\tfrac{1}{r^2}\wt{\e}_{a}{}^{c}\wt{D}_{c}(r^2\mathcal{F})
 +\tfrac{1}{r^2}[\wh{\Delta}+(\wt{\Delta}r^2)]h_a\right], \label{Rai-} \\
 \dot{R}^{-}_{ij}&=&-\wh{\e}_{(i}{}^{k}\wh{D}_{j)}\wh{D}_{k}\wt{D}^{a}h_{a},
\end{eqnarray}
where
\begin{equation}
 \mathcal{F}:=r^2\wt{\e}^{ab}\wt{D}_{a}(r^{-2}h_{b}).
\end{equation}
We will only need the $ai$ component of the Einstein tensor:
\begin{equation}
 \dot{G}^{-}_{ai}=\dot{R}^{-}_{ai}-2\l h_{ai}.
\end{equation}
Using the form (\ref{Sads}) of the metric, we have
\begin{equation}
 \wt{\Delta}r^2=-2+2\l r^2,
\end{equation}
whereby 
\begin{equation}\label{R-h}
 \dot{G}^{-}_{ai}+\l h^{-}_{ai}=+\tfrac{1}{2}\wh{\e}_{i}{}^{j}\wh{D}_{j}\left[\tfrac{1}{r^2}\wt{\e}_{a}{}^{c}\wt{D}_{c}(r^2\mathcal{F})
 +\tfrac{1}{r^2}(\wh{\Delta}-2)h_a\right].
\end{equation}
The Einstein equations $\dot{G}^{-}_{ai}+\l h^{-}_{ai}=0$ together with the fact that $(\wh{\Delta}-2)$ is invertible in the 
space of interest (that is, with $\ell\geq2$ in a decomposition into spherical harmonics, see \cite{Sarbach,Dotti2}), 
imply that the original metric perturbation can be recovered from $\mathcal{F}$:
\begin{equation}\label{hsarbach}
 h^{-}_{ai}=-\wh{\e}_{i}{}^{j}\wh{D}_{j}\wt{\e}_{a}{}^{b}\wt{D}_{b}[r^2(\wh{\Delta}-2)^{-1}\mathcal{F}].
\end{equation}
We will see now the relation of this formalism with our four-dimensional approach in this paper.

Using the explicit expression (\ref{R-h}), the general formula (\ref{sw0G}) can be checked directly in this case, 
as the following lemma shows:
\begin{lem}\label{lemma-schw}
 In the odd sector (\ref{odds}) of linearized gravity on the Schwarzschild-(A)dS spacetime, we have the equality
 \begin{equation}
 \Psi^{-1/3}_2{}^{*}C^{\a\g\b\d}\c_{\d}[\Psi^{-4/3}_2\c_{\g}(\dot{G}^{-}_{\a\b}[h]+\l h^{-}_{\a\b})]
 =\tfrac{3}{M^{2/3}}\left[\left(\Box-\tfrac{8M}{r^3}+\tfrac{2}{3}\l \right)\wh{\Delta}\mathcal{F}\right].
 \end{equation}
\end{lem}
\begin{proof}
Define $E_{\a\b}:=\dot{G}^{-}_{\a\b}[h]+\l h^{-}_{\a\b}$. We will first prove that 
\begin{equation}
  \Psi^{-1/3}_2{}^{*}C^{\a\g\b\d}\c_{\d}[\Psi^{-4/3}_2\c_{\g}E_{\a\b}]
  =-\tfrac{6}{M^{2/3}}\wh{\e}^{ij}\wh{D}_{j}\wt{\e}^{ab}\wt{D}_{b}E_{ai},
\end{equation}
this is actually true for any symmetric tensor $E_{\a\b}$.
The calculation is done by using the explicit expression (\ref{C*schw}) for the dual Weyl tensor. 
The term with the volume form vanishes in the contraction with a symmetric tensor. Using that $\c^{\d}{}^{*}C_{\a\b\g\d}=0$, 
we have
\begin{multline}
\Psi^{-1/3}_2{}^{*}C^{\a\g\b\d}\c_{\d}[\Psi^{-4/3}_2\c_{\g}E_{\a\b}]\\
 =-\tfrac{3}{k^2}r\c_{\d}[r^4\c_{\g}(\tfrac{1}{r^5}Y^{\a\g} {}^{*}Y^{\b\d}E_{\a\b})]
 -\tfrac{3}{k^2}r\c_{\d}[r^4\c_{\g}(\tfrac{1}{r^5} {}^{*}Y^{\a\g} Y^{\b\d}E_{\a\b})]. \label{SEschw}
\end{multline}
Let us focus on the first term of the last expression, the calculations for the second one are similar. We find
\begin{multline}
 \tfrac{3}{k^2}r\c_{\d}[r^4\c_{\g}(\tfrac{1}{r^5}Y^{\a\g} {}^{*}Y^{\b\d}E_{\a\b})]\\
 =\tfrac{3}{k^2}r\c_{\d}\left[-\tfrac{5}{r^2}r_{\g}Y^{\a\g} {}^{*}Y^{\b\d}E_{\a\b}
 +\tfrac{1}{r}\c_{\g}(Y^{\a\g} {}^{*}Y^{\b\d}E_{\a\b})\right],
\end{multline}
where $r_{\a}:=\c_{\a}r$. Note that $r_{\g}Y^{\a\g}=0$ because of the explicit form (\ref{kyschw}). 
Using (\ref{kyschw})-(\ref{kydschw}) and the relation between the covariant derivatives of the different spaces, we find
\begin{equation}
 \c_{\g}(Y^{\a\g} {}^{*}Y^{\b\d}E_{\a\b})=\wh{D}_i(Y^{\a i}{}^{*}Y^{\b\d}E_{\a\b})
 +\tfrac{1}{r}r_b Y^{\a\d} {}^{*}Y^{\b b}E_{\a\b},
\end{equation}
which implies 
\begin{multline}
 \tfrac{3}{k^2}r\c_{\d}[r^4\c_{\g}(\tfrac{1}{r^5}Y^{\a\g} {}^{*}Y^{\b\d}E_{\a\b})]\\
 =\tfrac{3}{k^2}r\c_{\d}[\tfrac{1}{r}\wh{D}_i(Y^{\a i}{}^{*}Y^{\b\d}E_{\a\b})
 +\tfrac{1}{r^2}r_b Y^{\a\d} {}^{*}Y^{\b b}E_{\a\b}].
\end{multline}
Once more, we use the relation between covariant derivatives, and find that
\begin{eqnarray*}
 \c_{\d}[\tfrac{1}{r}\wh{D}_i(Y^{\a i}{}^{*}Y^{\b\d}E_{\a\b})]
 &=&\tfrac{1}{r}\wt{D}_{b}\wh{D}_{i}(Y^{\a i} {}^{*}Y^{\b b}E_{\a\b})
 +\tfrac{1}{r^2}r_{b}\wh{D}_i(Y^{\a i}{}^{*}Y^{\b b}E_{\a\b})\\
 \c_{\d}(\tfrac{1}{r^2}r_b Y^{\a\d} {}^{*}Y^{\b b}E_{\a\b}) &=&
 \tfrac{1}{r^2}\wh{D}_i(r_{b} Y^{\a i}{}^{*}Y^{\b b}E_{\a\b});
\end{eqnarray*}
therefore
\begin{multline}
 \tfrac{3}{k^2}r\c_{\d}[r^4\c_{\g}(\tfrac{1}{r^5}Y^{\a\g} {}^{*}Y^{\b\d}E_{\a\b})]\\
 =\tfrac{3}{k^2}\left[\wt{D}_{b}\wh{D}_{i}(Y^{\a i} {}^{*}Y^{\b b}E_{\a\b})+
 \tfrac{2}{r}r_{b}\wh{D}_i( Y^{\a i}{}^{*}Y^{\b b}E_{\a\b})\right]. \label{SEschw1}
\end{multline}
The calculation for the second term in (\ref{SEschw}) is performed along the same lines, the result is
\begin{multline}
 \tfrac{3}{k^2}r\c_{\d}[r^4\c_{\g}(\tfrac{1}{r^5}{}^{*}Y^{\a\g} Y^{\b\d}E_{\a\b})]\\
 =\tfrac{3}{k^2}\left[\wh{D}_{i}\wt{D}_{b}({}^{*}Y^{\a b} Y^{\b i}E_{\a\b})-
 \tfrac{2}{r}r_{b}\wh{D}_i( {}^{*}Y^{\a b}Y^{\b i}E_{\a\b})\right]. \label{SEschw2}
\end{multline}
Putting together (\ref{SEschw1}) and (\ref{SEschw2}), using the explicit forms (\ref{kyschw})-(\ref{kydschw}) 
of the Killing-Yano tensors and $E_{\a\b}=\dot{G}^{-}_{\a\b}[h]+\l h^{-}_{\a\b}$, we obtain
\begin{equation}
 \Psi^{-1/3}_2{}^{*}C^{\a\g\b\d}\c_{\d}[\Psi^{-4/3}_2\c_{\g}(\dot{G}^{-}_{\a\b}[h]+\l h^{-}_{\a\b})]
  =-\tfrac{6}{M^{2/3}}\wh{\e}^{ij}\wh{D}_{j}\wt{\e}^{ab}\wt{D}_{b}(\dot{G}^{-}_{ai}[h]+\l h^{-}_{ai}).
\end{equation}

Next, we calculate $\wt{\e}^{ab}\wt{D}_{b}(\dot{G}^{-}_{ai}+\l h^{-}_{ai})$ using the explicit expression (\ref{R-h}), 
the background equations $r^{a}r_{a}=-f(r)$, $\wt{\Delta}r=-\p_{r}f(r)$, and the decomposition
of the wave operator
\begin{equation}\label{waveop}
 \Box=\wt{\Delta}+\tfrac{1}{r^2}\wh{\Delta}+\tfrac{2}{r}r^a\wt{D}_{a}.
\end{equation}
We find
\begin{equation}
 \wt{\e}^{ab}\wt{D}_{b}(\dot{G}^{-}_{ai}+\l h^{-}_{ai})=-\tfrac{1}{2}\wh{\e}_{i}{}^{j}\wh{D}_{j}
 \left[\left(\Box-\tfrac{8M}{r^3}+\tfrac{2}{3}\l \right)\mathcal{F}\right].
\end{equation}
Finally, using the fact that $[\wh{\Delta},\Box-\tfrac{8M}{r^3}+\tfrac{2}{3}\l]=0$ on scalar fields, we obtain
\begin{equation}
 \Psi^{-1/3}_2{}^{*}C^{\a\g\b\d}\c_{\d}[\Psi^{-4/3}_2\c_{\g}(\dot{G}^{-}_{\a\b}[h]+\l h^{-}_{\a\b})]
 =\tfrac{3}{M^{2/3}}\left[\left(\Box-\tfrac{8M}{r^3}+\tfrac{2}{3}\l \right)\wh{\Delta}\mathcal{F}\right].
\end{equation}
\end{proof}

\subsubsection{Metric reconstruction}\label{met.rec.schw}

We now explain how to recover the original metric perturbation (\ref{odds}) from a solution to the scalar equation. 
From section \ref{sw0G-sec} we know that in case $\Psi_2$ is real, if $\Phi$ is a solution to 
\begin{equation}
 (\Box+8\Psi_2+\tfrac{2}{3}\l)\Phi=0,  
\end{equation}
then the tensor field
\begin{equation}
 h_{\a\b}(\Phi)=-{}^{*}C_{(\a}{}^{\g\d}{}_{\b)}\c_{\g}[\Psi^{-4/3}_2\c_{\d}(\Psi^{-1/3}_2\Phi)]
\end{equation}
is a solution to the linearized Einstein equations.
First, let us show that in the Schwarzschild-(A)dS solution, this expression reproduces formula (25) in \cite{Dotti}.
Using $\Psi_2=-\frac{M}{r^3}$ and the explicit form of the dual Weyl tensor (\ref{C*schw}), we have:
\begin{eqnarray*}
 h_{\a\b}(\Phi)&=&\tfrac{1}{M^{5/3}}{}^{*}C_{(\a}{}^{\g\d}{}_{\b)}\c_{\g}[r^4\c_{\d}(r\Phi)]\\
 &=&\tfrac{3}{b^2M^{2/3}r^5}\left({}^{*}Y_{(\a}{}^{\g}Y^{\d}{}_{\b)}+Y_{(\a}{}^{\g} {}^{*}Y^{\d}{}_{\b)}\right)\c_{\g}[r^4\c_{\d}(r\Phi)]
\end{eqnarray*}
Now, using (\ref{kyschw}), (\ref{kydschw}) it is easy to see that
\begin{equation}
 {}^{*}Y_{(\a}{}^{\g}Y^{\d}{}_{\b)}\c_{\g}[r^4\c_{\d}(r\Phi)]={}^{*}Y_{(\a}{}^{\g}Y^{\d}{}_{\b)}\c_{\g}\c_{\d}(r^5\Phi),
\end{equation}
therefore:
\begin{eqnarray*}
 h_{\a\b}(\Phi)&=&\tfrac{3}{b^2M^{2/3}r^5}
 \left({}^{*}Y_{(\a}{}^{\g}Y^{\d}{}_{\b)}\c_{\g}\c_{\d}(r^5\Phi)+Y_{(\a}{}^{\g} {}^{*}Y^{\d}{}_{\b)}\c_{\g}[r^4\c_{\d}(r\Phi)]\right)\\
 &=&\tfrac{3}{b^2M^{2/3}r^5}Y_{(\a}{}^{\g} {}^{*}Y^{\d}{}_{\b)}\c_{\g}[\c_{\d}(r^5\Phi)+r^4\c_{\d}(r\Phi)]\\
 &=&\tfrac{3}{b^2M^{2/3}r^5}Y_{(\a}{}^{\g} {}^{*}Y^{\d}{}_{\b)}2r^2\c_{\g}\c_{\d}(r^3\Phi),
\end{eqnarray*}
where we have used the identity $\c_{\d}(r^5\Phi)+r^4\c_{\d}(r\Phi)=2r^2\c_{\d}(r^3\Phi)$ and the fact that $Y_{\a}{}^{\g}r_{\g}=0$.
Thus,
\begin{equation}\label{metric1}
 h_{\a\b}(\Phi)= \tfrac{1}{M^{5/3}} r^2 {}^{*}C_{\a}{}^{\g\d}{}_{\b}\c_{\g}\c_{\d}(r^3\Phi).
\end{equation}
Now, using that on scalar fields we have the commutator
\begin{equation}
 [(\wh\Delta-2)^{-1},\Box+8\Psi_2+\tfrac{2}{3}\l]=0,
\end{equation}
if $\Phi$ is a solution to $(\Box+8\Psi_2+\tfrac{2}{3}\l)\Phi=0$, so is the field
$\frac{M^{2/3}}{3}(\wh\Delta-2)^{-1}\Phi$. Defining $\Phi_o:=(\wh\Delta-2)^{-1}\Phi$, the metric perturbation 
(\ref{metric1}) constructed from $\frac{M^{2/3}}{3}(\wh\Delta-2)^{-1}\Phi$ is
\begin{equation}
 h_{\a\b}= \tfrac{r^2}{3M}{}^{*}C_{\a}{}^{\g\d}{}_{\b}\c_{\g}\c_{\d}(r^3\Phi_o),
\end{equation}
which coincides with \cite[Eq. (25)]{Dotti}, and is the original perturbation (\ref{hsarbach}) (with $\Phi\equiv\mathcal{F}$). 
Our general results (\ref{sw0G}) thus explain the mechanism behind (\ref{we})-(\ref{metricschw}) (and extend it to the cosmological 
setting).

\subsubsection{Killing-Yano tensors}\label{kyt.schw}

Finally, we want to derive the formula (\ref{Phi}) from our general formalism. For this we will use (\ref{KYCgeneric}) 
in a slighty different form. Using that $\wt{C}_{\a\b}{}^{\g\d}\wt{C}_{\g\d}{}^{\a\b}=\wt{C}_{\a\b}{}^{\g\d}C_{\g\d}{}^{\a\b}$,
an alternative expression to (\ref{KYCgeneric}) is
\begin{equation}
 12\Psi_2\dot{\Psi}_2=\tfrac{1}{4}(\dot{\wt{C}}_{\a\b}{}^{\g\d}C_{\g\d}{}^{\a\b}+\wt{C}_{\a\b}{}^{\g\d}\dot{C}_{\g\d}{}^{\a\b})
\end{equation}
The linearization of $\wt{C}_{\a\b}{}^{\g\d}$ is delicate because we have to take into account that when we perturb the dual 
${}^{*}C_{\a\b}{}^{\g\d}$ there are two terms: the perturbed volume form $\dot{\e}_{\a\b}{}^{\mu\nu}$ and the perturbed Weyl tensor
$\dot{C}_{\mu\nu}{}^{\g\d}$. A straightforward calculation shows that
\begin{equation}
 \tfrac{d}{d\e}|_{\e=0} ({}^{*}C_{\a\b}{}^{\g\d})=\tfrac{h}{2}{}^{*}C_{\a\b}{}^{\g\d}+\e_{\a\b\rho}{}^{\mu}h^{\nu\rho}C_{\mu\nu}{}^{\g\d}
 +\tfrac{1}{2}\e_{\a\b}{}^{\mu\nu}\dot{C}_{\mu\nu}{}^{\g\d},
\end{equation}
where we recall that $h=g^{\a\b}h_{\a\b}=-g_{\a\b}h^{\a\b}$; then
\begin{equation}
 12\Psi_2\dot{\Psi}_2=\tfrac{1}{2}\wt{C}_{\g\d}{}^{\a\b}\dot{C}_{\a\b}{}^{\g\d}
 +\tfrac{i}{8}(\tfrac{h}{2}{}^{*}C_{\a\b}{}^{\g\d}C_{\g\d}{}^{\a\b}-2h^{\mu\nu}{}^{*}C_{\mu\b\g\d}C_{\nu}{}^{\b\g\d}).
\end{equation}
Now, using the identities 
\begin{eqnarray}
{}^{*}C_{\a\b}{}^{\g\d}C_{\g\d}{}^{\a\b} &=& 48\im(\Psi^2_2), \\
{}^{*}C_{\mu\b\g\d}C_{\nu}{}^{\b\g\d} &=& 12g_{\mu\nu}\im(\Psi^2_2),
\end{eqnarray}
we get
\begin{equation}
 12\Psi_2\dot{\Psi}_2=\tfrac{1}{2}\wt{C}_{\g\d}{}^{\a\b}\dot{C}_{\a\b}{}^{\g\d}+6ih\im(\Psi^2_2)
\end{equation}
As we are interested in the case in which $\Psi_2$ is real, we take the imaginary part in the last equation and, 
using the explicit form (\ref{sdweylD}), we obtain
\begin{equation}
 \Psi^{-2/3}_2\im(\dot{\Psi}_2)=-\tfrac{1}{16k^2}({}^{*}Y_{\g\d}Y^{\a\b}+Y_{\g\d}{}^{*}Y^{\a\b})\dot{C}_{\a\b}{}^{\g\d}.
\end{equation}
The two terms on the RHS turn out to be equal, therefore:
\begin{equation}
 \Psi^{-2/3}_2\im(\dot{\Psi}_2)=-\tfrac{1}{8k^2} Y^{\a\b} {}^{*}Y_{\g\d}\dot{C}_{\a\b}{}^{\g\d},
\end{equation}
which demonstrates (\ref{Phi}).

\section{Conclusions}\label{conclusions}

Working in the class of vacuum Petrov type D spacetimes with cosmological constant, we have presented the general form of 
linear, four-dimensional differential operators mapping {\em off-shell} the equations for linear fields of spin $\s=\frac{1}{2}$, 
1 and 2 into a system of scalar  equations for spin weighted $s$ components of these linear fields that decouple on shell.
By using the Bianchi identities linearized around $\l$-vacuum solutions, we were able to relate off-shell the decoupled 
equations for Weyl scalars to the linearized Einstein equations.
Applying transposition of operators we obtained a way to reconstruct solutions of the original field equations from solutions 
of the decoupled equations. 
This mechanism works well for extreme spin weight $s=\pm\s$ in the Dirac, Maxwell and linearized gravity cases. 
For spin weight $s=0$, the reconstruction formula  works for Maxwell fields, but for gravitational perturbations 
the `inhomogeneous' term in the RHS of (\ref{sw0G}) (namely $(\dot{\Box}_h+\frac{\dot{R}_h}{6})\Psi^{1/3}_2$)
spoils the transposition of operators that would lead to a 
reconstruction formula. One can get rid of this term whenever $\Psi_2$ is a real field, the Schwarzschild-(A)dS solution 
being the most significant example in the present work. Applying our general results to this case, we explained the mechanisms 
behind the equations presented in \cite{Dotti,Dotti2} corresponding to the odd sector of linearized gravity around the 
Schwarzschild-(A)dS black hole. In particular, we corroborate our general formulae by translating the four-dimensional 
expressions of our formalism into the traditional $2+2$ decomposition of warped product spacetimes, setting in this 
way the connection between both approaches.

Our off-shell formulation is also useful for obtaining symmetry operators for the field equations, both for 
the higher spin (Dirac, Maxwell, linear gravity) field equations and for the scalar (Teukolsky, Fackerell-Ipser, etc.) equations. 
For further results about symmetry operators in the literature, 
we note that a comprehensive analysis of the second order symmetry operators for the field equations of
massless test fields of spin $0$, $\frac{1}{2}$ and $1$ is performed in \cite{Andersson2},
and that higher order symmetry operators for spin 1 and 2 and extreme spin weight are obtained in \cite{Aksteiner}.

We have also analyzed the role that Killing spinors (and its tensor analogues, Killing-Yano forms) have in the description of 
the spin weight zero scalar equations for linear fields. 
Killing spinors are certainly very used in the literature. They are the main object in Penrose's spin lowering process for massless 
fields in Minkowski spacetime. For Petrov type D spaces, the 2-index Killing spinor encodes all the information about the symmetries 
and hidden symmetries of the Kerr solution.
They are also central for the existence of symmetry operators for massless 
fields of spin $1/2$ and $1$ in curved spacetimes, as was proved in \cite{Andersson2}, see e.g. Theorems 4 and 6 there.
However, in this work we found that, although some proofs are somewhat simplified by the Killing spinor equation, 
and the general object (\ref{P}) used in the theorems turns out to be a Killing spinor for spins 1 and 2 and spin weight zero, 
the final results do not depend on this condition. Thus, regarding the Maxwell and linearized gravity systems considered in this work,
we may consider the appearance of these objects as merely `accidental', in the sense that the proof of the theorems 
can be done without use of the Killing spinor equation. (We mention that, although in the proof of the spin weight zero 
case of theorem \ref{thm-maxwell} for Maxwell we do use the Killing spinor equation to simplify the calculations, this proof can 
be performed without using this equation.)

There is a vast literature about the subject of symmetry operators and Debye potentials for higher spin fields. We particularly mention 
references \cite{Andersson1, Aksteiner, Aksteiner2, Andersson2, Chrzanowski, Dotti, Fackerell, Kegeles, Penrose4, Teukolsky, Wald}, 
whose connections with this work have been described throughout the text.
The results in this paper encompass a number of previously known results in the mentioned works
(and extend them to the cosmological setting), in particular:
\begin{itemize}
\item For extreme spin weight, the Teukolsky equations \cite{Teukolsky} are the on-shell version of the equations 
presented in this work:
\cite[Eqs.(B4) and (B5)]{Teukolsky} for the Dirac field are the on-shell case of (\ref{dirac1}) and (\ref{dirac2});
\cite[Eqs.(3.5) and (3.7)]{Teukolsky} for the Maxwell fields correspond to the on-shell case 
of equations (\ref{maxwell1}) and (\ref{maxwell3}); and 
\cite[Eqs.(2.12) and (2.14)]{Teukolsky} for linear gravity are the on-shell case of equations (\ref{grav1}) and (\ref{grav3}).
\item For spin weight zero, the on-shell case of (\ref{maxwell2}) for Maxwell fields is the Fackerell-Ipser equation
\cite[Eq.(20)]{Fackerell}, and the on-shell case of (\ref{grav2}) for linear gravity is the linearized equation 
\cite[Eq.(3.10)]{Andersson1} of Aksteiner $\&$ Andersson.
\item The reconstruction formula (\ref{m-rec-esw}) for spin weight $s=-2$ can be checked to agree with Kegeles $\&$ Cohen 
ansatz \cite[Eq.(5.4)]{Kegeles}.
\item For a Schwarzschild background, the on-shell case of (\ref{grav2}) is \cite[Eq.(24)]{Dotti}
(or \cite[Eq.(4DRWE)]{Dotti2} in the cosmological setting), and the reconstruction 
equation (\ref{m-rec-sw0}) is Dotti's formula \cite[Eq.(25)]{Dotti}.
\end{itemize}

\section*{Acknowledgments}

I would like to thank my advisor Gustavo Dotti for suggesting to me the problem and
for very helpful suggestions and conversations, Thomas B\"ackdahl for useful comments
and for pointing me towards some unnoticed references and relations with his work, and two anonymous 
referees for a number of comments that led to an improved version of this manuscript.
I would also like to thank Sergio Dain for discussions, and dedicate this work to his memory.
This work is supported by a doctoral fellowship from CONICET (Argentina).

\appendix

\section{Useful formulae}

In this appendix we collect some useful formulae we have made use of in the proofs of the results in the main text.

\subsection{Curvature spinors}\label{appA}

For completeness we recall the definition of spinor curvature operations used in this paper 
(we simply repeat the formulae of \cite[section 4.9]{Penrose1} relevant for this work).
Our convention for the definition of the Riemann curvature tensor (in the absence of torsion) is (see \cite[Eq.(4.2.31)]{Penrose1})
\begin{equation}
 (\c_{\a}\c_{\b}-\c_{\b}\c_{\a})V^{\g}=+R_{\a\b\d}{}^{\g}V^{\d}.
\end{equation}
Note that the RHS of this equation has the opposite sign to the more commonly used definition,
compare e.g. with \cite[Eq.(3.2.11)]{Wald2}. This implies that our Riemann and Ricci tensors have the opposite signs 
to those of this reference (note however that the curvature scalar and cosmological constant are the same).

The commutator of two covariant derivatives gives the curvature spinor operators $\Box_{AB}$ and $\Box_{A'B'}$ in the form
\begin{equation}
 \c_{\a}\c_{\b}-\c_{\b}\c_{\a}=\e_{A'B'}\Box_{AB}+\e_{AB}\Box_{A'B'},
\end{equation}
where $\Box_{AB}=\c_{A'(A}\c^{A'}_{B)}$, and its action on, for example, a spinor $\theta^{C}{}_{D}{}^{E'}{}_{F'}$, is
\begin{eqnarray}
\nonumber \Box_{AB}\theta^{C}{}_{D}{}^{E'}{}_{F'}&=&X_{ABQ}{}^{C}\theta^{Q}{}_{D}{}^{E'}{}_{F'}-X_{ABD}{}^{Q}\theta^{C}{}_{Q}{}^{E'}{}_{F'}\\
 & & +\Phi_{ABQ'}{}^{E'}\theta^{C}{}_{D}{}^{Q'}{}_{D'}-\Phi_{ABF'}{}^{Q'}\theta^{C}{}_{D}{}^{E'}{}_{Q'}; \label{boxAB}
\end{eqnarray}
a similar formula holds for $\Box_{A'B'}$ (see \cite[Eq.(4.9.14)]{Penrose1}). The curvature spinor $X_{ABCD}$ is decomposed as
\begin{equation}\label{X1}
 X_{ABCD}=\psi_{ABCD}+\tfrac{R}{24}(\e_{AC}\e_{BD}+\e_{AD}\e_{BC}),
\end{equation}
where $\psi_{ABCD}$ is the Weyl conformal curvature spinor. This implies in particular that
\begin{equation}\label{X2}
 X_{ABC}{}^{B}=\tfrac{R}{8}\e_{AC}.
\end{equation}

\subsection{Derivatives of the dyad spinors}\label{appB}

Using expressions for the directional derivatives of the dyad $\{o^A,\iota^A\}$ along the tetrad vectors 
(equations (4.5.26) in \cite{Penrose1}), it is easy to see that
\begin{eqnarray}
\nonumber \c^{M'}_{M}o^{A}&=&(\e o^A-\kappa\iota^A)\iota_{M}\bar{\iota}^{M'}-(\e' o^A+\tau\iota^A)o_M\bar{o}^{M'}\\
 & & +(\b' o^A+\rho\iota^A)o_{M}\bar{\iota}^{M'}-(\b o^A-\sigma\iota^A)\iota_{M}\bar{o}^{M'},\\
\nonumber \c^{M'}_{M}\iota^{A}&=&-(\e \iota^A+\tau' o^A)\iota_{M}\bar{\iota}^{M'}+(\e' \iota^A-\kappa'o^A)o_M\bar{o}^{M'}\\
 & & -(\b'\iota^A-\sigma' o^A)o_{M}\bar{\iota}^{M'}+(\b \iota^A+\rho' o^A)\iota_{M}\bar{o}^{M'}.
\end{eqnarray}
Contracting with $o^M$, $\iota^M$ and $\e_{A}{}^{M}$, we obtain the following useful formulae:
\begin{eqnarray}
 o^{M}\c^{M'}_{M}o^{A}&=&(-\e o^A+\kappa\iota^A)\bar{\iota}^{M'}+(\b o^A-\sigma\iota^A)\bar{o}^{M'}, \label{odo} \\
 \iota^{M}\c^{M'}_{M}o^{A}&=&-(\e' o^A+\tau\iota^A)\bar{o}^{M'}+(\b' o^A+\rho\iota^A)\bar{\iota}^{M'}, \label{ido} \\
 \c^{M'}_{A}o^A&=&(\rho-\e)\bar{\iota}^{M'}+(\b-\tau)\bar{o}^{M'}, \label{edo} \\
 o^{M}\c^{M'}_{M}\iota^{A}&=&(\e \iota^A+\tau' o^A)\bar{\iota}^{M'}-(\b\iota^A+\rho' o^A)\bar{o}^{M'}, \label{odi} \\
 \iota^{M}\c^{M'}_{M}\iota^{A}&=&-(\e'\iota^A-\kappa' o^A)\bar{o}^{M'}-(\b'\iota^A-\sigma o^A)\bar{\iota}^{M'}, \label{idi} \\
 \c^{M'}_{A}\iota^{A}&=&(\tau'-\b')\bar{\iota}^{M'}+(\e'-\rho')\bar{o}^{M'}. \label{edi}
\end{eqnarray}

For the proofs of the theorems in the text we also need expressions for the divergence of the tetrad vectors:
\begin{eqnarray}
 \c_{\a}l^{\a}&=&\e+\bar\e-\rho-\bar\rho, \label{dl} \\
 \c_{\a}n^{\a}&=&\e'+\bar\e'-\rho'-\bar\rho', \label{dn} \\
 \c_{\a}m^{\a}&=&\b+\bar\b'-\tau-\bar\tau', \label{dm} \\
 \c_{\a}\bar{m}^{\a}&=&\b'+\bar\b-\tau'-\bar\tau. \label{dmb}
\end{eqnarray}

\subsection{Killing spinors}

In the following proposition we gather useful identities involving the Killing spinor of type D solutions:
\begin{prop}
 Consider the Killing spinor $K_{AB}$ of a $\l$-vacuum type D spacetime, and let $\xi^{AA'}=\c^{A'B}K_{B}{}^{A}$ be the 
 associated Killing vector. We have:
 \begin{eqnarray}
  \c_{C'C}K_{AB}&=&\tfrac{2}{3}\c^{D}_{C'}K_{D(A}\e_{B)C}, \label{k1} \\
  \psi_{ABCD}K^{CD}&=&-2\Psi_2 K_{AB}, \label{k2} \\
  \Box K_{AB}&=&(2\Psi_2+\tfrac{2}{3}\l)K_{AB}, \label{k3} \\
  \c^{A'A}(\Psi_2 K_{AB})&=&0, \label{k4} \\
  \c_{B'A}\xi^{B'}_{B}&=& (3\Psi_2+\l)K_{AB}. \label{k5}
 \end{eqnarray}
\end{prop}

\begin{proof}
(\ref{k1}) follows immediately after using the Killing spinor condition $\c_{C'(C}K_{AB)}=0$ and 
\cite[Eq.(3.3.55)]{Penrose1}, 
\begin{equation*}
 \c_{C'C}K_{AB}=-\tfrac{1}{3}\e_{CA}\c^{D}_{C'}K_{DB}-\tfrac{1}{3}\c^{D}_{C'}K_{DA}\e_{CB}.
\end{equation*}
For (\ref{k2}) we just have to use the expressions (\ref{WSD}) for the Weyl spinor and (\ref{KSD}) for $K_{AB}$, together 
with the identity
\begin{equation}
 K_{AB}K^{AB}=-\tfrac{k^2}{2}\Psi^{-2/3}_2.
\end{equation}
For (\ref{k3}), we take an additional derivative in $0=\c^{C'}_{(C}K_{AB)}$ and use the decomposition (\ref{X1}) of $X_{ABCD}$:
\begin{eqnarray*}
 0&=&\c_{C'D}\c^{C'}_{(C}K_{AB)}\\
  &=& \tfrac{1}{2}\e_{D(C}\Box K_{AB)}-2\psi_{D(CA}{}^{E}K_{B)E}-\tfrac{R}{12}\e_{D(A}K_{BC)}.
\end{eqnarray*}
Expanding in $CAB$ and contracting with $\e^{CD}$, we get
\begin{equation*}
 0=\Box K_{AB}+\psi_{ABCD}K^{CD}-\tfrac{R}{6}K_{AB}
\end{equation*}
which, after using (\ref{k2}) and replacing $R=4\l$, reduces to (\ref{k3}).\\
Formula (\ref{k4}) follows after applying a derivative $\c^{A'A}$ to both sides of (\ref{k2}) and using the Bianchi 
identities $\c^{A'A}\psi_{ABCD}=0$ and the Killing spinor condition $\c^{A'(A}K^{CD)}=0$ (together with the fact that 
$\psi_{ABCD}$ is totally symmetric).\\
Finally, for (\ref{k5}) we use the definition $\xi^{B'}_{B}=\c^{B'C}K_{CB}$:
\begin{equation*}
 \c_{B'A}\xi^{B'}_{B}=\c_{B'A}\c^{B'C}K_{CB}=\tfrac{1}{2}\Box K_{AB}-\psi_{ABCD}K^{CD}+\tfrac{R}{6}K_{AB}.
\end{equation*}
Then, using (\ref{k3}), (\ref{k2}) and $R=4\l$ we easily obtain (\ref{k5}).
\end{proof}


\begin{thebibliography}{99}


\bibitem{Andersson1} 
  S.~Aksteiner and L.~Andersson,
  {\it Linearized gravity and gauge conditions,}
  Class.\ Quant.\ Grav.\  {\bf 28}, 065001 (2011)
  [arXiv:1009.5647 [gr-qc]].


\bibitem{Aksteiner} 
  S.~Aksteiner and T.~B\"ackdahl,
  {\it Symmetries of linearized gravity from adjoint operators},
  arXiv:1609.04584 [gr-qc].


\bibitem{Aksteiner2} 
  S.~Aksteiner, L.~Andersson and T.~B\"ackdahl,
  {\it On the structure of linearized gravity on vacuum spacetimes of Petrov type D},
  arXiv:1601.06084 [gr-qc].


\bibitem{Andersson2} 
  L.~Andersson, T.~B\"ackdahl and P.~Blue,
  {\it Second order symmetry operators},
  Class.\ Quant.\ Grav.\  {\bf 31}, 135015 (2014)
  doi:10.1088/0264-9381/31/13/135015
  [arXiv:1402.6252 [gr-qc]].


\bibitem{Araneda} 
  B.~Araneda and G.~Dotti,
  {\it Petrov type of linearly perturbed type D spacetimes},
  Class.\ Quant.\ Grav.\  {\bf 32}, no. 19, 195013 (2015)
  [arXiv:1502.07153 [gr-qc]].

  
\bibitem{AranedaDotti}
 B.~Araneda and G.~Dotti, 
 {\it Instability of asymptotically anti de Sitter black holes under Robin conditions at the timelike boundary}, 
 arXiv:1611.03534 [hep-th].
 
 
\bibitem{Backdahl} 
  T.~B\"ackdahl and J.~A.~V.~Kroon,
  {\it A formalism for the calculus of variations with spinors},
  J.\ Math.\ Phys.\  {\bf 57}, no. 2, 022502 (2016)
  doi:10.1063/1.4939562
  [arXiv:1505.03770 [gr-qc]].


\bibitem{Szekeres} 
  P.~Bell and P.~Szekeres,
  {\it Some properties of higher spin rest-mass zero fields in general relativity,}
  Int.\ J.\ Theor.\ Phys.\  {\bf 6}, 111 (1972). 

\bibitem{Bini} 
  D.~Bini, C.~Cherubini, R.~T.~Jantzen and R.~J.~Ruffini,
  {\it Teukolsky master equation: De Rham wave equation for the gravitational and electromagnetic fields in vacuum},
  Prog.\ Theor.\ Phys.\  {\bf 107}, 967 (2002)
  doi:10.1143/PTP.107.967
  [gr-qc/0203069].

\bibitem{Blue}
 P.~Blue,
 {\it Decay of the Maxwell field on the Schwarzschild manifold}, 
 Journal of Hyperbolic Differential Equations, 5(04), 807-856  (2008)  

\bibitem{Carter} 
  B.~Carter,
  {\it Killing Tensor Quantum Numbers and Conserved Currents in Curved Space,}
  Phys.\ Rev.\ D {\bf 16}, 3395 (1977).  


\bibitem{curtis} 
 G. E. Curtis, 
 {\it Twistors and linearized Einstein theory on plane-fronted impulsive wave backgrounds,} 
 Gen.\ Rel.\ Grav. \ {\bf 9},  (19878) 987  

\bibitem{Sarbach} 
  E.~Chaverra, N.~Ortiz and O.~Sarbach,
  {\it Linear perturbations of self-gravitating spherically symmetric configurations,}
  Phys.\ Rev.\ D {\bf 87}, no. 4, 044015 (2013)
  doi:10.1103/PhysRevD.87.044015
  [arXiv:1209.3731 [gr-qc]].


\bibitem{Chrzanowski} 
  P.~L.~Chrzanowski,
  {\it Vector Potential and Metric Perturbations of a Rotating Black Hole},
  Phys.\ Rev.\ D {\bf 11}, 2042 (1975).
  doi:10.1103/PhysRevD.11.2042


\bibitem{Dotti}
  G.~Dotti,
  {\it Nonmodal linear stability of the Schwarzschild black hole,}
  Phys.\ Rev.\ Lett.\  {\bf 112} (2014) 191101
  [arXiv:1307.3340 [gr-qc]].


\bibitem{Dotti2} 
  G.~Dotti,
  {\it Black hole nonmodal linear stability: the Schwarzschild (A)dS cases},
  Class.\ Quant.\ Grav.\  {\bf 33}, no. 20, 205005 (2016)
  doi:10.1088/0264-9381/33/20/205005
  [arXiv:1603.03749 [gr-qc]].


\bibitem{ehlers}
 J.~Ehlers,  
 {\it The geometry of the (modified) GHP-formalism},
 Communications in Mathematical Physics, 37(4), 327-329.  (1974).

\bibitem{Fackerell} 
  E.~D.~Fackerell and J.~R.~Ipser,
  {\it Weak electromagnetic fields around a rotating black hole,}
  Phys.\ Rev.\ D {\bf 5}, 2455 (1972).

\bibitem{Geroch:1973am} 
  R.~P.~Geroch, A.~Held and R.~Penrose,
  {\it A space-time calculus based on pairs of null directions},
  J.\ Math.\ Phys.\  {\bf 14}, 874 (1973).
  doi:10.1063/1.1666410
 

\bibitem{Geroch} 
  R.~P.~Geroch,
  {\it Spinor structure of space-times in general relativity. I,}
  J.\ Math.\ Phys.\  {\bf 9}, 1739 (1968).

\bibitem{Geroch2} 
  R.~P.~Geroch,
  {\it Spinor structure of space-times in general relativity. II},
  J.\ Math.\ Phys.\  {\bf 11}, 343 (1970).

\bibitem{Gibbons} 
  G.~W.~Gibbons, R.~H.~Rietdijk and J.~W.~van Holten,
 {\it SUSY in the sky,}
  Nucl.\ Phys.\ B {\bf 404}, 42 (1993)
  [hep-th/9303112]. 


\bibitem{Jezierski:2005cg} 
  J.~Jezierski and M.~Lukasik,
  {\it Conformal Yano-Killing tensor for the Kerr metric and conserved quantities,}
  Class.\ Quant.\ Grav.\  {\bf 23}, 2895 (2006)
  doi:10.1088/0264-9381/23/9/008
  [gr-qc/0510058].

\bibitem{Kegeles} 
  L.~S.~Kegeles and J.~M.~Cohen,
  {\it Constructive Procedure For Perturbations Of Space-times},
  Phys.\ Rev.\ D {\bf 19}, 1641 (1979).
  doi:10.1103/PhysRevD.19.1641
  

\bibitem{Nakahara} 
  M.~Nakahara,
  ``Geometry, topology and physics,''
  Boca Raton, USA: Taylor $\&$ Francis (2003) 573 p

\bibitem{Penrose1}
 R.~Penrose and W.~Rindler,
  ``Spinors And Space-time. 1. Two Spinor Calculus And Relativistic Fields,''
  Cambridge, Uk: Univ. Pr. ( 1984) 458 P. ( Cambridge Monographs On Mathematical Physics)

\bibitem{Penrose2}
  R.~Penrose and W.~Rindler,
  ``Spinors And Space-time. Vol. 2: Spinor And Twistor Methods In Space-time Geometry,''
  Cambridge, Uk: Univ. Pr. ( 1986) 501p

\bibitem{Penrose4}
 R.~Penrose,
  {\it Zero rest mass fields including gravitation: Asymptotic behavior,}
  Proc.\ Roy.\ Soc.\ Lond.\ A {\bf 284} (1965) 159.

\bibitem{Price} 
  R.~H.~Price,
  {\it Nonspherical Perturbations of Relativistic Gravitational Collapse. II. Integer-Spin, Zero-Rest-Mass Fields},
  Phys.\ Rev.\ D {\bf 5}, 2439 (1972).
  doi:10.1103/PhysRevD.5.2439

\bibitem{Santillan} 
  O.~P.~Santillan,
  {\it Hidden symmetries and supergravity solutions,}
  J.\ Math.\ Phys.\  {\bf 53}, 043509 (2012)
  [arXiv:1108.0149 [hep-th]].

\bibitem{Teukolsky} 
  S.~A.~Teukolsky,
  {\it Rotating black holes - separable wave equations for gravitational and electromagnetic perturbations},
  Phys.\ Rev.\ Lett.\  {\bf 29}, 1114 (1972).
  doi:10.1103/PhysRevLett.29.1114

\bibitem{vanNieu} 
  P.~van Nieuwenhuizen and N.~P.~Warner,
  {\it Integrability conditions for Killing spinors,}
  Commun.\ Math.\ Phys.\  {\bf 93}, 277 (1984).

\bibitem{Penrose3} 
  M.~Walker and R.~Penrose,
  {\it On quadratic first integrals of the geodesic equations for type [22] spacetimes,}
  Commun.\ Math.\ Phys.\  {\bf 18}, 265 (1970).  


\bibitem{Wald}
  R.~M.~Wald,
  \textit{Construction of Solutions of Gravitational, Electromagnetic, Or Other Perturbation Equations from Solutions of Decoupled Equations,}
  Phys.\ Rev.\ Lett.\  {\bf 41}, 203 (1978).
   
\bibitem{Wald2} 
  R.~M.~Wald,
  ``General Relativity,''
  Chicago, Usa: Univ. Pr. ( 1984) 491p
  doi:10.7208/chicago/9780226870373.001.0001



\end{thebibliography}
\end{document}